\documentclass[a4paper,12pt]{article}

\usepackage{color}
\usepackage{graphicx}
\usepackage{amsmath}
\usepackage{amsfonts}
\usepackage{amsthm}
\usepackage{amscd}
\usepackage{epsfig}
\usepackage{amssymb}
\usepackage{slashed}
\usepackage{tabularx}
\usepackage{calligra}
\usepackage{enumerate}
\usepackage{hyperref}
\usepackage{float}
\usepackage{stackrel}
\usepackage[T1]{fontenc}
\usepackage{longtable}
\usepackage{amsthm}
\usepackage{mathrsfs}
\usepackage{verbatim} % for comments
\usepackage{fancyhdr}
\usepackage{tikz}
\usepackage{tikz-cd}
\usepackage{slashed}
\usetikzlibrary{matrix}
\usetikzlibrary{positioning}
\usepackage[all]{xy}

\numberwithin{equation}{section}

\newtheorem{thm}{Theorem}[section]

\newtheorem{lem}{Lemma}[section]%[thm]{Lemma}
\newtheorem{cor}[thm]{Corollary}

\renewcommand{\a}{\alpha}
\renewcommand{\b}{\beta}
\newcommand{\g}{\gamma}
\newcommand{\G}{\Gamma}
\newcommand{\e}{\epsilon}
\newcommand{\ve}{\varepsilon}
\newcommand{\vp}{\varphi}
\renewcommand{\l}{\lambda}
\renewcommand{\L}{\Lambda}
\renewcommand{\k}{\kappa}

\renewcommand{\d}{\delta}

\newcommand{\Ds}{\mathscr{D}}

\newcommand{\w}{\omega}
\newcommand{\Om}{\Omega}
\renewcommand{\t}{\tau}
\newcommand{\T}{\Theta}
\newcommand{\U}{\Upsilon}
\newcommand{\p}{\partial}

\newcommand{\M}{\mathscr{M}}

\newcommand{\mc}{\mathcal}
\newcommand{\mf}{\mathfrak}
\newcommand{\mbb}{\mathbb}
\newcommand{\R}{\mathbb{R}}
\newcommand{\s}{\mathfrak{s}}
\newcommand{\mr}{\mathring}

\newcommand{\im}{\mathfrak{Im}}

\newcommand{\teuk}{\ensuremath{\hspace{.05cm}\mathaccent\Box{\text{\tiny \textsc{T}}}\hspace{.05cm}}}     %% modified wave operator
\newcommand{\wt}{\widetilde}
\newcommand{\wh}{\widehat}
\def\c{\nabla}
\def\wt{\widetilde}

\def\wh{\widehat}

\begin{document}

\title{Generalized wave operators, weighted Killing fields, and perturbations of higher dimensional spacetimes}

\author{Bernardo Araneda\footnote{E-mail: \texttt{baraneda@famaf.unc.edu.ar}} \\ 
 \\
Facultad de Matem\'atica, Astronom\'ia y F\'isica\\
Universidad Nacional de C\'ordoba\\ 
Instituto de F\'isica Enrique Gaviola, CONICET\\
Ciudad Universitaria, (5000) C\'ordoba, Argentina 
}
\date{November 27, 2017}

\maketitle

\begin{abstract}

We present weighted covariant derivatives and wave operators for perturbations of certain algebraically special
Einstein spacetimes in arbitrary dimensions, 
under which the Teukolsky and related equations become weighted wave equations.
We show that the higher dimensional generalization of the principal null directions
are weighted conformal Killing vectors with respect to the modified covariant derivative.
We also introduce a modified Laplace-de Rham-like operator acting on tensor-valued differential forms, and show that the 
wave-like equations are, at the linear level, appropriate projections {\em off shell} of this operator 
acting on the curvature tensor;
the projection tensors being made out of weighted conformal Killing-Yano tensors.
We give off shell operator identities that map the Einstein and Maxwell equations into weighted scalar equations, 
and using adjoint operators we construct solutions of the original field equations in a compact form 
from solutions of the wave-like equations. 
We study the extreme and zero boost weight cases; extreme boost corresponding to perturbations 
of Kundt spacetimes (which includes Near Horizon Geometries of extreme black holes), 
and zero boost to static black holes in arbitrary dimensions. 
In 4 dimensions our results apply to Einstein spacetimes of Petrov type D and make use of weighted Killing spinors.

\end{abstract}

\tableofcontents

\section{Introduction}

The study of higher dimensional spacetimes is a very active research area that is well motivated from 
modern developments in high 
energy physics, mainly from superstring and supergravity theories where spacetime is 10- or 11-dimensional.
As is the case in 4-dimensional theories of gravity, the stability properties against 
perturbations of solutions of the field equations constitute a very important aspect to study when determining their 
possible physical relevance.

The ultimate question in the perturbation study of a spacetime is its non-linear stability. 
Due to the huge complexity of this problem, it is useful to first focus on {\em linear} stability.
However, even at the linear level the field equations have a complicated tensorial structure 
that makes them very difficult to analyze; therefore one is often interested in looking for alternative equations 
that describe the perturbations (or at least some sector of them) and are easier to deal with. 
The variables usually considered are some appropriate components of the 
perturbed curvature, which, as a consequence of the field equations, in some important cases happen to satisfy decoupled 
--spacetime scalar-- equations. 

One way to get these scalar equations is by projecting the original field equations on a frame at each 
tangent space of the manifold. The frame can be conveniently chosen according to the possible algebraic symmetries 
of the spacetime: several interesting solutions in dimension $d\geq4$, such as a large class of black holes, 
turn out to be {\em algebraically special}, i.e., there exist preferred null directions on the spacetime. 
A framework especially adapted to this situation is the Geroch-Held-Penrose (GHP) formalism, 
in which one ``breaks'' the $SO(d-1,1)$ local invariance on a frame by restricting to frames aligned with the 
preferred null directions, thus the gauge group gets reduced to a ``little group'' which can be shown to be 
$G_o=\R^{\times}\times SO(d-2)$.
All tensor fields are then classified according to two real parameters, the {\em boost weight $b$} and the 
{\em spin weight $s$}, which simply label the representations of $G_o$. \\

In $d=4$, one of the most interesting cases of study corresponds to perturbations of the class of Einstein spacetimes
of Petrov type D, since it includes as a particular case the Kerr-(A)dS solution describing a rotating black hole 
in vacuum with cosmological constant.
The analysis is restricted to the cases of {\em extreme} and {\em zero} boost weight\footnote{in $d=4$ it 
is perhaps more traditional to refer to the spin --instead of boost-- weight of the fields, since both weights coincide 
in all relevant components; however, this is not the case in $d>4$ and we find it more convenient to refer to the boost weight.},
since these components have nice properties such as coordinate and/or tetrad gauge invariance.
If the linearized field equations are imposed, extreme boost weight fields always satisfy the well-known Teukolsky equations. 
Zero boost fields, on the other hand, decouple only in certain special cases, 
for example for spin $\s=1$, or for spin $\s=2$ when the background $\Psi_2$ is real, which is the case for static black holes.
It was recently shown in \cite{Araneda:2016iwr} that all of these equations can be obtained, {\em off shell} (i.e. without 
assuming that the corresponding field equations are satisfied), as projections of a 
certain spinorial operator applied to a spinor field $\phi_{A_1...A_{2\s}}$ 
of spin $\s=\frac{1}{2},1,2$. 
The projection is made in terms of a GHP weighted spinor ${}^{b}P^{A_1...A_{2\s}}$, $b=0,\pm\s$,
constructed out of the type D principal null frame; and as a result one obtains wave-like equations in terms of weighted wave operators.
More precisely, consider a monoparametric family of 4-dimensional spacetimes $(\M,g(\ve))$, analytic around $\ve=0$, 
such that $(\M,g(0))$ is an Einstein spacetime of Petrov type D, and denote the first order perturbation of a quantity $T$ as 
$\dot{T}:=\frac{d}{d\ve}|_{\ve=0}T(\ve)$.
One can introduce a 1-form $\G_{\a}$ taking values in the Lie algebra 
$\mf{g}_o=\text{Lie}(G_o)\cong\mbb{C}$, and 
together with it a covariant derivative $D_{\a}:=\c_{\a}+p\G_{\a}$ and the corresponding modified wave operator 
$\teuk_{p}:=g^{\a\b}D_{\a}D_{\b}$, such that the off shell operator equalities have the general structure
\begin{equation}\label{teuk4d}
 {}^{b}P^{A_1...A_{2\s}}(\c^{A'}_{A_1}+2\s\g^{A'}_{A_1})\c^{B}_{A'}\phi_{A_2...A_{2\s}B} \doteq 
 (\teuk_{2b}+V_{\s,b})[{}^{b}P^{A_1...A_{2\s}}\phi_{A_1...A_{2\s}}]
\end{equation}
for some (generally complex) potential $V_{\s,b}$, where the 1-form $\g_{AA'}$ on the left hand side is given by 
$\g_{\a}=\Psi^{1/3}_2\c_{\a}\Psi^{-1/3}_2$.
The notation ``$\doteq$'' means that this is an equation valid to linear order (alternatively, take the linearization operator 
$\tfrac{d}{d\ve}|_{\ve=0}$ on both sides of (\ref{teuk4d}) --note that this is only needed for the case of spin $\s=2$).
Note that the usual D'Alembertian is included in $\teuk_{2b}$, since $\teuk_{0}=\Box=g^{\a\b}\c_{\a}\c_{\b}$. 
The operator $\teuk_{2b}$ for the 4-dimensional Teukolsky equations on the Kerr spacetime
was found in \cite{Bini}; see also \cite{Andersson1} for the treatment of all vacuum type D solutions.
Introducing certain ``potentials''\footnote{for spin $\s=\frac{1}{2}$, the Dirac field is its own potential; for spin $\s=1$ one 
introduces a vector potential; and for spin $\s=2$ the ``potential'' is a metric perturbation.}
for the spin-$\s$ fields, one can formulate the left hand side of (\ref{teuk4d}) 
in terms of self-adjoint operators, and then following Wald's adjoint operator technique (introduced in \cite{Wald}),
it is possible to take the adjoint identity
and reconstruct solutions of the field equations from solutions of the scalar equations; this way, stability 
questions can be approached by the study of a scalar, wave-like equation instead of the more complex tensor field equations.
In view of the utility of equation (\ref{teuk4d}) to 
study perturbations and symmetry structures of field equations on black hole spacetimes
(see the review section \ref{sec-review4d}),
in this work we study the possible generalization of this identity
to spacetimes of arbitrary dimensions. 
We describe our main results in section \ref{sec-mainr}. \\

A number of difficulties arise when performing this generalization. 
First, eq. (\ref{teuk4d}) was demonstrated using 4-dimensional spinor methods. Since the structure of the spin groups 
depends on the number of spacetime dimensions, 
and the power of spinor methods in $d=4$ is associated to the particular isomorphism $\mathrm{Spin}(3,1)\cong SL(2,\mathbb{C})$,
we will not be able to apply the same techniques as in $d=4$. 
Also, decoupling in higher dimensions is much more restrictive than the $d=4$ case. 
Let us focus first on extreme boost weight.
For $d>4$, the validity of (the analog of) the Teukolsky equations was studied in \cite{Durkee:2010qu},
where it was found that extreme boost components of the perturbed curvature decouple {\em if and only if} 
the background spacetime is Kundt, i.e., it admits a null geodesic congruence which is shear-free, twist-free, 
and expansion-free. 
Therefore, we are constrained to the consideration of this class in the study of extreme boost weight, 
which is more restrictive than the type D class.
(On the other hand, although a general $d$-dimensional black hole solution is not Kundt,
a very important example of this kind is the Near Horizon Geometry of extreme black holes, 
which has attracted a lot of attention in recent years.)
A third difficulty has to do with the explicit expression of the higher dimensional Teukolsky equations. 
As was originally the case for the equations in $d=4$, they have a very complicated form: 
in GHP notation, for spin $\s=1$ the equations presented in \cite{Durkee:2010qu} are\footnote{our conventions for 
the cosmological constant are fixed by requiring that the Einstein equations are $G_{\a\b}+\L g_{\a\b}=0$, 
which implies that $R_{\a\b}=\frac{2\L}{(d-2)}g_{\a\b}$; note in particular that this is different from the conventions of 
\cite{Durkee:2010qu}.}
\begin{multline}\label{teukhd1}
 \left[2\text{\th}'\text{\th}+\text{\dh}^k\text{\dh}_k+\rho'\text{\th}
 -4\t^k\text{\dh}_k+\Phi-\tfrac{2(2d-3)\L}{(d-1)(d-2)} \right]\vp_i\\
 +[-2\t_i\text{\dh}^k+2\t^k\text{\dh}_i+2\Phi^S_{i}{}^{k}+4\Phi^A_{i}{}^{k}]\vp_k=0,
\end{multline} 
where $\vp_{i}$ is the $b=1$ component of the Maxwell field; and for spin $\s=2$
\begin{multline}\label{teukhd2}
 [2\text{\th}'\text{\th}+\text{\dh}^k\text{\dh}_k+\rho'\text{\th}-6\tau^k\text{\dh}_k+4\Phi
 -\tfrac{4d}{(d-1)(d-2)}\L]\dot\Om_{ij}\\
 +4[\tau^k\text{\dh}_{(i}-\tau_{(i}\text{\dh}^k+\Phi^{S}_{(i}{}^{k}+4\Phi^{A}_{(i}{}^{k}]\dot\Om_{j)k}
 +2\Phi_{i}{}^{k}{}_{j}{}^{l}\dot\Om_{kl}=0
\end{multline}
with $\dot\Om_{ij}$ the perturbed $b=2$ component of the Weyl tensor.
This explicit form certainly makes it difficult to analyze the possible existence of a general pattern 
of symmetries in terms of a family of weighted wave operators as 
in the 4-dimensional case (see the right hand side of (\ref{teuk4d})).
In section \ref{sec-covder} we will show that the higher dimensional 
equations (\ref{teukhd1}) and (\ref{teukhd2}) can be put in a wave-like form (as was originally done 
for the $d=4$ case in \cite{Bini}), in terms of a modified wave operator $D^{\a}D_{\a}$ that depends on the GHP type $\{b,s\}$.
This is achieved through the introduction of a particular $\mf{g}_o$-valued 1-form (where 
$\mf{g}_o=\R\oplus\mf{so}(d-2)$ is the Lie algebra of the $d$-dimensional GHP group), that defines a modified covariant
derivative $D_{\a}$. This modified derivative will be shown to be related to different kinds of ``weighted'' Killing symmetries 
(see below).
Furthermore, we will show that equations (\ref{teukhd1}) and (\ref{teukhd2}) arise as projections {\em off shell} 
of a certain ``modified Laplace-de Rham operator'' acting on tensor-valued differential forms, 
which will be the generalization of the spinorial operator on the left hand side of (\ref{teuk4d})
to arbitrary dimensions.

Consider now the zero boost weight case.
Fields of this type turned out to be relevant recently for proving the nonmodal linear stability of the 
4-dimensional Schwarzschild (-de Sitter) black hole in \cite{Dotti:2013uxa, Dotti:2016cqy}.
The general case for the class of 4-dimensional Einstein type D spacetimes was analyzed in \cite{Araneda:2016iwr}, 
where it was shown that the decoupling and field reconstruction results are only valid for the case of real $\Psi_2$, 
static black holes being a particular case of this situation.
Therefore, for spacetimes in arbitrary dimensions, it is natural to restrict consideration in the first place to the case of 
static black holes when studying decoupling and field reconstruction mechanisms for boost weight zero components; 
this is the class that we will study in this work for $b=0$.
We will show that the equations also admit a $d$-dimensional wave-like structure in terms of the {\em same} family of weighted wave 
operators as in the Kundt class; and, moreover, we will show that these equations are also projections {\em off shell} 
of the {\em same} modified Laplace-de Rham operator acting on tensor valued differential forms. \\

Finally, it is interesting to analyze the relation between the modified covariant derivatives associated to the generalized wave operator 
$D^{\a}D_{\a}$, and the subject of symmetries and hidden symmetries associated to different kind of Killing objects. 
We will show that the {\em WANDs} (Weyl Alligned Null Directions, the generalizations of the 
principal null directions to higher dimensions) of the algebraically special spacetimes considered in this work, are conformal Killing 
vectors with respect to the weighted covariant derivative $D_{\a}$.
Furthermore, already in the 4-dimensional case it is interesting to note that, from (\ref{teuk4d}) one can see that the $b=0$ 
case is naturally described in terms of the well-known hidden symmetries of the type D: the object ${}^{0}P^{A_1...A_{2\s}}$ turns out to 
be, for $\s=1$ and $\s=2$, a Killing spinor (see \cite{Araneda:2016iwr}). 
In tensorial language, this object is associated to conformal Killing-Yano tensors.
It is then also interesting to investigate the existence and possible role of this kind of objects for 
perturbations of higher dimensional spacetimes.
We will show that the generalization of the spinor field ${}^{b}P^{A_1...A_{2\s}}$ to arbitrary dimensions is 
made out of tensors that satisfy a conformal Killing-Yano equation with respect to the weighted derivative $D_{\a}$, 
for {\em both} extreme and zero boost weight cases. 
We will also show that the same holds in 4 dimensions for the more general class of Einstein type D spacetimes.
Thus, our results can be thought of as a generalization of the {\em spin reduction} mechanism of Penrose 
\cite{Penrose3, Penrose2} valid in 4-dimensional Minkowski spacetime, to a ``weighted spin reduction'' 
valid for the higher dimensional, algebraically special spacetimes studied in this work
(and also for the class of $d=4$ Einstein type D spacetimes, see subsection \ref{sec-WKS}).

\subsection*{Conventions and overview}

We will consider $d$-dimensional spacetimes, with $d\geq4$ and mostly plus metric signature $(-+...+)$ 
(except in sections \ref{sec-review4d} and \ref{sec-moreon4d}, where we review and elaborate on 
the 4-dimensional results of \cite{Araneda:2016iwr}, which uses the signature $(+---)$). 
Greek indices $\a,\b,\g,...$ denote spacetime indices, and latin indices $i,j,k,...$ denote internal indices 
associated to the $d-2$ spatial vectors of the higher dimensional GHP formalism. 
The notation for spin coefficients, Weyl tensor components, etc. is the same as in \cite{Durkee:2010xq}.
Further notation will be introduced in section \ref{sec-sbh} when considering warped product spacetimes, 
and in section \ref{sec-moreon4d} where we use 4-dimensional spinor methods.
The GHP type of a quantity will always be denoted $\{b,s\}$, with $b$ the boost weight and $s$ the spin weight.
(Note that our notation for the GHP type in this work is different from the traditional 4-dimensional 
notation, where the GHP type is indicated as $\{p,q\}$ with $p=b+s$ and $q=b-s$.)

In section \ref{sec-prelim} we give some preliminary material, together with the main ideas and results of this work:
in subsection \ref{sec-review4d} we review the 4-dimensional results in \cite{Araneda:2016iwr} that are relevant for this work;
in \ref{sec-GHPhd} we give some basic facts about the higher dimensional GHP formalism (which 
was developed in \cite{Durkee:2010xq});
in \ref{sec-covder} we construct generalized wave operators using modified covariant derivatives; 
in \ref{sec-derham} we define ``modified Laplace de-Rham operators'' for tensor valued differential forms; 
and finally in subsection \ref{sec-mainr} we describe our main results.
In section \ref{sec-kundt} we apply our general ideas to perturbations of Kundt spacetimes,
and in section \ref{sec-sbh} we do the same for static black holes.
In section \ref{sec-moreon4d} we briefly extend some of the concepts of the previous sections to the case of 
4-dimensional Einstein spacetimes of Petrov type D, complementing the results in \cite{Araneda:2016iwr}.
Finally, we present our conclusions in section \ref{sec-conclusions}. 
We have also included some appendices with several additional identities and calculations that are needed 
for the proofs of the results in this work.

\section{Preliminaries and main results}\label{sec-prelim}

\subsection{Review of the 4-dimensional case}\label{sec-review4d}

We start by giving a motivation for this work, based on previous results of the author in the 4-dimensional 
case \cite{Araneda:2016iwr}.
We consider it useful to make a brief summary of the results in \cite{Araneda:2016iwr}, 
both as a guiding principle and 
to stress the particular features of this case, that might not be present in the higher dimensional case. 
We change slightly the notation with respect to \cite{Araneda:2016iwr} in order to adjust it to the present work.
We recall that, in this section, we use the metric signature $(+---)$.

Consider a Petrov type D spacetime with associated principal null tetrad $\{\ell^{\a},n^{\a},m^{\a},\bar{m}^{\a}\}$ 
and Weyl scalar $\Psi_2$, that satisfies the vacuum Einstein equations with cosmological constant, 
$G_{\a\b}+\L g_{\a\b}=0$. In what follows we use the 4-dimensional GHP formalism adapted to the principal null tetrad.
Introduce the complex 1-form $\G_{\a}:=-(\rho+\e)n_{\a}+\e'\ell_{\a}-\b'm_{\a}+(\b+\t)\bar{m}_{\a}$, 
the modified covariant derivative 
\begin{equation}\label{mcd4d}
D_{\a}=\c_{\a}+p\G_{\a}
\end{equation}
acting on type $\{b=p/2,s=p/2\}$ GHP quantities, and the corresponding modified wave operator \cite{Bini, Andersson1}
\begin{equation}\label{opteuk4d}
 \teuk_p:=g^{\a\b}D_{\a}D_{\b}
\end{equation}
(note that $\teuk_{0}=\Box$). Define the 2-forms
\begin{equation}\label{Mforms}
 U_{\a\b}:=2\ell_{[\a}m_{\b]}, \;\;\;  W_{\a\b}:=2\ell_{[\a}n_{\b]}+2\bar{m}_{[\a}m_{\b]}, \;\;\; V_{\a\b}:=2\bar{m}_{[\a}n_{\b]},
\end{equation}
and the tensors with Weyl symmetries
\begin{align}
 U_{\a\b\g\d}&:= U_{\a\b}U_{\g\d}, \label{W0}\\
 W_{\a\b\g\d}&:=U_{\a\b}V_{\g\d}+V_{\a\b}U_{\g\d}+W_{\a\b}W_{\g\d}, \label{W2}\\
 V_{\a\b\g\d}&:=V_{\a\b}V_{\g\d}. \label{W4}
\end{align}
Define also the 1-form
\begin{equation}\label{gamma4d}
 \g_{\a}:=\Psi_2^{1/3}\c_{\a}\Psi_2^{-1/3},
\end{equation}
which, using the background Bianchi identities, can be rewritten as $\g_{\a}=-\rho n_{\a}-\rho'\ell_{\a}+\t'm_{\a}+\t\bar{m}_{\a}$.
Using these elements, {\em off shell} operator identities (i.e. without assuming that the field equations are satisfied) 
were found in \cite{Araneda:2016iwr} for Dirac, Maxwell and linearized gravitational fields, 
all of them as particular cases of the identity (\ref{teuk4d}).
These identities map the field equations operators to decoupled wave-like operators acting on rescaled 
components of the fields, 
and, using the adjoint operators method of Wald \cite{Wald}, one can reconstruct solutions of the original field equations 
from solutions of the scalar equations. 
Several results known in the literature (see e.g. \cite{Kegeles, Teukolsky, Fackerell, Andersson1, Dotti:2013uxa}, also 
\cite[section 6.4]{Penrose2}) 
follow as particular cases of these identities.
Below we show explicitly these operators equalities for Maxwell and linearized gravity. \\

Consider first the case of Maxwell fields.
For an arbitrary anti-self-dual 2-form $\wt{F}_{\a\b}:=F_{\a\b}+\frac{i}{2}\e_{\a\b}{}^{\g\d}F_{\g\d}$, where 
$\e_{\a\b\g\d}$ is the volume form, we have the operator equalities 
(see \cite[Theorem 1.2$'$]{Araneda:2016iwr})
\begin{align}
  U^{\b\g}(\c_{\g}+2\g_{\g})\c^{\a}\wt{F}_{\a\b}&=
 -(\teuk_{+2}-4\Psi_2 +\tfrac{2}{3}\L)[U^{\a\b}\wt{F}_{\a\b}] \label{maxwell1} \\
 \Psi^{-1/3}_2 W^{\b\g}(\c_{\g}+2\g_{\g})\c^{\a}\wt{F}_{\a\b}&=
 -(\Box+2\Psi_2 +\tfrac{2}{3}\L)[\Psi^{-1/3}_2 W^{\a\b}\wt{F}_{\a\b}] \label{maxwell2} \\
 \Psi_2^{-2/3}V^{\b\g}(\c_{\g}+2\g_{\g})\c^{\a}\wt{F}_{\a\b}&=
-(\teuk_{-2}-4\Psi_2 +\tfrac{2}{3}\L)[\Psi^{-2/3}_2 V^{\a\b}\wt{F}_{\a\b}].  \label{maxwell3}
\end{align}
For extreme boost weight $b=\pm 1$, the wave-like equations on the right hand side above are the Teukolsky equations.
For boost weight $b=0$, the corresponding wave-like equation is known as the {\em Fackerell-Ipser} equation.

A similar pattern holds also for gravitational perturbations. 
Consider a monoparametric family of spacetimes $(\M,g(\ve))$ such that $(\M,g(0))$ is an Einstein spacetime of Petrov type D, 
and denote the first order perturbation of a quantity $T$ as $\dot{T}:=\frac{d}{d\ve}|_{\ve=0}T(\ve)$.
Defining the anti-self-dual Weyl tensor 
$\wt{C}_{\a\b\g\d}=C_{\a\b\g\d}+\frac{i}{2}\e_{\a\b}{}^{\e\l}C_{\e\l\g\d}$, one has the off shell identities
(see \cite[Theorem 1.3]{Araneda:2016iwr})
\begin{align}
 \tfrac{d}{d\ve}|_{\ve=0}\left[U^{\a\b\g\e}(\c_{\e}+4\g_{\e})\c^{\d}\wt{C}_{\a\b\g\d}\right]
  =&(\teuk_{+4}-16\Psi_2 +\tfrac{2}{3}\L)\dot\Psi_0, \label{grav1}  \\
 \nonumber \tfrac{d}{d\ve}|_{\ve=0}\left[\tfrac{1}{6}\Psi^{-2/3}_2 W^{\a\b\g\e}(\c_{\e}+4\g_{\e})\c^{\d}\wt{C}_{\a\b\g\d}\right]
  =&(\Box+8\Psi_2 +\tfrac{2}{3}\L)[\Psi^{-2/3}_2\dot{\Psi}_2]\\
  &+3(\dot{\Box}_h+\tfrac{\dot{R}_h}{6})\Psi^{1/3}_2, \label{grav2} \\
  \tfrac{d}{d\ve}|_{\ve=0}\left[\Psi_2^{-4/3}V^{\a\b\g\e}(\c_{\e}+4\g_{\e})\c^{\d}\wt{C}_{\a\b\g\d}\right]
  =&(\teuk_{-4}-16\Psi_2 +\tfrac{2}{3}\L)[\Psi^{-4/3}_2\dot\Psi_4],  \label{grav3}
\end{align}
where $\dot{\Box}_h$ and $\dot{R}_h$ are, respectively, the linearized wave operator and curvature scalar associated to the 
metric perturbation $h_{\a\b}=\dot{g}_{\a\b}$. 
The scalar equations for extreme boost $b=\pm2$ are the Teukolsky equations. On the other hand,
note that, even on shell, eq. (\ref{grav2}) is not decoupled in the sense that $\dot\Psi_2$ is not the only perturbed 
quantity in the equation, since we have the linearized conformal wave operator $(\dot\Box_{h}+\frac{\dot{R}_h}{6})$ 
applied to the unperturbed field $\Psi^{1/3}_2$. In the particular case in which the background $\Psi_2$ is real, one can get a 
decoupled equation from (\ref{grav2}) by simply taking the imaginary part; this turns out to be the case for static black 
holes and it was fundamental for the proof of the nonmodal linear stability of Schwarzschild (-de Sitter) in
\cite{Dotti:2013uxa, Dotti:2016cqy}.

An additional complication of the gravitational system (with respect to the Maxwell case) is the fact that 
gravitational perturbations are not described in terms of the curvature tensor but in terms of the linearized metric, 
and the identities (\ref{grav1})-(\ref{grav3}) are for perturbations of the curvature.
In \cite{Araneda:2016iwr} this problem was solved by using the Bianchi identities and some properties of the differential
operator applied to $\wt{C}_{\a\b\g\d}$ on the left hand side of equations (\ref{grav1})-(\ref{grav3}), which 
allowed us to show that (see \cite[Theorem 1.3$'$]{Araneda:2016iwr})
\begin{align}
 U^{\a\g\b\e}(\c_{\e}+4\g_{\e})\c_{\g}(\dot{G}_{\a\b}[h]+\L h_{\a\b})
  =&(\teuk_{+4}-16\Psi_2 +\tfrac{2}{3}\L)\dot\Psi_0 , \label{gravh1} \\
\nonumber \tfrac{1}{6}\Psi^{-2/3}_2 W^{\a\g\b\e}(\c_{\e}+4\g_{\e})\c_{\g}(\dot{G}_{\a\b}[h]+\L h_{\a\b})
  =& (\Box+8\Psi_2 +\tfrac{2}{3}\L)[\Psi^{-2/3}_2\dot{\Psi}_2]\\
  &+3(\dot{\Box}_h+\tfrac{\dot{R}_h}{6})\Psi^{1/3}_2, \label{gravh2} \\
 \Psi_2^{-4/3}V^{\a\g\b\e}(\c_{\e}+4\g_{\e})\c_{\g}(\dot{G}_{\a\b}[h]+\L h_{\a\b})]
  =&(\teuk_{-4}-16\Psi_2 +\tfrac{2}{3}\L)[\Psi^{-4/3}_2\dot\Psi_4] \label{gravh3}  
\end{align}
(cf. equations (\ref{grav1})-(\ref{grav3})).
These identities are very useful because one can take the adjoint equations in the spirit of \cite{Wald},
and, using the self-adjoint properties of many of the involved operators, construct metric perturbations in a very compact 
form from solutions of the scalar equations.

In the present work we obtain the generalization of equations (\ref{maxwell1})-(\ref{gravh3}) to higher dimensions.
It is important to note here a unique feature of the 4-dimensional case, related to the use of 
{\em Hodge duality} in equations (\ref{maxwell1})-(\ref{grav3}). 
Both Maxwell equations, $\c^{\a}F_{\a\b}=0$ and $\c_{[\a}F_{\b\g]}=0$, can be combined into 
a single complex equation $\c^{\a}\wt F_{\a\b}=0$ for the anti-self-dual object $\wt{F}_{\a\b}$. 
Similarly, the gravitational equations are combined in the complex equation $\c^{\d}\wt{C}_{\a\b\g\d}=0$.
This is only possible because we have at our disposal a natural 4-form (the volume form) that maps, 
via Hodge duality, 2-forms into 2-forms. 
Spinor techniques in four dimensions are especially powerful when working with this kind of objects, the identities 
(\ref{maxwell1})-(\ref{grav3}) were demonstrated in \cite{Araneda:2016iwr} using these methods.
However, in arbitrary dimensions this is no longer possible, since the volume form is, naturally, a $d$-form. 
This problem will lead us in sections \ref{sec-derham}, \ref{sec-mainr} to the consideration of modified Laplace-de Rham-like operators. \\

Finally, it is worth noting an important feature of the equations for boost weight zero, (\ref{maxwell2}) and (\ref{grav2}), 
which is the appearance of objects associated to {\em hidden symmetries}.
In (\ref{maxwell2}), the tensor $\wt{Y}_{\a\b}:=\Psi^{-1/3}_2 W_{\a\b}$ 
is a (complex) conformal Killing-Yano tensor, or equivalently, the tensorial analog of a valence $(2,0)$ Killing spinor.
$\wt{Y}_{\a\b}$ can be written as $\wt{Y}_{\a\b}=\frac{1}{2}(Y_{\a\b}+iZ_{\a\b})$, 
where $Y_{\a\b}$ is an ordinary Killing-Yano tensor, and its dual $Z_{\a\b}$ is a (real) conformal Killing-Yano tensor.
Similarly, in eq. (\ref{grav2}), the object $L_{\a\b\g\d}=\Psi^{-2/3}_2 W_{\a\b\g\d}$ is the tensorial version of a 
valence $(4,0)$ Killing spinor \cite{Araneda:2016iwr}.
As we briefly review in subsection \ref{sec-WKS} below, Killing spinors are associated to (conformal) 
Killing vectors, Killing tensors, and Killing-Yano forms; and they are ultimately responsible (see \cite{Penrose4}) for the 
integrability of the geodesic equation and the separability of the Klein-Gordon equation on the Kerr spacetime.
For this reason, the subject of hidden symmetries on black hole spacetimes has been extensively 
studied in the literature; see e.g. the recent review \cite{Frolov:2017kze} and references therein. 
In the following sections we will show that, for the higher dimensional spacetimes studied in this work, 
a ``generalized'' form of Killing symmetries not only appears for boost weight zero but also for extreme boost weight, 
in connection with the covariant derivative associated to the generalization of the weighted wave operator (\ref{opteuk4d}).
Furthermore, in section \ref{sec-moreon4d} we will show that the same results hold in the 4-dimensional case.

\subsection{GHP formalism in higher dimensions}\label{sec-GHPhd}

The generalization of the GHP formalism to higher dimensions was developed in \cite{Durkee:2010xq}. 
Here we will give a description of it from the point of view of the theory of fiber bundles and connections on them, 
since it is better suited for our purposes.

Let $(\M,g)$ be a $d-$dimensional spacetime, with $d\geq4$. As in the standard 4-dimensional 
GHP formalism, we wish to adapt a framework to a pair of null directions $\ell^{\a}$, $n^{\a}$.
Therefore we consider in each tangent space a frame
$\{e^{\a}_{a}\}=\{e^{\a}_0=\ell^{\a},e^{\a}_1=n^{\a},e^{\a}_i=m^{\a}_i\}$, where $i=2,...,d-1$, 
the vectors $\ell^{\a}$ and $n^{\a}$ are null, and the vectors $m^{\a}_i$ are spacelike. 
The dual basis is $\{e^{a}_{\a}\}=\{e^{0}_{\a}=n_{\a},e^{1}_{\a}=\ell_{\a},e^{i}_{\a}=m^{i}_{\a}\}$.
The metric tensor in this frame is
\begin{equation}
 g_{\a\b}=2\ell_{(\a}n_{\b)}+\d_{ij}m^{i}_{\a}m^{j}_{\b},
\end{equation}
the only non-trivial products being $\ell^{\a}n_{\a}=1$ and $m^{\a}_i m_{j\a}=\d_{ij}$.
Internal indices $i,j,k,...$ will be raised and lowered using $\d^{ij}$ and $\d_{ij}$.

Under a choice of null directions, the $SO(d-1,1)$ orthonormal frame bundle is reduced to a 
$G_o$-principal bundle $B$, where $G_o$ is the subgroup of $SO(d-1,1)$ that preserves the null directions.
This subgroup is $G_o=\mbb R^{\times}\times SO(d-2)$\footnote{we denote by $\mbb{R}^{\times}$ the multiplicative 
group of real numbers.}, where its action on the frame $\{e^{\a}_{a}\}$ is
\begin{equation}\label{transf-frame}
 \ell^{\a}\to \l\ell^{\a}, \hspace{1cm} n^{\a}\to \l^{-1}n^{\a}, \hspace{1cm}  m^{\a}_i\to X_{i}{}^{j}m^{\a}_j
\end{equation}
with $\l\in\mbb R^{\times}$ and $X\in SO(d-2)$ (both $\l$ and $X$ depend on the spacetime points). 
The components of a tensor field $T$ obtained by projecting it on the frame $\{e^{\a}_{a}\}$ 
get transformed under (\ref{transf-frame}); in other words, 
the action (\ref{transf-frame}) defines a representation of $G_o$ given by 
\begin{equation}\label{repG}
 T_{i_1...i_s}\mapsto (\Pi_{(b,s)}(\l,X)T)_{i_1...i_s}=\l^{b}X_{i_1}{}^{j_1}...X_{i_s}{}^{j_s}T_{j_1...j_s}
\end{equation}
for some numbers $b\in\mbb{Z}$, $s\in\mbb{N}_0$.
We say that the elements transforming under this representation are {\em type $\{b,s\}$ quantities}, 
or equivalently they have {\em boost weight} $b$ and {\em spin weight} $s$.
Note that
\begin{equation}\label{Pi}
\Pi_{(b,s)}\equiv\Pi_b\otimes\Pi_s,
\end{equation}
where $\Pi_b$ is the representation of $\mbb{R}^{\times}$ given by $(\Pi_b(\l) T)_{i_1...i_s}=\l^{b}T_{i_1...i_s}$, 
and $\Pi_s=\Pi^{\otimes s}_1$, with $\Pi_1$ the spin 1 representation of the rotation group $SO(d-2)$.
At each spacetime point,
the quantities of a well-defined type $\{b,s\}$ form a vector space $V$ (carrying the representation (\ref{repG}) 
of $G_o$); the quantities of all types together form a graded algebra 
(which implies that the product of a quantity of type $\{b,s\}$ with another quantity of type 
$\{b',s'\}$ is of type $\{b+b',s+s'\}$).
A {\em GHP scalar field} of type $\{b,s\}$ is thought of as a section of the associated vector bundles 
$E_{(b,s)}:=B\times_{\Pi_{(b,s)}}V$.

Let $\w_{\a}{}^{a}{}_{b}=e^{a}_{\b}\c_{\a}e^{\b}_{b}$ be the connection 1-form on the $SO(d-1,1)$-bundle $P$. 
The behavior of the different components of $\w_{\a}{}^{a}{}_{b}$ under the GHP subgroup $G_o$ 
gives rise, on the one hand, to the so-called spin coefficients, and on the other hand, to the connection 1-form 
on the reduced bundle $B$, which allows to define the covariant derivative on GHP quantities. 
More specifically, the parts of $\w_{\a}{}^{a}{}_{b}$ that transform {\em covariantly} under $G_o$ are 
$\w_{\a}{}^{0}{}_{i}=:-m^{\b}_{i}N_{\b\a}$
and $\w_{\a}{}^{1}{}_{i}=:-m^{\b}_{i}L_{\b\a}$, where we have adopted the notation of \cite{Durkee:2010xq}:
\begin{equation}
 L_{\a\b}:=\c_{\b}\ell_{\a}, \hspace{1cm} N_{\a\b}:=\c_{\b}n_{\a}, \hspace{1cm} 
 M^{i}{}_{\a\b}:=\c_{\b}m^{i}_{\a}.
\end{equation}
Note that, projecting on the frame $\{e^{\a}_a\}$, we get the components $L_{ab}=e^{\a}_a e^{\b}_b L_{\a\b}$, 
$N_{ab}=e^{\a}_a e^{\b}_b N_{\a\b}$ and $M^{i}{}_{ab}=e^{\a}_a e^{\b}_b M^{i}{}_{\a\b}$.
The GHP spin coefficients are then defined through $N_{ia}:=e^{\a}_{a}m^{\b}_{i}N_{\b\a}$ and $L_{ia}:=e^{\a}_{a}m^{\b}_{i}L_{\b\a}$, 
and the specific notation for each $a=0,1,i$ is given in table \ref{tablaGHPcoef}  (see \cite{Durkee:2010xq} for the 
interpretation of the coefficients).
\begin{table}%[H]
\begin{center}
\begin{tabular}{|c|c|c|c|}\hline
 Component & Notation & boost $b$ & spin $s$ \\ \hline
 $L_{ij}$ & $\rho_{ij}$ & 1 & 2 \\
 $L_{ii}$ & $\rho=\d^{ij}\rho_{ij}$ & 1 & 0 \\
 $L_{i0}$ & $\kappa_{i}$ & 2 & 1 \\
 $L_{i1}$ & $\tau_{i}$ & 0 & 1 \\
 $N_{ij}$ & $\rho'_{ij}$ & $-1$ & 2 \\
 $N_{ii}$ & $\rho'=\d^{ij}\rho'_{ij}$ & $-1$ & 0 \\
 $N_{i1}$ & $\kappa'_{i}$ & $-2$ & 1 \\
 $N_{i0}$ & $\tau'_{i}$ & 0 & 1 \\ \hline
\end{tabular}
\caption{Notation for GHP spin coefficients \cite{Durkee:2010xq}.}
\label{tablaGHPcoef}
\end{center}
\end{table}
\noindent
We note that in the case in which $\ell^{\a}$ is geodesic, we have $\k_i\equiv 0$. 
On the other hand, it is useful to note also that the GHP quantity 
$\rho_{ij}$ can be decomposed in the following way:
\begin{equation}\label{rho}
 \rho_{ij}=\sigma_{ij}+\tfrac{\rho}{(d-2)}\d_{ij}+\rho_{[ij]},
\end{equation}
where $\sigma_{ij}=\sigma_{(ij)},\rho$ and $\rho_{[ij]}$ represent the shear, expansion, and twist 
of the congruence, respectively.
Analogous statements hold for $n^{\a}$ replacing the corresponding quantities by their primed versions.

Now, the parts of $\w_{\a}{}^{a}{}_{b}$ that {\em do not} transform covariantly under $G_o$ define the 
induced connection 1-form $\chi$ on the reduced bundle $B$.
As usual, $\chi$ takes values in the Lie algebra $\mf{g}_o=\mbb{R}\oplus\mf{so}(d-2)$, 
thus we can decompose it as $\chi=(A,B)$, where $A\in T^{*}\M\otimes\mbb R$ and $B\in T^{*}\M\otimes\mf{so}(d-2)$.
It is not difficult to show that these components are
\begin{equation}\label{connections}
 A_{\a}=\w_{\a}{}^{1}{}_{1}, \hspace{1cm} B_{\a}{}^{i}{}_{j}=\w_{\a}{}^{i}{}_{j}
\end{equation}
since they transform under $G_o$ as 
\begin{align}
  A_{\a} \to & \; \l\c_{\a}\l^{-1}+A_{\a}, \\
  B_{\a}{}^{i}{}_{j} \to & \; X^{i}{}_{k}\c_{\a}X_{j}{}^{k}+X^{i}{}_{k}B_{\a}{}^{k}{}_{l}X_{j}{}^{l}, 
\end{align} 
which corresponds to the transformation law for a connection 1-form 
\footnote{we recall that in the theory of fiber bundles, if the components $\psi$ of a section of an associated 
vector bundle $E=P\times_{\Pi}V$ transform as $\psi\to\Pi(g)\psi$ for $g\in G$ (which implies that a section $\sigma$ of the 
principal bundle $P$ transforms as $\sigma\to\sigma\cdot g^{-1}$), then the connection 1-form $\w$ transforms as 
$\w\to gdg^{-1}+g\w g^{-1}$ \cite{Nakahara}.}.
Explicitly:
\begin{align}
 A_{\a}&=-L_{1a}e^{a}_{\a}=-L_{10}n_{\a}-L_{11}\ell_{\a}-L_{1k}m^{k}_{\a}, \label{bconnection} \\
 B_{\a}{}^{i}{}_{j}&=-M^{i}{}_{ja}e^{a}_{\a}=-M^{i}{}_{j0}n_{\a}-M^{i}{}_{j1}\ell_{\a}-M^{i}{}_{jk}m^{k}_{\a}. \label{sconnection}
\end{align}
With these elements we can now introduce the covariant derivative on the associated bundles $E_{(b,s)}$.
As explained (for general situations) in section \ref{sec-covder} below, the induced GHP covariant derivative is 
\begin{equation}
 \T_{\a}T_{i_1...i_s}=\c_{\a}T_{i_1...i_s}+(\pi_{(b,s)}(\chi_{\a})T)_{i_1...i_s}, \label{Theta}
\end{equation}
where $\pi_{(b,s)}:\mf{g}_o\to \mf{gl}(V)$ is the representation of the Lie algebra $\mf{g}_o$ 
induced (by differentiation) by (\ref{repG}). In view of (\ref{Pi}), the structure of $\pi_{(b,s)}$ is
$\pi_{(b,s)}=\pi_b\otimes\mbb I+\mbb I\otimes\pi_s$
(with $\mbb{I}$ the identity map in the corresponding space) and its action is given by
\begin{equation}\label{repg}
 (\pi_{(b,s)}(a,x)T)_{i_1...i_s}=baT_{i_1...i_s}+x_{i_1}{}^{k}T_{ki_2...i_s}+...+x_{i_s}{}^{k}T_{i_1...k}
\end{equation}
where $a\in\mbb R$ and $x\in\mf{so}(d-2)$. Therefore, the explicit form of the covariant derivative (\ref{Theta}) is
\begin{equation}\label{Theta2}
 \T_{\a}T_{i_1...i_s}=\c_{\a}T_{i_1...i_s}+bA_{\a}T_{i_1...i_s}+B_{\a i_1}{}^{k}T_{ki_2...i_s}
 +...+B_{\a i_s}{}^{k}T_{i_1...k}. 
\end{equation}
The generalization of the 4-dimensional {\em thorn} and {\em eth} operators \th,\th$'$,\dh,\dh$'$ to higher dimensions 
is simply given by the directional derivatives along the frame vectors using the GHP connection:
\begin{equation}
 \text{\th}:=\ell^{\a}\T_{\a}, \hspace{1cm} \text{\th}':=n^{\a}\T_{\a}, \hspace{1cm} \text{\dh}_j:=m^{\a}_{j}\T_{\a}.
\end{equation}
Their action on a GHP scalar of type $\{b,s\}$ is deduced from (\ref{Theta2}) and (\ref{bconnection}), (\ref{sconnection}):
\begin{align}
 \text{\th}T_{i_1...i_s}&=DT_{i_1...i_s}-bL_{10}T_{i_1...i_s}+M^{k}{}_{i_10}T_{ki_2...i_s}+...+M^{k}{}_{i_s0}T_{i_1...k} \\
 \text{\th}'T_{i_1...i_s}&=D'T_{i_1...i_s}-bL_{11}T_{i_1...i_s}+M^{k}{}_{i_11}T_{ki_2...i_s}+...+M^{k}{}_{i_s1}T_{i_1...k} \\
 \text{\dh}_jT_{i_1...i_s}&=\d_jT_{i_1...i_s}-bL_{1j}T_{i_1...i_s}+M^{k}{}_{i_1j}T_{ki_2...i_s}+...+M^{k}{}_{i_sj}T_{i_1...k},
\end{align}
where $D=\ell^{\a}\c_{\a}$, $D'=n^{\a}\c_{\a}$ and $\d_{j}=m^{\a}_j\c_{\a}$ are traditional Newman-Penrose derivatives.
These formulae coincide of course with those given in the literature, see e.g. \cite{Durkee:2010xq}. 
For explicit expressions for the commutators between GHP weighted derivatives (which will be extensively used in our work), 
we refer the reader to \cite{Durkee:2010xq}.\\

Finally, we recall the basics of the classification of the Weyl tensor in higher dimensions, 
which is the generalization of the Petrov Classification of 4-dimensional spacetimes 
and was originally developed in \cite{Coley:2004jv}. 
We will, however, follow the conventions of \cite{Durkee:2010qu}, where algebraically special spacetimes are defined 
slightly differently.
The notation for the components of the higher dimensional Weyl tensor is given in table \ref{WeylComp}. 
For the Riemann tensor components we will use the same capital Greek letters but with a tilde 
(note that the trace identities in table \ref{WeylComp} do not hold in general for the Riemann tensor).

\begin{table}%[H]
\begin{center}
\begin{tabular}{|c|c|c|c|l|}\hline
 Component & Notation & Boost & Spin & Identities  \\ \hline
 $C_{0i0j}$ & $\Om_{ij}$ & $+2$ & 2 & $\Om_{ij}=\Om_{ji}$, $\Om_{ii}=0$ \\
 $C_{0ijk}$ & $\Psi_{ijk}$ & $+1$ & 3 & $\Psi_{ijk}=-\Psi_{ikj}$, $\Psi_{[ijk]}=0$\\
 $C_{010i}$ & $\Psi_i$ & $+1$ & 1 & $\Psi_i=\Psi_{kik}$\\
 $C_{ijkl}$ & $\Phi_{ijkl}$ & $0$ & 4 & $\Phi_{ijkl}=\Phi_{[ij][kl]}=\Phi_{klij}$, $\Phi_{i[jkl]}=0$ \\
 $C_{0i1j}$ & $\Phi_{ij}$ & $0$ & 2 & $\Phi_{(ij)}\equiv\Phi^S_{ij}=-\frac{1}{2}\Phi_{ikjk}$\\
 $C_{01ij}$ & $2\Phi^{A}_{ij}$ & $0$ & 2 & $\Phi^A_{ij}\equiv\Phi_{[ij]}$\\
 $C_{0101}$ & $\Phi$ & $0$ & 0 & $\Phi=\Phi_{ii}$\\
 $C_{1ijk}$ & $\Psi'_{ijk}$ & $-1$ & 3 & $\Psi'_{ijk}=-\Psi'_{ikj}$, $\Psi'_{[ijk]}=0$ \\
 $C_{101i}$ & $\Psi'_i$ & $-1$ & 1 & $\Psi'_i=\Psi'_{kik}$ \\
 $C_{1i1j}$ & $\Om'_{ij}$ & $-2$ & 2 & $\Om'_{ij}=\Om'_{ji}$, $\Om'_{ii}=0$ \\ \hline
\end{tabular}
\caption{Weyl tensor components in higher dimensions \cite{Durkee:2010qu}.}
\label{WeylComp}
\end{center}
\end{table}

A null vector field $\ell^{\a}$ is said to be a {\em Weyl Aligned Null Direction} (WAND)
if all the boost weight $b=+2$ components of the Weyl tensor vanish everywhere in a frame containing 
$\ell^{\a}$, i.e. $\Om_{ij}\equiv 0$. This condition is independent of the choice of the rest of the 
frame vectors.
In $d=4$, the WAND condition is equivalent to the PND (Principal Null Direction) 
condition, but for $d>4$ this is not the case. 
A WAND $\ell^{\a}$ is called a {\em multiple-WAND} if all boost weight $b=+1$ components of the Weyl tensor also 
vanish, i.e. $\Om_{ij}=\Psi_{ijk}\equiv 0$. In $d=4$, this condition is equivalent to $\ell^{\a}$ being 
a repeated PND.
We then say that a spacetime is {\em algebraically special} if it admits a multiple-WAND.
In Table \ref{WeylClassification} we summarize the different algebraic types according to the existence 
of multiple-WANDs.

This classification is called {\em primary} because it is only focused in choosing the vector $\ell^{\a}$
such that as many as possible of the high boost weight components vanish. Once $\ell^{\a}$ is fixed, 
one can try to choose the other null vector $n^{\a}$ in a way such that as many low boost weight 
components as possible vanish; this defines then a {\em secondary classification}. 
This gives rise to the definition of a {\em type D spacetime} in higher dimensions: 
a spacetime is said to be type D if both null directions $\ell^{\a}$ and $n^{\a}$ are 
multiple-WANDs. In such a frame, we have $\Om_{ij}=\Psi_{ijk}=\Om'_{ij}=\Psi'_{ijk}\equiv0$.

\begin{table}%[H]
\begin{center}
\begin{tabular}{|c|l|}\hline
 Algebraic type & Definition  \\ \hline
 O & $C_{\a\b\g\d}\equiv 0$ \\
 N & $\Om_{ij}\equiv\Psi_{ijk}\equiv\Phi_{ijkl}\equiv\Psi'_{ijk}\equiv0$ \\
 III & $\Om_{ij}\equiv\Psi_{ijk}\equiv\Phi_{ijkl}\equiv0$ \\
 II & $\Om_{ij}\equiv\Psi_{ijk}\equiv0$ \\
 I &  $\Om_{ij}\equiv0$ \\
 G & There does not exist a WAND. \\ \hline
\end{tabular}
\caption{Algebraic types of the Weyl tensor in higher dimensions (primary classification), 
according to \cite{Durkee:2010xq}.}
\label{WeylClassification}
\end{center}
\end{table}

Several black hole solutions in higher dimensional spacetimes are of algebraic type D. 
For instance, the Schwarzschild-Tangherlini, Myers-Perry and Kerr-(A)dS black holes, and 
the black string/brane metrics, all belong to this algebraic type for all dimensions $d$.
There are other black objects, however, which are not type D, for example the black ring 
(see \cite{Pravda:2005kp}).

\subsection{Modified covariant derivatives and wave operators}\label{sec-covder}

In this section we define wave operators on the associated bundles $E_{(b,s)}$ that generalize the 
standard D'Alembertian. In terms of these operators, the higher dimensional Teukolsky and related equations 
will adopt a very compact and elegant form, as happens in the 4-dimensional case.
Furthermore, in the next sections we will show interesting relations between the modified covariant derivatives 
here introduced, and tensors associated to different kind of Killing symmetries on the spacetime.
Our approach here is very natural from the point of view of
the general theory of connections on principal fiber bundles, and the covariant 
derivatives they induce on associated vector bundles. In the Physics literature, this is often applied 
to the study of Yang-Mills theories.
We will very briefly summarize the general ideas that are needed for this work, and then apply them to construct 
the generalized wave operators.

Consider a principal $G$-bundle $P$ over the spacetime $\M$. A representation $(\Pi,V)$ of $G$ on the vector space 
$V$ gives rise to the associated vector bundle $E=P\times_{\Pi} V$, the sections of which are tensor fields that 
transform under $\Pi$. 
One can define a covariant derivative on these fields by using a connection 1-form $\U_{\a}$ on $P$.
As is well known, the connection takes values in the Lie algebra $\mf{g}=\text{Lie}(G)$.
By differentiation, the representation $\Pi$ gives origin to a Lie algebra representation 
$\pi:\mf{g}\to\mf{gl}(V)$; then it can be shown that an explicit formula for the covariant derivative of a 
section $\Psi:\M\to E$ is \footnote{we note that there is a sign mistake in formula (2.5) in \cite{Araneda:2016iwr}, 
which is compensated along the text by using minus the connection form $\w_{\a}$ and the transformation rule 
$\w\to \w+g^{-1}dg$ for the 4-dimensional (abelian) GHP group. 
One can check that in the present formula (\ref{covder}), under a gauge transformation 
$\Psi\to\Pi(g)\Psi$, $\U \to gdg^{-1}+g\U g^{-1}$, one has $D_{\a}\Psi \to \Pi(g)D_{\a}\Psi$.}
\begin{equation}\label{covder}
 D_{\a}\Psi:=\p_{\a}\Psi+\pi(\U_{\a})\Psi,
\end{equation}
see e.g. \cite{Nakahara}.
Note that the Lie algebra $\mf{g}$ will in general be of the form $\mf{g}=\mf{gl}(d,\mbb{R})\oplus\mf{g}_o$, where 
$\mf{g}_o$ is the Lie algebra of some internal symmetry group, therefore the connection 1-form $\U_{\a}$ 
can be split into the Christoffel symbols $\G$ and the internal connection, $\U\equiv(\G,\chi)$; 
this allows us to write (\ref{covder}) as $D_{\a}\Psi=\c_{\a}\Psi+\pi_o(\chi_{\a})\Psi$.
There is a natural wave operator associated to the covariant derivative $D_{\a}$, 
namely $g^{\a\b}D_{\a}D_{\b}$, and this is the kind of operators we will extensively use.

For the higher dimensional GHP formalism, as described in section \ref{sec-GHPhd}, 
the gauge group is $G_o=\mbb{R}^{\times}\times SO(d-2)$,
therefore a connection 1-form takes values in $\mf{g}_o=\mbb{R}\oplus\mf{so}(d-2)$.
The natural connection is given by $\chi=(A,B)$, defined in (\ref{bconnection}), (\ref{sconnection}), and the 
covariant derivative is (\ref{Theta}). Therefore we have the generalized wave operator
\begin{equation}\label{gbox}
 \Box_{(b,s)}:=g^{\a\b}\T_{\a}\T_{\b}
\end{equation}
acting on GHP tensors of type $\{b,s\}$. In particular, note that $\Box_{(0,0)}=\Box$.
In terms of weighted directional derivatives, one can check that
\begin{equation}\label{MBox1}
 \Box_{(b,s)}=(\text{\th}+\rho)\text{\th}'+(\text{\th}'+\rho')\text{\th}+(\text{\dh}^{k}-\t^k-\t'^k)\text{\dh}_{k}.
\end{equation}
This formula is general and the type $\{b,s\}$ does not appear explicitly.
On the other hand, it is sometimes useful to commute the derivatives $\text{\th\th}'$ 
in (\ref{MBox1}), in which case the type $\{b,s\}$ enters explicitly in the commutator $[\text{\th},\text{\th}']$ 
and the general formula depends on it.

We are now interested in modifying the connection form $\chi$ by adding to it a $\mf{g}_o$-valued 
1-form $\zeta$ \footnote{see \cite{ehlers, Held, Andersson1} for related discussion in the 4-dimensional case.}. We define 
\begin{equation}
 \zeta:=(X,Z) \;\in T^{*}\M\otimes\mf{g}_o,
\end{equation}
where the (respectively) $\mbb{R}$-valued and $\mf{so}(d-2)$-valued 1-forms $X_{\a}$ and $Z_{\a i}{}^{j}$ are
\begin{align}
 X_{\a}&:=\tfrac{2\rho}{(d-2)}n_{\a}-\tau_i m^{i}_{\a}, \label{X} \\
 Z_{\a i}{}^{j}&:=(\rho_{i}{}^{j}-\rho^{j}{}_{i})n_{\a}-\tau_i m^{j}_{\a}+\tau^j m_{i\a}. \label{Z}
\end{align}
With this 1-form we introduce a modified covariant derivative on the vector bundles $E_{(b,s)}$,
\begin{equation}\label{defD}
 D_{\a}:=\T_{\a}+\pi_{(b,s)}(\zeta_{\a})
\end{equation}
where $\pi_{(b,s)}$ is the representation of $\mf{g}_o=\mbb R\oplus\mf{so}(d-2)$ 
given in (\ref{repg}). The explicit action of this derivative on a GHP tensor $T_{i_1...i_s}$ of 
type $\{b,s\}$ is
\begin{equation}\label{explD}
 D_{\a}T_{i_1...i_s}=\T_{\a}T_{i_1...i_s}+bX_{\a}T_{i_1...i_s}+Z_{\a i_1}{}^{k}T_{k...i_s}+...+Z_{\a i_s}{}^{k}T_{i_1...k}.
\end{equation}
Of course, by projecting $D_{\a}$ on a frame we obtain modified directional derivatives 
$\wt{\text{\th}}=\ell^{\a}D_{\a}$, $\wt{\text{\th}}{}'=n^{\a}D_{\a}$, and $\wt{\text{\dh}}_{i}=m^{\a}_i D_{\a}$. 
One can check that, on a type $\{b,1\}$ quantity $T_i$ for example,
\begin{align}
 \wt{\text{\th}}T_i=&\text{\th}T_i+\tfrac{2\rho}{(d-2)}bT_i+(\rho_{i}{}^{j}-\rho^{j}{}_{i})T_j, \\
 \wt{\text{\th}}{}'T_i=&\text{\th}'T_i, \\
 \wt{\text{\dh}}_kT_i=&\text{\dh}_kT_i -b\t_k T_i-\t_i T_k+\d_{ik}\t^j T_j
\end{align}
(the operator $\wt{\text{\th}}{}'$ actually coincides with the usual one for any type $\{b,s\}$ quantity, 
as a consequence of the definitions (\ref{X}), (\ref{Z})).
Naturally associated to this derivative is the modified wave operator
\begin{equation}\label{mgbox}
 \teuk_{(b,s)}:=g^{\a\b}D_{\a}D_{\b}.
\end{equation}
Again, note that $\teuk_{(0,0)}=\Box$. 
In terms of the modified directional derivatives, we have 
\begin{equation}\label{mgbox2}
 \teuk_{(b,s)}=(\wt{\text{\th}}+\tfrac{2}{(d-2)}\rho)\text{\th}'+(\text{\th}'+\rho')\wt{\text{\th}}
   +(\wt{\text{\dh}}{}^{k}-(d-2)\t^k-\t'^k)\wt{\text{\dh}}_{k}.
\end{equation}
This is the modified wave operator that we will use in the rest of the paper, for the spacetimes of interest in this work.

For the Kundt class, we have $\rho_{ij}\equiv 0$ and the modified directional derivatives $\wt{\text{\th}}$ and $\wt{\text{\th}}{}'$
coincide with the usual ones. In this case it will be useful to commute the $\text{\th}\text{\th}'$ derivatives, in order to match 
the formulae known in the literature. 
In particular, the modified wave operator (\ref{mgbox2}) on an Einstein Kundt spacetime,
acting on a type $\{b,1\}$ GHP scalar $T_i$, is: 
\begin{align}
\nonumber \teuk_{(b,1)}T_i=& 2\text{\th}'\text{\th}T_i+\text{\dh}^k\text{\dh}_k T_i+\rho'\text{\th}T_i-2(b+1)\tau^k\text{\dh}_kT_i
 +2[\tau^k\text{\dh}_i - \tau_i\text{\dh}^k] T_k\\
 & +\left[(b^2-1)\t^k\t_k-\tfrac{2bd\L}{(d-1)(d-2)}\right]T_i-(d-4)\t_i\t^kT_k, \label{teukb1}  
\end{align}
while on a symmetric type $\{b,2\}$ GHP scalar $T_{ij}=T_{(ij)}$:
\begin{align} 
\nonumber \teuk_{(b,2)}T_{ij}=& 2\text{\th}'\text{\th}T_{ij}+\text{\dh}^k\text{\dh}_kT_{ij}+\rho'\text{\th}T_{ij}
 -2(b+1)\tau^k\text{\dh}_kT_{ij}+4[\tau^k\text{\dh}_{(i}-\tau_{(i}\text{\dh}^k]T_{j)k}\\
\nonumber &+\left[(b^2-2)\t^k\t_k-\tfrac{2bd\L}{(d-1)(d-2)} \right]T_{ij}-2(d-2)\tau^k\tau_{(i}T_{j)k} \\
 &+2\left[\t_i\t_j\d^{kl}+\d_{ij}\tau^k\tau^l \right]T_{kl}. \label{teukb2}
\end{align}
Note that, setting $b=1$ in (\ref{teukb1}) and $b=2$ in (\ref{teukb2}), the derivative terms match exactly those of equations
(\ref{teukhd1}) and (\ref{teukhd2}) respectively; the difference between these operators and those formulae will then constitute 
``potential'' terms, as in the 4-dimensional Teukolsky equations.

\subsection{Modified de Rham-like operators}\label{sec-derham}

Let $V_p$ be the bundle of $p$-forms on $\M$. The exterior derivative $d$ is a map $d:V_p\to V_{p+1}$, and it 
has a formal adjoint $d^{\dag}:V_p\to V_{p-1}$. One can compose both operations to form the Laplace-de Rham operator 
on $p$-forms, $dd^{\dag}+d^{\dag}d$. 
If we now consider {\em tensor-valued} $p$-forms, there is a natural generalization of these operations.
Let $V_{q,p}$ denote the bundle of $p$-forms with $q$ tensorial indices, so that an element of this space is
for example $\w_{\a_1...\a_q \b_1...\b_p}=\w_{\a_1...\a_q [\b_1...\b_p]}$. 
The {\em covariant exterior derivative} is a map $\Ds:V_{q,p}\to V_{q,p+1}$ that acts as the exterior derivative 
on the $p$-form indices and as the covariant derivative on the rest of the indices. It has a formal adjoint 
$\Ds^{\dag}:V_{q,p}\to V_{q,p-1}$ (with respect to the inner product (\ref{ip}) given below); and explicitly we have
\begin{align}
 (\Ds \w)_{\a_1...\a_q \b_1...\b_{p+1}}:=& (p+1)\c_{[\b_1}\w_{|\a_1...\a_q|\b_2...\b_{p+1}]}, \\
 (\Ds^{\dag} \w)_{\a_1...\a_q \b_2...\b_p}:=& -\c^{\b_1}\w_{\a_1...\a_q \b_1\b_2...\b_p}.
\end{align}
Note that for $q=0$ these operators coincide with the usual ones acting on ordinary $p$-forms. 
On the other hand, as a natural generalization of the ordinary Laplace-de Rham operator,
$\Ds$ and $\Ds^{\dag}$ give rise to a ``generalized'' Laplace-de Rham operator for tensor-valued $p$-forms, 
\begin{equation}\label{deRham}
 \Ds\Ds^{\dag}+\Ds^{\dag}\Ds.
\end{equation}

Now we modify these operators by adding a particular (type $\{0,0\}$) 1-form $\g_{\a}$, given by
\begin{equation}\label{1-form}
 \g_{\a}:=\tfrac{\rho}{(d-2)}n_{\a}+\tfrac{\rho'}{(d-2)}\ell_{\a}-\t_i m^{i}_{\a}.
\end{equation}
Consider a positive integer number $\s$ (this will play the role of the spin of the fields we will study). 
We define the following operators:
\begin{align}
 (\Ds_{\s} \w)_{\a_1...\a_q \b_1...\b_{p+1}}:=& (p+1)(\c_{[\b_1}+2\s\g_{[\b_1})\w_{|\a_1...\a_q|\b_2...\b_{p+1}]}, \label{D_s} \\
 (\Ds^{\star}_{\s} \w)_{\a_1...\a_q \b_2...\b_p}:=& -(\c^{\b_1}+2\s\g^{\b_1})\w_{\a_1...\a_q\b_1\b_2...\b_p}. \label{Dstar_s}
\end{align}
Observe that $\Ds^{\star}_{\s}$ {\em is not the adjoint of} $\Ds_{\s}$; they are related by $\Ds^{\star}_{\s}=\Ds^{\dag}_{-\s}$.
On the other hand, note that $\Ds_{0}=\Ds$, then we have $\Ds^{\dag}=\Ds^{\star}$.
We then introduce an operator $V_{q,p}\to V_{q,p}$ given by
\begin{equation}\label{mdeRham}
 \Ds_{\s}\Ds^{\star}+\Ds^{\star}_{\s}\Ds.
\end{equation}
Because of the resemblance of this operator with (\ref{deRham}), in the following we will call 
(\ref{mdeRham}) the ``spin $\s$ modified Laplace-de Rham operator''.
In differential forms notation, it can be written as
\begin{equation}
 (\Ds+2\s\g\wedge)\Ds^{\dag}+(\Ds^{\dag}+2\s\g^{\#}\lrcorner)\Ds,
\end{equation}
where $\g^{\#}$ is the ``sharp'' vector field obtained from the 1-form $\g$, and the exterior and interior products act on 
the differential form indices.

It is instructive to see some examples. In this work we will study Maxwell and gravitational fields. 
Maxwell fields have spin $\s=1$, and their curvature is represented by an ordinary 2-form, thus we have $\s=1$, $q=0$ and $p=2$. Then
\begin{equation}\label{mdrspin1}
 -[(\Ds_{1}\Ds^{\star}+\Ds^{\star}_{1}\Ds)F]_{\a\b}=2(\c_{[\a}+2\g_{[\a})\c^{\g}F_{|\g|\b]}+3(\c^{\g}+2\g^{\g})\c_{[\a}F_{\b\g]}.
\end{equation}
On the other hand, the gravitational field has spin $\s=2$, and its curvature is represented by a $2$-form with 2 extra 
tensorial indices; then $\s=2$, $q=2$, and $p=2$. Therefore
\begin{equation}\label{mdrspin2}
 -[(\Ds_{2}\Ds^{\star}+\Ds^{\star}_{2}\Ds)C]_{\a\b\g\d}=2(\c_{[\g}+4\g_{[\g})\c^{\e}C_{|\a\b\e|\d]}+3(\c^{\e}+4\g^{\e})\c_{[\e}C_{|\a\b|\g\d]}.
\end{equation}

\subsection{Main results}\label{sec-mainr}

Consider the spin $\s$ modified Laplace-de Rham operator (\ref{mdeRham}) for spins $\s=1$ and $\s=2$ 
(see (\ref{mdrspin1}) and (\ref{mdrspin2})).
Let $\phi_{\a_1...\a_n}$ be a (type $\{0,0\}$) tensor-valued $p$-form, with $n=2\s$ and $p\leq n$.
Consider also the generalized, modified wave operator (\ref{mgbox}), $\teuk_{(b,s)}=D^{\a}D_{\a}$.
We will define an appropriate type $\{b,s\}$ GHP tensor field ${}^{b}P^{\a_1...\a_n}_{i_1...i_s}$ with 
$n=2\s$ tensorial indices, constructed out of the principal null frame of an algebraically special spacetime 
(as mentioned, we will only study Kundt spacetimes and static black holes),
and we will prove that the projections {\em off shell} of the modified Laplace-de Rham operator under 
this tensor give, at the linear level, wave-like equations for ${}^{b}P^{\a_1...\a_n}_{i_1...i_s}\phi_{\a_1...\a_n}$.
More precisely, consider a monoparametric family of $d$-dimensional Lorentzian manifolds $(\M,g_{\a\b}(\ve))$, 
analytic around $\ve=0$, such that $(\M,g_{\a\b}(0))$ is an algebraically special Einstein spacetime; 
and  denote the first order perturbation of a quantity $T$ as $\dot{T}:=\frac{d}{d\ve}|_{\ve=0}T(\ve)$.
For $\s=1$, let $\phi_{\a_1\a_2}\equiv F_{\a_1\a_2}$ (an arbitrary --type $\{0,0\}$-- 2-form), and for $\s=2$ let
$\phi_{\a_1\a_2\a_3\a_4}\equiv C_{\a_1\a_2\a_3\a_4}(\ve)$ (the curvature tensor of $g_{\a\b}(\ve)$).
Our first main result is to show that, on the spacetimes of interest in this work, we have the off shell identity
\footnote{we note that in \cite{Bini, Bini2} it is argued that, in 4-dimensional spacetimes and assuming that the linearized field 
equations are satisfied (i.e. {\em on shell}), the Teukolsky equations are certain components of the {\em ordinary} --not modified-- 
Laplace-de Rham operator for tensor-valued differential forms; 
however, our results here are valid independently of the field equations (i.e. {\em off shell}), and, as our calculations reveal, 
it is the projection 
of the {\em modified} Laplace-de Rham operator which gives the decoupled equations, both in 4 and higher dimensions.}
\begin{equation}\label{maineq}
 -{}^{b}P^{\a_1...\a_n}_{i_1...i_s}[(\Ds_{\s}\Ds^{\star}+\Ds^{\star}_{\s}\Ds)\phi]_{\a_1...\a_n}
 \doteq (\teuk_{(b,s)}+V_{(b,s)})[{}^{b}P^{\a_1...\a_n}_{i_1...i_s}\phi_{\a_1...\a_n}]
\end{equation}
where again ``$\doteq$'' means that this is an equation valid to linear order 
(i.e., ``$\doteq$'' means ``$=$'' for $\s=1$, and for $\s=2$ it represents the operation of taking 
the linearization operator $\frac{d}{d\ve}|_{\ve=0}$ on both sides of the equation). 
See theorems \ref{thm-max-kundt}, \ref{thm-grav-kundt}, \ref{thm-cky-max}, \ref{thm-gky-max}, \ref{thm-grav-sbh}.
We emphasize that, since it is not assumed that the $\phi_{\a_1...\a_n}$ satisfies any field equation, 
this field can be dropped from both sides of (\ref{maineq}), leaving an operator identity.
The potential $V_{(b,s)}$ is really a ``matrix potential'', in the sense that it acts as matrix multiplication 
on the internal indices $i_1...i_s$. 
We will also show that the WANDs of the spacetime turn out to be {\em weighted} conformal Killing vectors with respect to the 
modified covariant derivative $D_{\a}$ (lemmas \ref{lem-wckv-kundt} and \ref{lem-wckv-sbh}), 
and that the GHP tensor ${}^{b}P^{\a_1...\a_n}_{i_1...i_s}$ can 
be associated to {\em weighted} conformal Killing-Yano tensors with respect to $D_{\a}$, 
for both extreme and zero boost weight cases (lemmas \ref{lem-wcky-kundt} and \ref{lem-wcky-sbh}). 
Furthermore, in section \ref{sec-moreon4d} we will show that the same results hold for 4-dimensional Einstein spacetimes 
of Petrov type D.

Now, assuming that the Maxwell or --vacuum with cosmological constant-- Einstein equations are satisfied, 
in view of (\ref{mdrspin1}) and (\ref{mdrspin2})
we have $[(\Ds_{\s}\Ds^{\star}+\Ds^{\star}_{\s}\Ds)\phi]_{\a_1...\a_n}=0$; i.e.
the left hand side of (\ref{maineq}) is zero, 
and we are left with a wave-like equation for the (rescaled) type $\{b,s\}$ component of $\phi_{\a_1...\a_n}$. 
We can also use the operator identity implied in (\ref{maineq}) to construct solutions of the original (i.e. Maxwell or 
linearized Einstein) field equations from solutions of the scalar wave-like equation.
This is our second main result, for which we will use the very powerful and well-known 
method of Wald based on the 
transposition of differential operators (see \cite{Wald}); this is why we prove eq. (\ref{maineq}) {\em off shell}, 
i.e. because we need operator identities. The inner product we use to define the adjoint of an operator is as usual
\begin{equation}\label{ip}
 \langle\Phi,\Psi\rangle=\int_{\M}\Phi\Psi
\end{equation}
(where it is understood a total contraction between the indices of $\Phi$ and $\Psi$); then the adjoint $O^{\dag}$
of an operator $O$ is given by $\langle\Phi,O\Psi\rangle=\langle O^{\dag}\Phi,\Psi\rangle$. 
When considering the adjoint operators in the {\em right hand side} of (\ref{maineq}), 
we can see an important advantage of our use of generalized wave operators, which
is the fact that the self-adjointness property of the wave-like equation is 
very clear: it is easily shown that, if $\Psi$ and $\Phi$ are GHP fields of type $\{b,s\}$ and $\{b',s'\}$ 
respectively, then 
\begin{equation}\label{adjteuk}
 \langle\Phi,\teuk_{(b,s)}\Psi\rangle=\langle\teuk_{(b',s')}\Phi,\Psi\rangle,
\end{equation}
where we should recall that, in order for the inner product to be a real number, we must have $b'=-b$ and $s'=s$.

In principle the {\em left hand side} of (\ref{maineq}) does not involve self-adjoint operators,
but it can be recasted using alternative fields and operators that turn out to be self-adjoint.
For the Maxwell system, it is well-known that this is achieved by considering vector potentials for the 
electromagnetic field, i.e. assuming that $F=dA$ for a 1-form $A$. See corollaries \ref{cor-maxkundt}, 
\ref{cor-maxsbh1} and \ref{cor-maxsbh2}.
For the gravitational case this is more complicated since, although it is also well-known that the linearized 
Einstein operator, 
\begin{equation}
 \dot{G}_{\a\b}[h]=-\tfrac{1}{2}\Box h_{\a\b}-\tfrac{1}{2}\c_{\a}\c_{\b}h+\c^{\g}\c_{(\a}h_{\b)\g}
 -\tfrac{1}{2}g_{\a\b}(\c^{\g}\c^{\d}h_{\g\d}-\Box h),
\end{equation}
is self-adjoint, it is not clear how to relate the identities (\ref{maineq}) with the Einstein equations. 
In the same way as we did in \cite{Araneda:2016iwr}, we will solve this problem by using the Bianchi identities 
for the curvature tensor, as we now proceed to explain.

Recall that the Bianchi identity for the Riemann tensor, $\c_{[\e}R_{\a\b]\g\d}=0$, is a geometric identity independent 
of the field equations. Therefore it allows us to relate, off shell, derivatives of the Riemann, Weyl and Ricci tensors. 
Consider now the spin $\s=2$ modified Laplace-de Rham operator, (\ref{mdrspin2}), applied to the Weyl tensor. 
The Bianchi identity implies that
\begin{equation}\label{BianchiWeyl1}
 3\c_{[\e}C_{\a\b]\g\d}=-\tfrac{2}{(d-2)}\left[2g_{[\g|[\a}\c_{\e}R_{\b]|\d]}-\tfrac{1}{(d-1)}g_{\g[\a}\c_{\e}R \; g_{\b]\d} \right],
\end{equation}
and contracting with $g^{\d\e}$,
\begin{equation}\label{BianchiWeyl2}
 \c^{\d}C_{\a\b\g\d}=\tfrac{-2}{(d-2)}\left[-(d-3)\c_{[\a}R_{\b]\g}-g_{\g[\a}\c^{\d}R_{\b]\d}+\tfrac{1}{(d-1)}g_{\g[\a}\c_{\b]}R\right].
\end{equation}
Now we apply the operator $P^{\a\b\g\d}_{i_1...i_s}\Ds^{\star}_2$ to (\ref{BianchiWeyl1}), and $P^{\a\b\g\d}_{i_1...i_s}\Ds_2$ 
to (\ref{BianchiWeyl2}).
The tensor $P^{\a\b\g\d}_{i_1...i_s}$ that we will use has the symmetries 
$P^{\a\b\g\d}_{i_1...i_s}=P^{[\a\b]\g\d}_{i_1...i_s}=P^{\a\b[\g\d]}_{i_1...i_s}=P^{\g\d\a\b}_{i_1...i_s}.$
Consider first the case in which $P^{\a\b\g\d}_{i_1...i_s}$ is trace-free, $P^{\a\b\g\d}_{i_1...i_s}g_{\b\g}=0$. 
Then it is easy to show that, off shell,
\begin{equation}\label{deRham-Einstein1}
 P^{\a\b\g\d}_{i_1...i_s}[(\Ds^{\star}_2\Ds+\Ds_2\Ds^{\star})C]_{\a\b\g\d}= 
 4P^{\a\b\g\d}_{i_1...i_s}(\c_{\d}+4\g_{\d})\c_{\a}(G_{\b\g}+\L g_{\b\g})
\end{equation}
where we have used the fact that $P^{\a\b\g\d}_{i_1...i_s}$ is trace-free to replace the Ricci tensor by the Einstein tensor, 
and the cosmological constant term can be freely added since $\c_{\a}(\L g_{\b\g})=0$. 
Linearizing this equation around an Einstein solution,
i.e. such that $(G_{\b\g}+\L g_{\b\g})|_{\ve=0}=0$, 
we get the operator $\dot{G}_{\b\g}[h]+\L h_{\b\g}$, which is self-adjoint as a linear functional on $h_{\a\b}=\dot{g}_{\a\b}$. 
Therefore, considering (\ref{deRham-Einstein1}) together with (\ref{maineq}) (for $\s=2$) and taking the adjoint equation, we will 
be able to reconstruct the metric perturbation in a very compact form.
For the case in which $P^{\a\b\g\d}_{i_1...i_s}$ is {\em not} trace-free, say
$P^{\a\d}_{i_1...i_s}:=P^{\a\b\g\d}_{i_1...i_s}g_{\b\g}\neq 0$, we can decompose it
into its trace-free part and $P^{\a\d}_{i_1...i_s}g^{\b\g}$. 
Using the fact that $P^{\a\d}_{i_1...i_s}=P^{(\a\d)}_{i_1...i_s}$, it can be shown that 
$P^{\a\d}_{i_1...i_s}g^{\b\g}[(\Ds^{\star}_2\Ds+\Ds_2\Ds^{\star})C]_{\a\b\g\d}\equiv0$; then we can use the 
same argument as before. See corollaries \ref{cor-gravkundt} and \ref{cor-gravsbh}.

\section{Extreme boost weight: Kundt spacetimes}\label{sec-kundt}

A $d$-dimensional spacetime is said to be {\em Kundt} if it admits a null geodesic congruence that is 
shear-free, expansion-free, and twist-free. From (\ref{rho}), this implies that, if $\ell^{\a}$ is the tangent to this congruence, 
it holds
\begin{equation}\label{kundt}
 \rho_{ij}=0, \hspace{1cm} \kappa_i=0.
\end{equation}
For later convenience, it is useful to note here the following Ricci identities:
\begin{align}
 \text{\th}\tau_i&=0, \label{idr1-kundt} \\
 \text{\th}\rho'_{ij}-\text{\dh}_j\t'_i &=-\t'_i\t'_j-\Phi_{ji}-\tfrac{2\L}{(d-1)(d-2)}\d_{ij}, \label{idr2-kundt} 
\end{align}
Note that the 1-form (\ref{1-form}) reduces in this case to
\begin{equation}\label{gamma-kundt}
 \g_{\a}=\tfrac{\rho'}{(d-2)}\ell_{\a}-\t_i m^{i}_{\a}.
\end{equation}

The Kundt class has attracted attention in recent years because the {\em Near Horizon Geometries} (NHGs)
of all known extreme black hole solutions belong to it, see e.g. \cite{Durkee:2010ea, Hollands:2014lra, Dias:2009ex, Tanahashi:2012si}.
Actually these spacetimes are a bit more special as they are {\em doubly-Kundt}: both null directions, 
$\ell^{\a}$ and $n^{\a}$, are
shear-free, expansion-free, and twist-free: $\rho_{ij}=\kappa_i=\rho'_{ij}=\kappa'_i=0$.
Doubly-Kundt spacetimes are algebraically special of type D or O.
NHGs turn out to be important, on the one hand, in the context of the Kerr/CFT correspondence 
\cite{Guica, Dias:2009ex}, 
and on the other hand they are also interesting in the sense that an instability of the NHG may imply 
instabilities in the full extreme black hole \cite{Durkee:2010ea, Hollands:2014lra}.

\subsection{Weighted conformal Killing fields}\label{sec-wckf-kundt}

In this section we give some interesting relations between the modified covariant derivative (\ref{defD}) that produces
the generalized wave operator $\teuk_{(b,s)}=D^{\a}D_{\a}$, and tensors associated to different kind of Killing symmetries 
with respect to $D_{\a}$, more precisely {\em weighted} conformal Killing vectors and conformal Killing-Yano tensors.

The results stated below can be demonstrated by direct calculation using the explicit form (\ref{explD}) for the modified covariant 
derivative and the identities (\ref{vd1})-(\ref{vd3}) in appendix \ref{app-vd}.
The first thing to note is that the WAND of a Kundt spacetime is a weighted Killing vector with respect to $D_{\a}$:
\begin{lem}\label{lem-wckv-kundt}
 Let $\ell^{\a}$ be a vector field aligned to the WAND of a Kundt spacetime, and consider the modified 
 covariant derivative (\ref{covder}). Then $\ell^{\a}$ is a Killing vector with respect to $D_{\a}$:
 \begin{equation}\label{lkilling}
  D_{(\a}\ell_{\b)}=0.
 \end{equation}
\end{lem}
\noindent
If the spacetime is not Kundt, it can be easily shown that $\ell^{\a}$ will be a conformal Killing vector 
with respect to $D_{\a}$ as long as $\ell^{\a}$ is geodesic and the congruence is shear-free (i.e. $\sigma_{ij}=0$ in (\ref{rho})).

Now consider the type $\{1,1\}$ 2-form
 \begin{equation}\label{U2f}
  U^{i}_{\a\b}:=2\ell_{[\a}m^i_{\b]},
 \end{equation}
where $\ell^{\a}$ is aligned to the WAND of the spacetime. Then we have:
\begin{lem}\label{lem-wcky-kundt}
 In a Kundt spacetime, the 2-form $U^{i}_{\a\b}$ is a conformal Killing-Yano tensor with respect to $D_{\a}$:
 \begin{equation}
  D_{(\a}U^{i}_{\b)\g}=g_{\g(\a}\xi^{i}_{\b)}-g_{\a\b}\xi^{i}_{\g},
 \end{equation}
 where $\xi^{i}_{\a}=\frac{1}{(d-1)}D^{\b}U^{i}_{\a\b}$.
\end{lem}

As mentioned at the end of section \ref{sec-review4d},
ordinary conformal Killing-Yano tensors are crucial in the context of the so-called hidden symmetries
in General Relativity. The reason for this is that they give rise to conserved quantities for null geodesics, as well as to operators 
that commute with the ordinary wave operator $\Box$. 
This is particularly important in the 4-dimensional Kerr solution, where the well-known Carter constant and Carter operator, which 
allow to integrate the geodesic equation and to separate the Klein-Gordon equation, have their origins in the existence of 
a Killing-Yano tensor in the Kerr spacetime. 

In four dimensions, conformal Killing-Yano tensors are also closely related to the process of {\em spin reduction} in a flat spacetime, 
introduced by Penrose in \cite{Penrose3} (see also \cite[section 6.4]{Penrose2}), 
which allow to reduce the massless higher spin field equations to the spin zero wave equation 
$\p^{\a}\p_{\a}\phi=0$ (we will very briefly review this 4-dimensional mechanism in section \ref{sec-WKS}). 

In the context of electromagnetic and gravitational perturbations of Kundt spacetimes,
in the following sections we will find wave-like equations for $U^{\a\b}_i F_{\a\b}$ and $U^{\a\b}_iU^{\g\d}_j \dot{C}_{\a\b\g\d}$ 
(where $F_{\a\b}$ is a Maxwell field and $\dot{C}_{\a\b\g\d}$ a linearized Weyl tensor)
in terms of the modified wave operator $D^{\a}D_{\a}$. These results, together with the fact that $U^i_{\a\b}$ is a conformal 
Killing-Yano tensor with respect to $D_{\a}$, imply that, as anticipated in the introduction, 
our results can be thought of as a ``weighted spin reduction'' 
generalized to the higher dimensional spacetimes studied in this work.

\subsection{Maxwell fields}\label{sec-maxkundt}

We will now prove the identity (\ref{maineq}) for spin $\s=1$ and boost weight $b=1$ on a Kundt spacetime. 
The tensor ${}^{b}P^{\a_1...\a_n}_{i_1...i_s}$ is in this case given by ${}^{1}P^{i}_{\a\b}\equiv U^{i}_{\a\b}$,
where the type $\{1,1\}$ 2-form $U^{i}_{\a\b}$ is defined in (\ref{U2f}). 

\begin{thm}\label{thm-max-kundt}
 Let $F_{\a\b}$ be an arbitrary (type $\{0,0\}$) 2-form on a $d$-dimensional Einstein-Kundt spacetime, and consider the 
 spin $\s=1$ modified Laplace-de Rham operator (\ref{mdeRham}). Then we have the equality
 \begin{equation}\label{main-max-kundt}
  -U^{\a\b}_i[(\Ds_1\Ds^{\star}+\Ds^{\star}_1\Ds)F]_{\a\b}=
  [\teuk_{(1,1)}+V]_{i}{}^{j}[U^{\a\b}_j F_{\a\b}],
 \end{equation}
 where the generalized wave operator $\teuk_{(1,1)}$ is defined in (\ref{mgbox}),
 and the potential on the right hand side is given by
 \begin{equation}
  V_{i}{}^{j}=2\Phi^S_{i}{}^{j}+4\Phi^A_{i}{}^{j}+(d-4)\t_i\t^j+\left(\Phi-\tfrac{2\L}{(d-2)}\right)\d_{i}{}^{j}, \label{Vmaxkundt}
 \end{equation}
\end{thm}

\begin{proof}

Start with the term $U^{\a\b}_i(\Ds^{\star}_1\Ds F)_{\a\b}$. From the definition of $U^{\a\b}_i$, we have
\begin{align}
\nonumber -2U^{\b\g}_i(\c_{\g}+2\g_{\g})\c^{\a}F_{\a\b}=&
 2\text{\th}[m^{\b}_i\c^{\a}F_{\a\b}]-2(\text{\dh}_i-2\tau_i-\tau'_i)[\ell^{\b}\c^{\a}F_{\a\b}]\\
 \equiv & T^1_i[\vp]+T^2_i[F_{ij}]+T^3_i[F]+T^4_i[\vp'],
\end{align}
where we have used (\ref{max-div}) to define
\begin{align*}
 T^{1}_i[\vp]=& 2\text{\th}(\text{\th}'+\rho')\vp_i-2\text{\th}(\rho'_{i}{}^{k}\vp_k)+2(\text{\dh}_i-2\t_i-\t'_i)(\text{\dh}^k-\t'^k)\vp_k \\
 T^{2}_i[F_{ij}]=& 2\text{\th}(\text{\dh}^k-\t^k-\t'^k)F_{ki}, \\
 T^{3}_i[F]=& 2\text{\th}(\t_i-\t'_i)F+2(\text{\dh}_i-2\t_i-\t'_i)\text{\th}F,  \\
 T^{4}_i[\vp']=& 2\text{\th\th}\vp'_i.
\end{align*}
For the second term, $U^{\a\b}_i(\Ds_1\Ds^{\star} F)_{\a\b}$, we have 
\begin{align}
\nonumber & 3U^{\a\b}_i(\c^{\g}+2\g^{\g})\c_{[\g}F_{\a\b]} \\
\nonumber & = 6\text{\th}[\ell^{\a}m^{\b}_i n^{\g}\c_{[\g}F_{\a\b]}]+6(\text{\dh}^k-2\tau^k-\tau'^k)[\ell^{\a}m^{\b}_i m^{\g}_k\c_{[\g}F_{\a\b]}]\\
 & \equiv  T^5_i[\vp]+T^6_i[F_{ij}]+T^7_i[F]+T^8_i[\vp'],
\end{align}
where
\begin{align*}
 T^{5}_i[\vp]=& 2\text{\th}\text{\th}'\vp_i+2\text{\th}(\rho'^{k}{}_{i}\vp_k)+2(\text{\dh}^k-2\t^k-\t'^k)[2(\text{\dh}_{[k}-\t'_{[k})\vp_{i]}] \\
 T^{6}_i[F_{ij}]=& 2\text{\th}(\t'^k-\t^k)F_{ki}-2(\text{\dh}^k-2\t^k-\t'^k)\text{\th}F_{ki}, \\
 T^{7}_i[F]=& -2\text{\th}(\text{\dh}_i-\t_i-\t'_i)F,  \\
 T^{8}_i[\vp']=& -2\text{\th\th}\vp'_i.
\end{align*}
We see that $T^{8}_i[\vp']=-T^{4}_i[\vp']$, thus these terms cancel each other.
Using the background Ricci identity (\ref{idr1-kundt}), together with the background commutator identity for $[\text{\th},\text{\dh}_i]$, 
it is easy to see that $T^{7}_i[F]+T^{3}_i[F]=0$, and similarly $T^{6}_i[F_{ij}]+T^{2}_i[F_{ij}]=0$.
The only term left is then $T^{1}_i[\vp]+T^{5}_i[\vp]$. 
To evaluate this term, we need the following commutator identities:
\begin{align*}
 [\text{\dh}_i,\text{\dh}^k]\vp_k =& 2\text{\th}(\rho'_{[ik]}\vp^k)-2(\text{\th}\rho'_{[ik]})\vp^k+2\Phi^A_{i}{}^{k}+2\Phi^S_{i}{}^{k}
   -\tfrac{2(d-3)\L}{(d-1)(d-2)}\vp_i, \\
 [\text{\th},\text{\th}']\vp_i=& (-\t^k+\t'^k)\text{\dh}_k\vp_i+\left(-\t^k\t'_k+\Phi-\tfrac{4\L}{(d-1)(d-2)}\right)\vp_i \\
  & +(\t'_i\t^k-\t_i\t'^k+2\Phi^A_{i}{}^{k})\vp_k
\end{align*}
Inserting these expressions in $T^{1}_i[\vp]+T^{5}_i[\vp]$, and using the antisymmetrization of the background Ricci 
identity (\ref{idr2-kundt}) together with its trace, we get
\begin{align}
\nonumber T^{1}_i[\vp]+T^{5}_i[\vp] = & 2\left[2\text{\th}'\text{\th}+\text{\dh}^k\text{\dh}_k+\rho'\text{\th}
 -4\t^k\text{\dh}_k+\left(\Phi-\tfrac{2(2d-1)\L}{(d-1)(d-2)}\right) \right]\vp_i\\
 &+2[-2\t_i\text{\dh}^k+2\t^k\text{\dh}_i+2\Phi^S_{i}{}^{k}+4\Phi^A_{i}{}^{k}]\vp_k.
\end{align}
Recalling now the expression (\ref{teukb1}) for the generalized wave operator with $b=1$, the definition (\ref{Vmaxkundt}) of 
the potential $V_{i}{}^{j}$, and using the fact that $\vp_i=\frac{1}{2}U^{\a\b}_i F_{\a\b}$, the result (\ref{main-max-kundt}) follows.

\end{proof}

\begin{cor}\label{cor-maxkundt}
 Let $A_{\a}$ be an arbitrary, type $\{0,0\}$ 1-form $A_{\a}$ on a $d$-dimensional Einstein Kundt spacetime.
 Then we have the following equality:
 \begin{equation}\label{max-kundt}
  \mc{S}\mc{E}[A_{\a}]=\mc{O}\mc{T}[A_{\a}],
 \end{equation}
 where the linear differential operators are defined by
 \begin{align}
 \mc{S}(J_{\b}):=& -U^{\b\g}_i(\c_{\g}+2\g_{\g})J_{\b}, \\
 \mc{E}(A_{\a}):=& \Box A_{\b}-\c^{\a}\c_{\b}A_{\a}, \\
 \mc{O}(\phi_i):=& [\teuk_{(1,1)}+V]_{i}{}^{j}\phi_j, \label{Omaxk} \\
 \mc{T}(A_{\a}):=& 2U^{\a\b}_i\c_{\a}A_{\b}.
 \end{align}
\end{cor}

\begin{proof} 
Define the 2-form $F_{\a\b}=2\c_{[\a}A_{\b]}$. Then $\Ds^{*}_1\Ds F=0$, and the identity (\ref{max-kundt}) 
follows from (\ref{main-max-kundt}).
\end{proof}

\begin{cor}
 Let $\psi_i$ be a type $\{-1,1\}$ GHP scalar field which is a solution of the $d$-dimensional adjoint
 Teukolsky equation for spin $\s=1$,
 \begin{equation}
  \teuk_{(-1,1)}\psi_i+\psi_jV^{j}{}_{i}=0,
 \end{equation}
 on an Einstein-Kundt spacetime.
 Then $F_{\a\b}(\psi)=2\c_{[\a}[A(\psi)]_{\b]}$ is a solution of Maxwell equations, where
 \begin{equation}\label{Akundt}
  A_{\a}(\psi):=(\c^{\b}-2\g^{\b})(U^i_{\a\b}\psi_i).
 \end{equation}
\end{cor}

\begin{proof}
Taking the adjoint equation to (\ref{max-kundt}), $\mc{E}^{\dag}\mc{S}^{\dag}(\psi)=\mc{T}^{\dag}\mc{O}^{\dag}(\psi)$ ,
using the fact that $\mc{E}=\mc{E}^{\dag}$, and the adjointness property (\ref{adjteuk}), 
the result follows after defining $[A(\psi)]_{\a}:=[\mc{S}^{\dag}(\psi)]_{\a}$.
\end{proof}

\subsection{Gravitational perturbations}\label{sec-gravkundt}

Using the 2-form $U^i_{\a\b}$ introduced in (\ref{U2f}), we define the type $\{2,2\}$ tensor
\begin{equation}\label{Uij}
 U^{\a\b\g\d}{}_{ij}:= U^{\a\b}_{i}U^{\g\d}_{j} + U^{\a\b}_{j}U^{\g\d}_{i}.
\end{equation}
The tensor ${}^{b}P^{\a_1...\a_n}_{i_1...i_s}$ of the main identity (\ref{maineq})
is in this case ${}^{2}P^{\a\b\g\d}_{ij}\equiv U^{\a\b\g\d}{}_{ij}$.

For this particular case, we find that off shell it is more convenient to work with the Riemann instead of the Weyl tensor 
for describing the perturbations.
Note that the boost weight $+2$ component of the Riemann tensor is 
\begin{equation}
 \wt{\Om}_{ij}=\tfrac{1}{8}U^{\a\b\g\d}{}_{ij}R_{\a\b\g\d}.
\end{equation}
The relationship with the corresponding Weyl component is
\begin{equation}
 \wt{\Om}_{ij}=\Om_{ij}+\tfrac{1}{(d-2)}\w \d_{ij},
\end{equation}
where $\w=R_{\a\b}\ell^{\a}\ell^{\b}$ (we adopt the notation in \cite{Durkee:2010xq} for the Ricci tensor 
components on the GHP frame). 
We have:
\begin{thm}\label{thm-grav-kundt}
Consider gravitational perturbations of a $d$-dimensional Einstein Kundt spacetime. 
Consider also the spin $\s=2$ modified Laplace-de Rham operator (\ref{mdeRham}).
Then we have the following off shell equality:
\begin{equation}\label{gravkundt}
 -\tfrac{1}{8}\tfrac{d}{d\ve}|_{\ve=0}\left[U^{\a\b\g\d}{}_{ij}[(\Ds_2\Ds^{\star}+\Ds^{\star}_2\Ds)R]_{\a\b\g\d}\right]
  =\teuk_{(2,2)}\dot{\wt{\Om}}_{ij}+V_{ij}{}^{kl}\dot{\wt{\Om}}_{kl}
\end{equation}
where the generalized wave operator $\teuk_{(2,2)}$ 
is defined in (\ref{mgbox}), and the potential $V_{ij}{}^{kl}$ on the right hand side is
\begin{align}
\nonumber V_{ij}{}^{kl}:=& 2[2\Phi-\tau^m\tau_m-\tfrac{4\L}{(d-1)(d-2)}]\d_{i}{}^{(k}\d_{j}{}^{l)} 
 +4[\Phi^S_{(i}{}^{(k}+4\Phi^A_{(i}{}^{(k}+\tfrac{(d-2)}{2}\t_{(i}\t^{(k}]\d^{l)}_{j)}\\
 &+\left[-2\Phi^{(k}{}_{ij}{}^{l)}-2\d_{ij}\t^k\t^l+(-\Phi^S_{ij}+\tfrac{2\L}{(d-1)(d-2)}\d_{ij}-2\t_i\t_j)\d^{kl} \right]. \label{Vgravkundt}
\end{align}
\end{thm}

\noindent
We will give a proof of this theorem in appendix \ref{proofgkundt}.
If the linearized Einstein equations $\dot{G}_{\a\b}+\L h_{\a\b}=0$ are satisfied, then 
the left hand side of (\ref{gravkundt}) is zero, and the linearized components $\dot{\wt\Om}_{ij}$ and $\dot\Om_{ij}$ coincide.
Using the expressions for the generalized wave operator $\teuk_{(2,2)}$ and the potential $V_{ij}{}^{kl}$ given above, the equation
\begin{equation}
 \teuk_{(2,2)}\dot\Om_{ij}+V_{ij}{}^{kl}\dot\Om_{kl}=0
\end{equation}
in terms of GHP derivatives is explicitly
\begin{multline}
 [2\text{\th}'\text{\th}+\text{\dh}^k\text{\dh}_k+\rho'\text{\th}-6\tau^k\text{\dh}_k+4\Phi
 -\tfrac{4(d+2)}{(d-1)(d-2)}\L]\dot\Om_{ij}\\
 +4[\tau^k\text{\dh}_{(i}-\tau_{(i}\text{\dh}^k+\Phi^{S}_{(i}{}^{k}+4\Phi^{A}_{(i}{}^{k}]\dot\Om_{j)k}
 +2\Phi_{i}{}^{k}{}_{j}{}^{l}\dot\Om_{kl}=0.
\end{multline}
The decoupled equation that theorem (\ref{thm-grav-kundt}) gives on shell is then equivalent to the equation 
found in \cite{Durkee:2010qu}.

\begin{cor}\label{cor-gravkundt}
 Consider gravitational perturbations of a $d$-dimensional Einstein Kundt spacetime, 
 and let $h_{\a\b}=\dot{g}_{\a\b}$. 
 Define the following variant of the tensor (\ref{Uij}):
\begin{equation}\label{Uijmod}
 \breve{U}^{\a\b\g\d}{}_{ij}=U^{\a\b\g\d}{}_{ij}-\tfrac{2}{(d-2)}\d_{ij}\ell^{\a}g^{\b\g}\ell^{\d}.
\end{equation}
 Then we have the off shell operator identity
 \begin{equation}\label{seot-gk}
  \mc{S}\mc{E}[h_{\a\b}]=\mc{O}\mc{T}[h_{\a\b}],
 \end{equation}
 where the linear differential operators are defined by
 \begin{align}
  \mc{S}[T_{\a\b}]&:=-\tfrac{1}{2}\breve{U}^{\a\b\g\e}{}_{ij}(\c_{\e}+4\g_{\e})\c_{\a}T_{\b\g},\\
  \mc{E}[h_{\a\b}]&:=\dot{G}_{\a\b}[h]+\L h_{\a\b},\\
  \mc{O}[\phi_{ij}]&:=\teuk_{(2,2)}\phi_{ij}+V_{ij}{}^{kl}\phi_{kl},\\
  \mc{T}[h_{\a\b}]&:=\tfrac{1}{8}\tfrac{d}{d\ve}|_{\ve=0}(U^{\a\b\g\d}{}_{ij}R_{\a\b\g\d}).
 \end{align}
\end{cor}

\begin{proof}
The proof of this result is analogous to that in section \ref{sec-mainr} (see discussion around equation 
(\ref{deRham-Einstein1})), but replacing the Weyl tensor with the Riemann tensor.
Recalling the contracted form of the Bianchi identity, we have $\c^{\d}R_{\a\b\g\d}=-2\c_{[\a}R_{\b]\g}$.
We can write the Ricci tensor in terms of the Einstein tensor in the form
$R_{\a\b}=G_{\a\b}+\frac{1}{(2-d)}g_{\a\b}G$ (with $G=g^{\g\d}G_{\g\d}$),
then the left hand side of the identity (\ref{gravkundt}) takes the form
\begin{align*}
 U^{\a\b\g\e}{}_{ij}(\c_{\e}+4\g_{\e})\c^{\d}R_{\a\b\g\d}&=
 -2U^{\a\b\g\e}{}_{ij}(\c_{\e}+4\g_{\e})\c_{\a}[G_{\b\g}+\tfrac{1}{(2-d)}g_{\b\g}G]\\
 &=-2\left[U^{\a\b\g\e}{}_{ij}-\tfrac{2}{(2-d)}\d_{ij}\ell^{\a}\ell^{\e}g^{\b\g}\right](\c_{\e}+4\g_{\e})\c_{\a}G_{\b\g}\\
 &=-2\breve{U}^{\a\b\g\e}{}_{ij}(\c_{\e}+4\g_{\e})\c_{\a}(G_{\b\g}+\L g_{\b\g}),
\end{align*}
where in the second line we used the identity $U^{\a\b\g\e}{}_{ij}g_{\b\g}=-2\d_{ij}\ell^{\a}\ell^{\e}$, 
while in the third line we used the definition of the tensor $\breve{U}^{\a\b\g\d}{}_{ij}$,
and the fact that $\c_{\a}(\L g_{\b\g})=0$ to add the cosmological constant term.
Now, linearizing, using (\ref{gravkundt}) and the background Einstein equations $(G_{\b\g}+\L g_{\b\g})|_{\ve=0}=0$, we get 
the result (\ref{seot-gk}).
\end{proof}

Note that, from the proof above, we see that the second term in the definition (\ref{Uijmod}) of $\breve{U}^{\a\b\g\d}{}_{ij}$ 
is needed in order to get the Einstein tensor from the identity (\ref{gravkundt}). 

As a corollary of the previous results, we can now construct, in a very compact form, solutions of the linearized Einstein equation 
from solutions of the $d$-dimensional Teukolsky equation:

\begin{cor}
 Let $\psi_{ij}$ be a type $\{-2,2\}$ GHP scalar field which is a solution of the $d$-dimensional adjoint 
 Teukolsky equation for spin $\s=2$,
 \begin{equation}\label{adjT1}
  \teuk_{(-2,2)}\psi_{ij}+\psi_{kl}V^{kl}{}_{ij}=0, 
 \end{equation}
 on an Einstein Kundt spacetime. Then the tensor field
 \begin{equation}\label{hkundt}
  h_{\a\b}[\psi]=\c^{\g}[(\c^{\d}-4\g^{\d})\breve{U}_{\g(\a\b)\d}{}^{ij}\psi_{ij}]
 \end{equation}
 is a solution of the linearized Einstein equations, $\dot{G}_{\a\b}[h]+\L h_{\a\b}=0$.
\end{cor}

\begin{proof}
Take the adjoint of the identity (\ref{seot-gk}), use the fact that the linearized Einstein operator is self-adjoint, 
$\mc{E}^{\dag}=\mc{E}$, and recall the adjointness property (\ref{adjteuk}). 
Defining $h_{\a\b}[\psi]=[\mc{S}^{\dag}(\psi)]_{\a\b}$, the result follows. 
\end{proof}

We note that there are other reconstruction formulae known in the literature, see \cite{Godazgar:2011sn, Hollands:2014lra}.

\section{Zero boost weight: static black holes}\label{sec-sbh}

We consider $d=(2+n)$-dimensional spacetimes with a locally warped product structure 
$\tilde\M\times_{r^2} \mathscr{K}$ for some $n$-dimensional complete manifold $\mathscr{K}$, that is
\begin{equation}\label{wp}
 g_{\a\b}(z)dz^{\a}dz^{\b}=\wt{g}_{ab}(x)dx^a dx^b+r^2(x)\wh{g}_{MN}(y)dy^M dy^N.
\end{equation}
Lowercase latin indices $a,b,c,...$ and tilded quantities refer to the orbit space $\tilde\M$, while indices
$M,N,P...$ and hatted quantities refer to $\mathscr{K}$. 
The orbit space has covariant derivative $\wt{\c}_a$, and volume form $\wt{\e}_{ab}$.
The covariant derivative on $\mathscr K$ is $\wh \c_M$, and the wave operator $\wh{\Box}=\wh{g}^{MN}\wh{\c}_M\wh{\c}_N$.
The Christoffel symbols are
\begin{align}
 & \G^{d}{}_{ab}=\wt{\G}^{d}{}_{ab}, \;\;\;\; \G^{d}{}_{aM}=0, \;\;\;\; \G^{d}{}_{MN}=-rr^{d}\wh{g}_{MN}, \label{Chr1} \\
 & \G^{M}{}_{ab}=0, \;\;\;\; \G^{M}{}_{aN}=\tfrac{r_a}{r}\d^{M}_{N}, \;\;\;\; \G^{P}{}_{MN}=\wh{\G}^{P}{}_{MN}, \label{Chr2} 
\end{align}
where $r_{a}:=\c_{a}r$.
For the orbit space $\tilde\M$, we assume the form $\wt{g}_{ab}(x)dx^a dx^b=-f(r)dt^2+\frac{1}{f(r)}dr^2$,
where the function $f$ is given by $f(r)=\kappa-\frac{2M}{r^{n-1}}-\l r^2$
for some constants $\k$, $M$ and $\l$. The full spacetime is then Einstein
\begin{equation}
 R_{\a\b}=\tfrac{2\L}{n} g_{\a\b}, \hspace{1cm} \L=\tfrac{n(n+1)}{2}\l,
\end{equation}
provided that the $n$-dimensional manifold $\mathscr{K}$ is itself an Einstein space, $\wh{R}_{MN}=\k(n-1) \wh{g}_{MN}$.
Furthermore, according to \cite{Birmingham},
in order for the spacetime to be (locally) asymptotically (anti-) de Sitter, the horizon must be a constant 
curvature space (and not just an Einstein space). The horizon topology is determined by the sign of $\k$; we 
can have elliptic ($\k=+1$), flat $(\k=0)$ or hyperbolic horizons ($\k=-1$).
For the Riemann curvature tensor, we have the following nontrivial components:
\begin{align}
 & R_{abcd}=\wt{k}(\wt{g}_{ac}\wt{g}_{bd}-\wt{g}_{ad}\wt{g}_{bc}) \label{R_abcd} \\
 & R_{aMcN}=-rr_{ac}\wh{g}_{MN}, \label{R_aBcD} \\
 & R_{MNPQ}=r^2(\kappa-r^ar_a)(\wh{g}_{MP}\wh{g}_{NQ}-\wh{g}_{MQ}\wh{g}_{NP}), \label{R_ABCD}
\end{align}
where $\wt{k}=-\frac{1}{2}\p^2_r f$ and $r_{ab}=\c_a r_b$.
The Weyl tensor can be obtained from the above expressions taking into account that the full spacetime is Einstein.
Perturbations of the spacetime (\ref{wp}) have been studied in the literature using a $2+n$ decomposition of the 
perturbations, see e.g. \cite{Kodama, Ishibashi:2003ap, Ishibashi:2011ws}.

We now define the frame $\{\ell^{\a},n^{\a},m^{\a}_i\}$ in order to implement the GHP formalism.
We choose the null vectors $\ell^{\a}$ and $n^{\a}$ to be
\begin{equation}
 \ell^{\a}\p_{\a}=\frac{1}{\sqrt{2f}}(\p_t+f\p_r), \hspace{1cm} n^{\a}\p_{\a}=\frac{1}{\sqrt{2f}}(-\p_t+f\p_r),
\end{equation}
The spatial vectors $m^{\a}_i$ are defined such that
\begin{equation}
 \wh{g}_{MN}=\frac{1}{r^2}\d_{ij}m^{i}_{M}m^{j}_{N}.
\end{equation}
The GHP spin coefficients associated to this frame are
\begin{equation}\label{spincoeff-sbh}
 \rho_{ij}=\tfrac{\rho}{n}\d_{ij}, \hspace{1cm} \rho'_{ij}=\tfrac{\rho'}{n} \d_{ij}, \hspace{1cm} \t_i=\t'_i=\k_i=\k'_i=0,
\end{equation}
where $\rho=n\ell^a r_{a}/r=\frac{n}{r}(\frac{f}{2})^{1/2}$ and $\rho'=nn^a r_{a}/r=\frac{n}{r}(\frac{f}{2})^{1/2}$.
The 1-form (\ref{1-form}) then reduces to 
\begin{equation}\label{gamma-sbh}
 \g_{\a}=\tfrac{\rho'}{n}\ell_{\a}+\tfrac{\rho}{n}n_{\a}.
\end{equation}
For the Weyl tensor components, we find 
\begin{equation}
 \Om_{ij}=\Om'_{ij}=\Psi_{ijk}=\Psi'_{ijk}=0,
\end{equation}
therefore the spacetime is of algebraic type D; the only nontrivial components  are
\begin{align}
 & \Phi=-\frac{n(n-1)M}{r^{n+1}} ,\\
 & \Phi_{ij}=\frac{\Phi}{n}\d_{ij} ,\\
 & \Phi_{ijkl}=-\frac{2\Phi}{n(n-1)}(\d_{ik}\d_{jl}-\d_{il}\d_{jk}).
\end{align}
(Note that for $n=2$, i.e. in four dimensions, we have $\Phi=2\Psi_2$.) The Bianchi identities reduce to
\begin{equation}
 \text{\th}\Phi=-\tfrac{(n+1)}{n}\Phi\rho, \hspace{1cm} \text{\dh}_k\Phi=0.
\end{equation}
This implies that the 1-form (\ref{gamma-sbh}) can be rewritten as $\g_{\a}=\Phi^{-q}\c_{\a}\Phi^q$,
where $q=-1/(n+1)$ (note that this reduces to (\ref{gamma4d}) in the 4-dimensional case).
For the calculations that follow, it is useful to note here the following background identities:
\begin{align}
 \text{\th}\Phi^{mq}=&m\tfrac{\rho}{n}\Phi^{mq}, \hspace{1cm} m\in\mbb{R}, \label{idb-bsbh} \\
 \text{\th}'\rho=&-\tfrac{\rho\rho'}{n}-\Phi-\tfrac{2\L}{n+1}. \label{idr-bsbh}
\end{align}

In the following sections we will study equations for boost weight zero quantities. 
Note that in this case, in view of (\ref{spincoeff-sbh}), the modified wave operator (\ref{mgbox}) reduces to
\begin{equation}\label{teuk0s}
 \teuk_{(0,s)}=\Box_{(0,s)}.
\end{equation}

\subsection{Weighted conformal Killing fields}\label{sec-gky}

Similarly as was done in section \ref{sec-wckf-kundt} for the case of Kundt spacetimes, 
in this section we want to show, for the case of the spacetime (\ref{wp}), the relation between 
the modified covariant derivative $D_{\a}$ that produces the generalized wave operator $\teuk_{(b,s)}=D^{\a}D_{\a}$, and 
tensor fields that satisfy a conformal Killing equation with respect to $D_{\a}$.

The first result is that the WANDs of the spacetime (\ref{wp}) turn out to be conformal Killing vectors with 
respect to $D_{\a}$, analogously as in the Kundt case:

\begin{lem}\label{lem-wckv-sbh}
 Let $\ell^{\a}$ and $n^{\a}$ be vector fields aligned to the WANDs of the spacetime (\ref{wp}), and consider the modified 
 covariant derivative (\ref{covder}). Then $k^{0\a}:=\ell^{\a}$ and $k^{1\a}:=\Phi^{2q}n^{\a}$, with $q=-1/(n+1)$, are conformal 
 Killing vectors with respect to $D_{\a}$,
 \begin{equation}
   D_{(\a}k^I_{\b)}=\l^I g_{\a\b}, \hspace{1cm} I=0,1 \\
 \end{equation}
 where $\l^0=\frac{\rho}{n}$ and $\l^1=\tfrac{\rho'}{n}\Phi^{2q}$.
\end{lem}

Now we turn to weighted conformal Killing-Yano tensors.

\begin{lem}\label{lem-wcky-sbh}
 Consider the following type $\{0,0\}$ and type $\{0,2\}$ 2-forms on the spacetime (\ref{wp}):
 \begin{align}
  Z_{\a\b} :=& 2\Phi^{q}\ell_{[\a}n_{\b]}, \label{cky} \\
  Y_{\a\b}{}^{ij} :=& 2\Phi^{q}m^{i}_{[\a}m^{j}_{\b]}, \label{gky-ss}
 \end{align}
 where $q=-1/(n+1)$. Then:
 \begin{enumerate}
  \item $Z_{\a\b}$ is an ordinary conformal Killing-Yano tensor,
  \begin{equation}\label{ocky}
   \c_{(\g}Z_{\a)\b}=g_{\a\g}\xi_{\b}-g_{\b(\g}\xi_{\a)}
  \end{equation}
  where the divergence $\xi_{\a}:=\frac{1}{n+1}\c^{\b}Z_{\b\a}$ is a Killing vector. 
  (Of course, $\c_{\a}$ can be replaced by $D_{\a}$ in (\ref{ocky}), since $Z_{\a\b}$ is type $\{0,0\}$.)
  \item $Y_{\a\b}{}^{ij}$ is a Killing-Yano tensor with respect to $D_{\a}$:
  \begin{equation}
   D_{(\a}Y_{\b)\g}{}^{ij}=0.
  \end{equation}
 \end{enumerate}
\end{lem}

The same remark at the end of section \ref{sec-wckf-kundt} applies here. 
Namely, in the following sections we will find wave-like equations for $Z^{\a\b}F_{\a\b}$, $Y^{\a\b}{}_{ij}F_{\a\b}$
and $Z^{\a\b}Y^{\g\d}{}_{ij} \dot{C}_{\a\b\g\d}$ in terms of the modified wave operator $D^{\a}D_{\a}$;
and, once more, the fact that $Z_{\a\b}$ and $Y_{\a\b}{}^{ij}$ are conformal 
Killing-Yano tensors with respect to $D_{\a}$ implies that our results can be thought of as a {\em weighted spin reduction}
generalized to the higher dimensional spacetimes studied in the present work.

\subsection{Maxwell fields}\label{sec-max-sbh}

In this section we prove the spin $\s=1$ and boost weight zero case of the identity (\ref{maineq}), 
for the case of the higher dimensional static black holes given by (\ref{wp}).

Consider the weighted conformal Killing-Yano tensors (\ref{cky}) and (\ref{gky-ss}).
The possible non-trivial contractions with a 2-form $F_{\a\b}$ are
\begin{align}
 Z^{\a\b}F_{\a\b}&=2\Phi^q F, \label{ZF} \\
 Y^{\a\b}{}_{ij}F_{\a\b}&=2\Phi^q F_{ij} \label{YF}
\end{align}
(see table \ref{tableGHPmaxwell} in appendix \ref{app-maxwell} for the definition of the Maxwell scalars). 
(\ref{ZF}) and (\ref{YF}) are the two (rescaled) boost weight zero components of a Maxwell field.

We will now prove (\ref{maineq}) for spin $\s=1$, boost weight $b=0$, and spin weight $s=0$, 
in which case the tensor ${}^{b}P^{\a_1...\a_n}_{i_1...i_s}$ is ${}^{0}P_{\a\b}\equiv Z_{\a\b}$. We have
\begin{thm}\label{thm-cky-max}
 Let $F_{\a\b}$ be an arbitrary 2-form on the $(n+2)$-dimensional spacetime (\ref{wp}), and consider the 
 spin $\s=1$ modified Laplace-de Rham operator (\ref{mdeRham}) and the conformal Killing-Yano tensor (\ref{cky}). 
 Then we have the equality
 \begin{equation}\label{main-max-cky}
  -Z^{\a\b}[(\Ds_1\Ds^{\star}+\Ds^{\star}_1\Ds)F]_{\a\b}=(\Box_{(0,0)}+V^{+})[Z^{\a\b}F_{\a\b}]
 \end{equation}
 where $\Box_{(0,0)}=\Box$ and the potential on the right hand side is given by
 \begin{equation}\label{V-max-cky}
  V^{+}=-\tfrac{2(n-1)}{n}\left[\Phi+\tfrac{2\L}{n+1}\right].
 \end{equation}
\end{thm}

\noindent
Note that for $n=2$ (i.e. $d=4$), the potential (\ref{V-max-cky}) reduces to
\begin{equation}
 V^{+}=-2\Psi_2-\tfrac{2}{3}\L, \hspace{1cm} \text{for } d=4
\end{equation} 
which is in agreement with the 4-dimensional result (see eq. (\ref{maxwell2})).

\begin{proof}[Proof of theorem \ref{thm-cky-max}]

Consider first the term $Z^{\a\b}[\Ds_1\Ds^{\star}F]_{\a\b}$. It is not difficult to show that
\begin{align}
\nonumber -2Z^{\b\g}(\c_{\g}+2\g_{\g})\c^{\a}F_{\a\b}=& -2(\text{\th}'+\tfrac{\rho'}{n})[\Phi^q\ell^{\b}\c^{\a}F_{\a\b}] 
  +2(\text{\th}+\tfrac{\rho}{n})[\Phi^q n^{\b}\c^{\a}F_{\a\b}] \\
  = & T^1[F]+T^2[\vp,\vp'],
\end{align}
where in the second line we have used (\ref{max-div}) and introduced
\begin{align}
 T^1[F]\equiv &2(\text{\th}'+\tfrac{\rho'}{n})[\Phi^q(\text{\th}+\rho)F]+2(\text{\th}+\tfrac{\rho}{n})[\Phi^q(\text{\th}'+\rho')F]\\
 T^2[\vp,\vp']\equiv &2(\text{\th}'+\tfrac{\rho'}{n})[\Phi^q\text{\dh}^k\vp_k]-2(\text{\th}+\tfrac{\rho}{n})[\Phi^q\text{\dh}^k\vp'_k].
\end{align}
The other term, $Z^{\a\b}[\Ds^{\star}_1\Ds F]_{\a\b}$, is easily shown to be
\begin{align}
\nonumber 3Z^{\a\b}(\c^{\g}+2\g^{\g})\c_{[\g}F_{\a\b]}=& 6\text{\dh}^{i}[\Phi^q \ell^{\a}n^{\b}m^{\g}_i\c_{[\g}F_{\a\b]}] \\
 \equiv & T^3[F]+T^4[\vp,\vp'],
\end{align}
where 
\begin{align}
 T^3[F]=& 2\text{\dh}^k\text{\dh}_k (\Phi^q F), \\
 T^4[\vp,\vp']=& 2\text{\dh}^k[\Phi^q(\text{\th}+\tfrac{\rho}{n})\vp'_k]-2\text{\dh}^k[\Phi^q(\text{\th}'+\tfrac{\rho'}{n})\vp_k].
\end{align}
Using the background Bianchi identity (\ref{idb-bsbh}) with $m=1$, together with the background commutator identities 
for $[\text{\th},\text{\dh}_i]$, one immediately sees that $T^4[\vp,\vp']=-T^2[\vp,\vp']$, therefore these terms cancel each other.
On the other hand, using again (\ref{idb-bsbh}), and the background Ricci identity (\ref{idr-bsbh}), we see that
\begin{equation}
 T^1[F]=2(\text{\th}'+\rho')\text{\th}[\Phi^q F]+2(\text{\th}+\rho)\text{\th}'[\Phi^q F]
    -\tfrac{4(n-1)}{n}\left(\Phi+\tfrac{2\L}{n+1}\right)\Phi^q F.
\end{equation}
We then recognize that $T^1[F]+T^3[F]=2(\Box+V^{+})[\Phi^q F]$, where $V^{+}$ is defined in (\ref{V-max-cky}). 
Noting that $\Phi^q F=\frac{1}{2}Z^{\a\b}F_{\a\b}$, the results (\ref{main-max-cky}) follows.

\end{proof}

\begin{cor}\label{cor-maxsbh1}
 Let $A_{\a}$ be an arbitrary 1-form on the $(n+2)$-dimensional spacetime (\ref{wp}).
 We have the operator equality
 \begin{equation}\label{seot-maxsbh1}
  \mc{S}\mc{E}[A_{\a}]=\mc{O}\mc{T}[A_{\a}],
 \end{equation}
 where the linear differential operators are defined by
 \begin{align}
 \mc{S}(J_{\b}):=& Z^{\b\g}(\c_{\g}+2\g_{\g})J_{\b}, \\
 \mc{E}(A_{\a}):=& \Box A_{\b}-\c^{\a}\c_{\b}A_{\a}, \\
 \mc{O}(\phi):=& (\Box_{(0,0)}+V^{+})\phi, \label{Omaxsbh1} \\
 \mc{T}(A_{\a}):=& 2Z^{\a\b}\c_{\a}A_{\b}.
 \end{align}
\end{cor}

\begin{proof} 
Define the 2-form $F_{\a\b}=2\c_{[\a}A_{\b]}$. Then $\Ds^{\star}_1\Ds F=0$, and the identity (\ref{seot-maxsbh1}) 
follows from (\ref{main-max-cky}).
\end{proof}

\begin{cor}\label{cor-max-cky}
 Let $\psi$ be a type $\{0,0\}$ GHP scalar field which is a solution of the wave-like equation 
 \begin{equation}
 (\Box_{(0,0)}+V^{+})\psi=0
 \end{equation}
 on the spacetime (\ref{wp}). 
 Then $F_{\a\b}(\psi)=2\c_{[\a}[A(\psi)]_{\b]}$ is a solution of Maxwell equations, where 
 \begin{equation}\label{VPcky}
 [A(\psi)]_{\a}=(\c^{\b}-2\g^{\b})(Z_{\a\b}\psi). 
 \end{equation}
\end{cor}

Consider now the identity (\ref{maineq}) for spin $\s=1$, boost weight $b=0$, and spin weight $s=2$. 
The tensor ${}^{b}P^{\a_1...\a_n}_{i_1...i_s}$ is in this case ${}^{0}P^{ij}_{\a\b}\equiv Y_{\a\b}{}^{ij}$. 
We have:

\begin{thm}\label{thm-gky-max}
 Let $F_{\a\b}$ be an arbitrary 2-form on the $(n+2)$-dimensional spacetime (\ref{wp}), and consider the 
 spin $\s=1$ modified Laplace-de Rham operator (\ref{mdeRham}) and the weighted Killing-Yano tensor (\ref{gky-ss}). 
 Then we have the equality
 \begin{equation}
  -Y^{\a\b}{}_{ij}[(\Ds_1\Ds^{\star}+\Ds^{\star}_1\Ds)F]_{\a\b}=(\Box_{(0,2)}+V^{-})[Y^{\a\b}{}_{ij}F_{\a\b}]
 \end{equation}
 where $\Box_{(0,2)}$ is defined in (\ref{gbox}) and the potential on the right hand side is given by
 \begin{equation}\label{V-max-gky}
  V^{-}=\tfrac{2(n-3)}{n(n-1)}\Phi-\tfrac{4(n-1)}{n(n+1)}\L-\tfrac{6(n-2)}{n^2}\rho\rho'.
 \end{equation}
\end{thm}
\noindent
Similarly as in the previous case, 
note that for $n=2$ ($d=4$), $\Box_{(0,2)}$ equals $\Box$ and the potential (\ref{V-max-gky}) reduces to
\begin{equation}
 V^{-}=-2\Psi_2-\tfrac{2}{3}\L, \hspace{1cm} \text{for } d=4
\end{equation}
which reproduces the 4-dimensional result.

\begin{proof}[Proof of theorem \ref{thm-gky-max}]
The proof goes along very similar lines as the one of theorem \ref{thm-cky-max}, and it can be done 
following the same ideas (although the calculations are a bit more tedious).
\end{proof}

\begin{cor}\label{cor-maxsbh2}
 Let $A_{\a}$ be an arbitrary 1-form on the $(n+2)$-dimensional spacetime (\ref{wp}).
 We have the operator equality
 \begin{equation}\label{seot-maxsbh2}
  \mc{S}\mc{E}[A_{\a}]=\mc{O}\mc{T}[A_{\a}],
 \end{equation}
 where the linear differential operators are defined by
 \begin{align}
 \mc{S}(J_{\b}):=& Y^{\b\g}{}_{ij}(\c_{\g}+2\g_{\g})J_{\b}, \\
 \mc{E}(A_{\a}):=& \Box A_{\b}-\c^{\a}\c_{\b}A_{\a}, \\
 \mc{O}(\phi):=& (\Box_{(0,2)}+V^{-})\phi_{ij}, \label{Omaxsbh2} \\
 \mc{T}(A_{\a}):=& 2Y^{\a\b}{}_{ij}\c_{\a}A_{\b}.
 \end{align}
\end{cor}

\begin{cor}
 Let $\psi_{ij}$ be a type $\{0,2\}$ GHP scalar field which is a solution of the wave-like equation 
 \begin{equation}
 (\Box_{(0,2)}+V^{-})\psi_{ij}=0
 \end{equation}
 on the spacetime (\ref{wp}).
 Then $F_{\a\b}(\psi)=2\c_{[\a}[A(\psi)]_{\b]}$ is a solution of Maxwell equations, where 
 \begin{equation}\label{VPgky}
 [A(\psi)]_{\a}=(\c^{\b}-2\g^{\b})(Y_{\a\b}{}^{ij}\psi_{ij}).
 \end{equation}
\end{cor}

\subsection{Gravitational perturbations}\label{sec-grav-sbh}

An arbitrary symmetric tensor field $h_{\a\b}=h_{(\a\b)}$ on the spacetime (\ref{wp}) can be decomposed as 
\begin{equation}
 h_{\a\b}dz^{\a}dz^{\b}=h_{ab}dx^a dx^b+2h_{aM}dx^a dy^M +h_{MN}dy^M dy^N.
\end{equation}
In the $n+2$ decomposition, each piece on the right hand side above is further decomposed into scalar, vector and tensor parts 
with respect to $\mathscr{K}$ (see \cite{Ishibashi:2011ws}):
\begin{align}
 h_{ab}&=h^{\text{\sc s}}_{ab},\\
 h_{aM}&=\wh{\c}_M h^{\text{\sc s}}_a+h^{\text{\sc v}}_{aM}, \\
 h_{MN}&=h^{\text{\sc t}}_{MN}+2\wh{\c}_{(M}h^{\text{\sc v}}_{N)}+\wh{L}_{MN}(h^{\text{\sc s}}_{\bot})+\wh{g}_{MN}h^{\text{\sc s}}_{\|},
\end{align}
where $\wh{L}_{MN}=\wh{\c}_M\wh{\c}_N-\frac{1}{n}\wh{g}_{MN}\wh{\Box}$. 
The tensor part of $h_{\a\b}$ is $h^{\text{\sc t}}_{MN}$, its vector part is $\{h^{\text{\sc v}}_{aM}, h^{\text{\sc v}}_{M}\}$, and its
scalar part is $\{h^{\text{\sc s}}_{ab}, h^{\text{\sc s}}_{a}, h^{\text{\sc s}}_{\bot}, h^{\text{\sc s}}_{\|}\}$.

Consider the weighted conformal Killing-Yano tensors (\ref{gky-ss}) and (\ref{cky}). 
The possible contractions with the curvature tensor are
\begin{align}
Z^{\a\b}Z^{\g\d}C_{\a\b\g\d} =& 4\Phi^{2q+1} ,\\
Z^{\a\b}Y^{\g\d}{}_{ij}C_{\a\b\g\d} =& 8 \Phi^{2q} \Phi^A_{ij}, \label{ZYC} \\
Y^{\a\b}{}_{ij}Y^{\g\d}{}_{kl}C_{\a\b\g\d} =& 4 \Phi^{2q}\Phi_{ijkl}.
\end{align}
We see that the only gauge invariant quantity is (\ref{ZYC}) (since in the background spacetime (\ref{wp}) we have $\Phi^A_{ij}=0$), 
therefore we will focus on this variable. 
Using the frame $\{\ell^{\a},n^{\a},m^{\a}_i\}$ defined above, it is not difficult to show that an off shell 
expression for $\dot\Phi^A_{ij}$ is
\begin{equation}\label{Phi^A_ij}
 \dot\Phi^A_{ij}=\tfrac{1}{2}\dot{R}_{abMN}\ell^a n^b m^{M}_i m^{N}_j.
\end{equation}
(In deriving this equation, one uses background identities to see that it is possible to replace the Weyl tensor 
by the Riemann tensor, and that all terms containing perturbed frame vectors vanish.) 
In terms of an arbitrary metric perturbation, it is tedious but straightforward to show that 
\begin{equation}\label{R_abCD}
 \dot{R}_{abMN}=-2r^2\p_{[M|}\p_{[a}\left[\frac{1}{r^2}h_{b]|N]}\right].
\end{equation}
From this structure one easily sees that $\dot\Phi^A_{ij}$ is non-trivial {\em only} for vector perturbations 
(see also \cite{Godazgar:2011sn}). 
This is in analogy to the 4-dimensional case, since $\Phi^A_{ij}$ is the higher dimensional analog of $\im\Psi_2$ and 
vector perturbations correspond to the so-called odd sector in $d=4$.

We define the type $\{0,2\}$ tensor
\begin{equation}
 W^{\a\b\g\d}{}_{ij}:=Y^{\a\b}{}_{ij}Z^{\g\d}+Z^{\a\b}Y^{\g\d}{}_{ij}.
\end{equation}
This tensor has the symmetries $W^{\a\b\g\d}{}_{ij}=W^{[\a\b]\g\d}{}_{ij}=W^{\a\b[\g\d]}{}_{ij}=W^{\g\d\a\b}{}_{ij}$, as well 
as being trace-free. 
The spin $\s=2$, boost weight $b=0$ and spin weight $s=2$ version of the identity (\ref{maineq}), for 
perturbations of the spacetime (\ref{wp}), is the following:
\begin{thm}\label{thm-grav-sbh}
Consider gravitational perturbations of the $(n+2)$-dimensional spacetime (\ref{wp}). 
Consider also the spin $\s=2$ modified Laplace-de Rham operator (\ref{mdeRham}). Then we have the off shell equality
\begin{equation}\label{grav-sbh}
 -\tfrac{1}{16}\tfrac{d}{d\ve}|_{\ve=0}\left\{W^{\a\b\g\d}{}_{ij}[(\Ds^{\star}_2\Ds+\Ds_2\Ds^{\star})C]_{\a\b\g\d}\right\}
 = (\Box_{(0,2)}+V)[\Phi^{2q}\dot\Phi^{A}_{ij}]
\end{equation}
where the potential on the right hand side is
\begin{equation}\label{Vgsbh}
 V=-\tfrac{2(n+2)}{n}\Phi-\tfrac{4}{n(n+1)}\L+\tfrac{2(n-2)}{n^2}\rho\rho'
\end{equation}
and the generalized wave operator $\Box_{(0,2)}$ (or equivalently $\teuk_{(0,2)}$, in view of (\ref{teuk0s})) 
was defined in (\ref{gbox}).
\end{thm}

\noindent
We give a proof of this result in appendix \ref{proofgsbh}.
Note that in the four dimensional case ($n=2$), 
the generalized wave operator $\Box_{(0,2)}$ reduces to the standard D'Alembertian $\Box$, and $\Phi=2\Psi_2$, so the 
potential is
\begin{equation}
 V=-8\Psi_2-\tfrac{2}{3}\L \hspace{1cm} \text{for } d=4
\end{equation}
which gives the correct 4-dimensional limit (see eq. (\ref{grav2})).
We also note that, assuming that the linearized Einstein equations are satisfied (which implies that the left hand side of (\ref{grav-sbh}) 
is zero), a similar equation for $\dot\Phi^A_{ij}$ was found in \cite{Godazgar:2011sn}.

\begin{cor}\label{cor-gravsbh}
 Consider gravitational perturbations of the $(n+2)$-dimensional spacetime (\ref{wp}), and let $h_{\a\b}=\dot{g}_{\a\b}$.
 Then we have the off shell operator identity
 \begin{equation}\label{seot-gsb}
  \mc{S}\mc{E}[h_{\a\b}]=\mc{O}\mc{T}[h_{\a\b}],
 \end{equation}
 where the linear differential operators are defined by
 \begin{align}
  \mc{S}[T_{\a\b}]&:=W^{\a\b\g\d}{}_{ij}(\c_{\d}+4\g_{\d})\c_{\a}T_{\b\g},\\
  \mc{E}[h_{\a\b}]&:=\dot{G}_{\a\b}[h]+\L h_{\a\b},\\
  \mc{O}[\phi_{ij}]&:=-4(\Box_{(0,2)}+V)\phi_{ij},\\
  \mc{T}[h_{\a\b}]&:=\tfrac{1}{8}\tfrac{d}{d\ve}|_{\ve=0}(W^{\a\b\g\d}{}_{ij}C_{\a\b\g\d}).
 \end{align}
\end{cor}

\begin{proof}
Since the tensor $W^{\a\b\g\d}{}_{ij}$ is trace-free,
the result follows from equation (\ref{grav-sbh}) and our general identity (\ref{deRham-Einstein1}).
\end{proof}

\begin{cor}
 Let $\psi_{ij}$ be a type $\{0,2\}$ GHP scalar field which is a solution of the 
 wave-like equation
 \begin{equation}
  (\Box_{(0,2)}+V)\psi_{ij}=0
 \end{equation}
 on the spacetime (\ref{wp}). Then the tensor field
 \begin{equation}\label{hsbh}
  h_{\a\b}[\psi]=\c^{\g}[(\c^{\d}-4\g^{\d})W_{\g(\a\b)\d}{}^{ij}\psi_{ij}]
 \end{equation}
 is a solution of the linearized Einstein equations, $\dot{G}_{\a\b}[h]+\L h_{\a\b}=0$.
\end{cor}

It is worth studying the $n+2$ decomposition of (\ref{hsbh}).
A lengthy and tedious calculation using the explicit form of $W_{\a\b\g\d ij}$, and the Christoffel symbols 
(\ref{Chr1}), (\ref{Chr2}) of the background, gives
\begin{align}
 h_{ab}[\psi]&=0, \\
 h_{aM}[\psi]&=2\wt{\e}_{a}{}^{b}\wt{\c}_b\wh{\c}_N[\Phi^{2q}m^i_M m^{jN}\psi_{ij}], \\
 h_{MN}[\psi]&=0.
\end{align}
Consider for example the case in which $\psi_{ij}$ comes from a metric perturbation $\mr{h}_{\a\b}$,
$\mr\psi_{ij}:=\Phi^{2q}\dot\Phi^A_{ij}[\mr{h}]$.  
As we have seen, $\dot\Phi^A_{ij}$ is non-trivial only if $\mr{h}_{\a\b}$ is a vector-type perturbation.
In the Ishibashi and Kodama approach \cite{Ishibashi:2003ap, Ishibashi:2011ws, Kodama},
a vector-type metric perturbation is written as $\mr{h}^{\text{\sc v}}_{aM}=h_a(x)\wh{V}_M(y)$, where 
$\wh{V}_M$ are divergence-free vector harmonics on the Einstein space $\mathscr{K}$:
\begin{equation}
 (\wh{\Box}+k^2_v)\wh{V}_M=0, \hspace{1cm} \wh{\c}^M\wh{V}_M=0,
\end{equation}
with the eigenvalues $k^2_v$ being non-negative. 
Using (\ref{Phi^A_ij}) and (\ref{R_abCD}), it is not difficult to show then that
\begin{equation}\label{vh}
 \dot{\Phi}^A_{ij}[\mr{h}]=-\tfrac{1}{2}\mc{F}m^M_i m^N_j\wh{\c}_{[M}\wh{V}_{N]}, 
\end{equation}
where 
\begin{equation}
 \mc{F}:=r^2\wt{\e}^{ab}\wt{\c}_a(r^{-2}h_b).
\end{equation}
Defining $\mr\psi:=-\frac{1}{2}c\Phi^{2q}\mc{F}=-\frac{1}{2}c^2 r^2\mc{F}$, with $c=[-Mn(n-1)]^{-2/(n+1)}$,  
and using (\ref{vh}) and $\wh{R}_{MN}=(n-1)\k\wh{g}_{MN}$, we arrive at
\begin{equation}
 h_{aM}[\mr\psi]=m_{\text{\sc v}} \wt{\e}_{a}{}^{b}\wt{\c}_b\mr\psi \; \wh{V}_M,
\end{equation}
where $m_{\text{\sc v}}:=k^2_v-(n-1)\k$. Therefore we conclude that, as long as $m_{\text{\sc v}}\neq0$, 
metric perturbations generated in the form (\ref{hsbh}) by $\dot{\Phi}^A_{ij}$ are of vector-type.

\section{Remarks on the 4-dimensional case}\label{sec-moreon4d}

In this section we apply our results to the particular case of 4-dimensional Einstein spacetimes of Petrov type D. 
As is well-known, the type D class includes all the stationary, asymptotically flat vacuum 
black hole solutions, as well as their cosmological counterparts.
In what follows we will use 2-spinor notation, together with the definitions of modified derivatives, wave operators, 
tensors, etc. that were given in section \ref{sec-review4d}. The metric signature for this section is $(+---)$.

\subsection{Off shell identities}\label{sec-offshell4d}

We would like to point out that the spin $\s$ modified de Rham operator (\ref{mdeRham}) introduced in section \ref{sec-derham}, 
is {\em the same} (using the appropriate 1-form (\ref{gamma4d})) that appears implicitly in the 4-dimensional 
case in the formulae in section \ref{sec-review4d}, for the more general class of Einstein type D spacetimes. 
This can be directly seen from our general identities (\ref{maxwell1})-(\ref{grav3}) for the Maxwell and
gravitational systems.
Consider for instance the gravitational case. If $P_{\a\b\g\d}$ is any of the anti-self-dual tensors (\ref{W0})-(\ref{W4}),
using (for example) spinor techniques it is straightforward to show that 
$3P^{\a\b\g\d}(\c^{\e}+4\g^{\e})\c_{[\e}C_{\a\b]\g\d}=2P^{\a\b\g\e}(\c_{\e}+4\g_{\e})\c^{\d}C_{\a\b\g\d}$.
Using the anti-self-dual property of $P_{\a\b\g\d}$, that is $\frac{1}{2}\e_{\a\b}{}^{\e\l}P_{\e\l\g\d}=-iP_{\a\b\g\d}$, 
it is also easy to see that we can replace the anti-self-dual Weyl tensor in (\ref{grav1})-(\ref{grav3}), $\wt{C}_{\a\b\g\d}$, 
by the ordinary Weyl tensor, $C_{\a\b\g\d}$. Therefore
\begin{equation}
 P^{\a\b\g\e}(\c_{\e}+4\g_{\e})\c^{\d}\wt{C}_{\a\b\g\d}=-\tfrac{1}{4}P^{\a\b\g\d}[(\Ds_2\Ds^{\star}+\Ds^{\star}_2\Ds)C]_{\a\b\g\d},
\end{equation}
and from (\ref{grav1})-(\ref{grav3}) we conclude that, just as we have shown for the higher dimensional spacetimes 
studied in the present work, in the 4-dimensional case the projections of the spin $\s=2$ modified de Rham operator
\begin{equation}
 [(\Ds_2\Ds^{\star}+\Ds^{\star}_2\Ds)C]_{\a\b\g\d}
\end{equation}
give, {\em off shell}, the wave-like scalar equations for the gravitational system. 
Similarly, it is not difficult to show that
\begin{equation}
 P^{\b\g}(\c_{\g}+2\g_{\g})\c^{\a}\wt{F}_{\a\b}=\tfrac{1}{4}P^{\a\b}[(\Ds_1\Ds^{\star}+\Ds^{\star}_1\Ds)F]_{\a\b},
\end{equation}
(where $P_{\a\b}$ is any of the 2-forms (\ref{Mforms}))
and then from formulae (\ref{maxwell1})-(\ref{maxwell3}) we deduce that the projections of the spin $\s=1$ 
modified de Rham operator 
\begin{equation}
 [(\Ds_1\Ds^{\star}+\Ds^{\star}_1\Ds)F]_{\a\b}
\end{equation}
give, off shell, the decoupled equations for the Maxwell system in $d=4$.

It is worth noting also that, in the 4-dimensional case, operators of the kind of (\ref{D_s}), (\ref{Dstar_s}) 
are currently being considered using $SL(2,\mbb{C})$-spinor techniques (see \cite{Aksteiner1, Aksteiner2}).

\subsection{Weighted Killing spinors}\label{sec-WKS}

Now we will give some relations between the covariant derivative in the generalized, 4-dimensional wave operator 
(\ref{opteuk4d}), and {\em weighted} Killing symmetries. 
These relations will justify the name of {\em weighted spin reduction} that we have been using.

A {\em valence $(m,n)$ Killing spinor} is a totally symmetric spinor field 
$L^{A'_1...A'_n}_{A_1...A_m}=L^{(A'_1...A'_n)}_{(A_1...A_m)}$ satisfying the twistor equation
\begin{equation}
 \c^{(B'}_{(B}L^{A'_1....A'_n)}_{A_1...A_m)}=0.
\end{equation}
Valence $(1,0)$ Killing spinors are known as {\em twistors} \cite{Penrose2}. 
For higher valence, the corresponding Killing spinors are related to 
different kind of Killing symmetries on the spacetime; for example, for valence $(1,1)$ they turn out to be conformal Killing vectors
(more generally, for valence $(n,n)$ they are rank $n$ conformal Killing tensors), and for valence $(2,0)$ they
are associated to conformal Killing-Yano 2-forms.

In addition to the close relation that Killing spinors have with different kind of Killing symmetries, these objects are also 
useful for reducing the higher spin field equations to the massless wave equation in Minkowski spacetime, a mechanism 
introduced by Penrose in \cite{Penrose3} (see also \cite[section 6.4]{Penrose2}). 
This ``spin reduction process'' implies that, if $\phi_{A_1...A_{n}}$ satisfies the massless free field equation
$\p^{A'_1A_1}\phi_{A_1...A_n}=0$, and $L_{A_1...A_n}$ is a valence 
$(n,0)$ Killing spinor, then $\Box(L^{A_1...A_n}\phi_{A_1...A_n})=0$ on a flat spacetime.

Using the spinor version $D_{A'A}$ of the modified covariant derivative given in (\ref{mcd4d}), 
we can define valence $(m,n)$ ``weighted Killing spinors'' with respect to $D_{A'A}$, 
as totally symmetric solutions of
\begin{equation}
 D^{(B'}_{(B}L^{A'_1....A'_n)}_{A_1...A_m)}=0.
\end{equation}

\begin{lem}\label{lem-WT}
Consider a 4-dimensional Einstein spacetime of Petrov type D. Let $o_{A}$ and $\iota_{A}$ be spinor fields 
aligned to the principal null directions.
Then $o_A$ and $\Psi^{-1/3}_2\iota_A$ are valence $(1,0)$ weighted Killing spinors: 
\begin{align}
 & D_{B'(B}o_{A)}=0, \\
 & D_{B'(B}[\Psi^{-1/3}_2\iota_{A)}]=0.
\end{align}
\end{lem}
It is easy to check that the product of weighted Killing spinors of valence $(r,0)$ and $(r',0)$ is again a weighted Killing spinor, 
of valence $(r+r',0)$. For example, denoting $o_{A_1...A_n}=o_{A_1}...o_{A_n}$, 
$\iota_{A_1...A_m}=\iota_{A_1}...\iota_{A_m}$, we have
\begin{align}
 & D_{B'(B}o_{A_1...A_n)}=0 \label{wks1} \\
 & D_{B'(B}[\Psi^{-m/3}_2\iota_{A_1...A_m)}]=0, \label{wks2} \\
 & D_{B'(B}[\Psi^{-m/3}_2\iota_{A_1...A_m}o_{A_{m+1}...A_{m+n})}]=0. \label{wks3}
\end{align}
In particular, setting $m=1$, $n=1$, the spinor $K_{AB}:=\Psi^{-1/3}_2 o_{(A}\iota_{B)}$ is of type $\{0,0\}$, then 
$D_{A'A}$ reduces to $\c_{A'A}$ and we have
\begin{equation}
 \c_{B'(B}K_{AC)}=0,
\end{equation}
thus we recover the well-known ordinary Killing spinor of type D solutions, in this case as a particular 
example of the more general weighted Killing spinors.

Now consider the 4-dimensional identities (\ref{maxwell1})-(\ref{grav3}). 
The discussion that follows is more transparent if we think of these identities 
in their spinorial formulation, which can be found in the introduction in \cite{Araneda:2016iwr}. 
If the linearized field equations are imposed (i.e. left hand side equal to zero), 
we get the generalized wave equations for rescaled components of the fields, 
in terms of the modified wave operator $\teuk_p=D^{\a}D_{\a}$,
and it is evident that these rescaled components are given by $L^{A_1...A_n}\phi_{A_1...A_n}$, 
where $L^{A_1...A_n}$ is one of the weighted Killing spinors (\ref{wks1})-(\ref{wks3}) 
(and $\phi_{A_1...A_n}$ is the Dirac, Maxwell, or Weyl curvature spinor). 
Therefore, the conclusion is that the identities (\ref{maxwell1})-(\ref{grav3}) (together with the identities 
(1.10) and (1.11) in \cite{Araneda:2016iwr}, which are valid for Dirac fields) can be thought of as 
a generalization of the spin reduction process of Penrose in Minkowski to a
{\em weighted spin reduction} for 4-dimensional curved Einstein spacetimes of Petrov type D. 
The tensorial formulation of these results (see below) gives then a justification for the name of 
``higher dimensional weighted spin reduction'' that we have been using throughout this work.

For the tensor version of these weighted Killing spinors, we have the following:
\begin{lem}
 Let $\ell^{\a}$ and $n^{\a}$ be vector fields aligned to the principal null directions of a 4-dimensional Einstein 
 spacetime of Petrov type D. Then $k^{0\a}=\ell^{\a}$ and $k^{1\a}=|\Psi_2|^{-2/3}n^{\a}$ are conformal Killing vectors with 
 respect to $D_{\a}$:
 \begin{equation}
  D_{(\a}k^{I}_{\b)}=\l^{I}g_{\a\b}, \hspace{1cm} I=0,1,
 \end{equation}
 where $\l^0=-\frac{(\rho+\bar\rho)}{2}$ and $\l^1=-\frac{(\rho+\bar\rho)}{2}|\Psi_2|^{-2/3}$.
\end{lem}
\noindent
This implies that $o_A\bar{o}_{A'}$ and $|\Psi_2|^{-2/3}\iota_A\bar{\iota}_{A'}$ are valence $(1,1)$ weighted Killing spinors.
On the other hand, note that $o_Ao_B$, $\Psi^{-1/3}_2o_{(A}\iota_{B)}$ and $\Psi^{-2/3}_2\iota_A\iota_B$ are 
valence $(2,0)$ weighted Killing spinors.
Then:
\begin{lem}
The complex 2-forms $Z^0_{\a\b}:=U_{\a\b}$, $Z^1_{\a\b}:=\Psi^{-1/3}_2 W_{\a\b}$ and $Z^2_{\a\b}:=\Psi^{-2/3}_2 V_{\a\b}$ are 
weighted conformal Killing-Yano tensors with respect to $D_{\a}$:
\begin{align}
 D_{(\a}Z^I_{\b)\g}=g_{\g(\a}\xi^I_{\b)}-g_{\a\b}\xi^I_{\g}, \hspace{1cm} I=0,1,2,
\end{align}
where $\xi^I_{\a}=\frac{1}{3}D^{\b}Z^{I}_{\a\b}$.
\end{lem}

\section{Conclusions}\label{sec-conclusions}

A generalization to higher dimensions of the results in \cite{Araneda:2016iwr} for linear operators on algebraically 
special backgrounds is given. One of the main obstructions to obtain this generalization is the extensive 
use in \cite{Araneda:2016iwr} of spinorial methods which, although powerful in four dimensions, are not
particularly well suited  to our purposes when working in higher dimensions. Instead, we have made extensive 
use of the generalized GHP formalism and interpreted some of its operators in terms of
newly introduced covariant derivatives on GHP related fiber bundles. 
As was shown in \cite{Durkee:2010xq}, decoupling in $d>4$ is much more restrictive than in $d=4$, 
and, although our methods are quite general, they give decoupled equations in special backgrounds 
to which we restricted, such as the Kundt class and static black holes.

In \cite{Araneda:2016iwr}, for the class of 4-dimensional Einstein spacetimes of Petrov type D,
a general pattern is given of operator identities between spinor/tensor field equations and 
scalar weighted (Teukolsky-like) equations, which is summarized in our formula (\ref{teuk4d}) in the present work.
The left hand side of (\ref{teuk4d}) involves the projection on a principal null frame of a spinor operator 
applied to a spin-$\s$ field,
while the right hand side uses a family of weighted wave operators introduced in \cite{Andersson1, Bini}. 
In the present work we generalize each side of (\ref{teuk4d}) in a somewhat 
different way, and the required connection is given in the main identity (\ref{maineq}).

For the generalization of the {\em right} hand side of (\ref{teuk4d}),
we modified the $d$-dimensional GHP covariant derivative by the addition of a particular $\mf{g}_o$-valued 1-form, 
where $\mf{g}_o=\R\oplus\mf{so}(d-2)$ is the Lie algebra of the GHP group,
and we showed that the higher dimensional Teukolsky and related 
equations adopt a wave-like form in terms of the weighted wave operator associated to the modified connection, 
as it was shown to happen in the 4-dimensional case in \cite{Bini}.
The adjointness property of the higher dimensional Teukolsky system becomes transparent in
these terms, as can be seen from equation (\ref{adjteuk}).

For the generalization of the {\em left} hand side of (\ref{teuk4d}), it is crucial to understand 
the tensorial structure of the operators involved, in a manner independent of the spacetime dimension. 
This is non-trivial since 4-dimensional spinor methods make extensive use of Hodge duality
to work with (anti)-self-dual objects, something that is not available in arbitrary dimensions. 
To this end, we introduced modified covariant exterior derivative and codifferential
operators for tensor-valued differential forms, adding a particular (ordinary) 1-form $\g$. 
Composing these operators with the standard ones, we constructed a generalized ``modified Laplace-de Rham operator'',
\begin{equation*}
 \Ds_{\s}\Ds^{\star}+\Ds^{\star}_{\s}\Ds=
 (\Ds+2\s\g\wedge)\Ds^{\dag}+(\Ds^{\dag}+2\s\g^{\#}\lrcorner)\Ds,
\end{equation*}
and we showed that, for the algebraically special spacetimes considered in this work,
appropriate projections of this operator lead to the higher dimensional scalar, weighted wave equations,
{\em without assuming that the linearized field equations are satisfied}. 
The result is then formula (\ref{maineq}), which is the higher dimensional generalization of (\ref{teuk4d}).

Since our results hold off shell, we find operator identities 
that enable us to apply the adjoint operator method introduced by  Wald  in \cite{Wald} to produce 
solutions of higher spin linear field equations from solutions of wave-like scalar equations.
Wald's technique requires, however, that some of the involved operators be 
self-adjoint, whereas the identities we first found do not satisfy this requirement. 
For the spin 1 Maxwell case we solved this problem by introducing a vector potential, in terms of which the Maxwell
operator is self-adjoint (see corollaries \ref{cor-maxkundt}, \ref{cor-maxsbh1} and \ref{cor-maxsbh2}). 
For linearized gravity we used the Bianchi identities to 
recast our  results in terms of the linearized Einstein operator, which is also
self-adjoint (see corollaries \ref{cor-gravkundt} and \ref{cor-gravsbh}). 
This allowed to apply Wald's method to generate solutions of the  Maxwell and linearized
gravitational field equations, in a very compact form, from solutions of scalar, weighted wave equations. 
It is  worth noting that {\em  all the equations} involve the
{\em same} modified Laplace-de Rham operator, therefore all of the reconstruction formulae, 
(\ref{Akundt}), (\ref{hkundt}), (\ref{VPcky}), (\ref{VPgky}), (\ref{hsbh}), follow a general symmetry pattern. 
Further applications of our off shell identities are the construction of symmetry
operators for the Maxwell and linearized Einstein equations, as well as for
the scalar weighted wave equations.

We have also given relations between the modified covariant derivatives $D_{\a}$ that produce the 
higher dimensional weighted wave operators $\teuk_{(b,s)}$, and tensor fields associated to 
different kind of Killing symmetries with respect to $D_{\a}$. 
More specifically, we showed that the (appropriately rescaled) null vector fields aligned to the 
WANDs of the algebraically special spacetimes considered in this 
work, turn out to be conformal Killing vectors with respect to $D_{\a}$. 
We have also shown that the ``projection'' tensors in the main identity (\ref{maineq}) 
are made out of conformal Killing-Yano tensors with respect to $D_{\a}$.

Finally, we applied the results of this work to the case of 4-dimensional Einstein spacetimes 
of Petrov type D, complementing the results in \cite{Araneda:2016iwr}.
More precisely, we showed that the tensorial structure of the spinorial operator 
in the right hand side of (\ref{teuk4d}) corresponds to a modified Laplace-de Rham operator acting on 
tensor-valued differential forms, as happens for the higher dimensional case. 
We also showed that the ``projection'' spinor in (\ref{teuk4d}) (that is ${}^{b}P^{A_1...A_{2\s}}$)
turns out to be a {\em weighted} Killing spinor
with respect to the (4-dimensional) modified connection $D_{\a}$, which is made out of 
``weighted twistors'' (see lemma \ref{lem-WT}). 
This allows to interpret our results as a generalization of Penrose's spin reduction in a 4-dimensional flat 
spacetime (introduced in \cite{Penrose3}), to a {\em weighted spin reduction} valid for 4-dimensional curved, 
Einstein spacetimes of Petrov type D, 
as well as for the higher dimensional algebraically special spacetimes considered in this work.
Further applications of these weighted Killing objects and the underlying connection are currently
under investigation \cite{Araneda2}.

\section*{Acknowledgments}

It is a pleasure to thank my advisor, Gustavo Dotti, for a careful reading of this manuscript and valuable comments on it. 
Part of this work was done while I was a participant in the programme ``Between Geometry and Relativity'' 
at the Erwin Schr\"odinger Institute for Mathematical Physics (ESI), Vienna, Austria, in July 2017, 
and in the School ``Quantum fields, condensed matter and information theory'' at the Instituto Balseiro, 
Bariloche, Argentina, in October 2017; I thank both institutions for financial support.
This work is supported by a doctoral fellowship from CONICET (Argentina).

\appendix

\section{Additional identities}

\subsection{Vielbein derivatives}\label{app-vd}

The weighted directional derivatives of the vielbein vectors are
\begin{align}
  & \text{\th}\ell^{\a}=\kappa^i m^{\a}_i,  \hspace{0.8cm} \text{\th}n^{\a}=\tau'^i m^{\a}_i,
  \hspace{0.8cm}  \text{\th}m^{\a}_i=-\t'_i\ell^{\a}-\k_i n^{\a}, \label{vd1} \\
  & \text{\th}'\ell^{\a}=\tau^i m^{\a}_i,  \hspace{0.8cm} \text{\th}'n^{\a}=\k'^i m^{\a}_i,  
  \hspace{0.8cm}   \text{\th}'m^{\a}_i=-\k'_i\ell^{\a}-\t_i n^{\a}, \label{vd2} \\
  & \text{\dh}_i\ell^{\a}=\rho^{k}{}_{i} m^{\a}_k,  \hspace{0.8cm} \text{\dh}_i n^{\a}=\rho'^{k}{}_{i} m^{\a}_k, 
  \hspace{0.8cm}  \text{\dh}_i m^{\a}_j=-\rho'_{ji}\ell^{\a}-\rho_{ji}n^{\a}; \label{vd3}
\end{align}
and their divergence is
\begin{align}
 \T_{\a}\ell^{\a}&=\rho, \label{dl} \\
 \T_{\a}n^{\a}&=\rho', \label{dn} \\
 \T_{\a}m^{\a}_i&=-\tau_i-\tau'_i. \label{dm}
\end{align}

\subsection{Maxwell fields}\label{app-maxwell}

An arbitrary 2-form $F_{\a\b}$ can be expanded in the frame $\{e^{a}_{\a}\}$ as
\begin{equation}
 F_{\a\b}=2F n_{[\a}\ell_{\b]}+2\vp_i n_{[\a} m^i_{\b]}+2\vp'_i \ell_{[\a} m^i_{\b]} 
 +F_{ij} m^i_{[\a} m^j_{\b]},
\end{equation}
where we have defined the components in table \ref{tableGHPmaxwell}.
\begin{table}[H]
\begin{center}
\begin{tabular}{|l|c|c|}\hline
 Component & boost $b$ & spin $s$ \\ \hline
 $\vp_i:=F_{0i}=F_{\a\b}\ell^{\a}m^{\b}_i$ & 1 & 1 \\
 $F:=F_{01}=F_{\a\b}\ell^{\a}n^{\b}$ & 0 & 0 \\
 $F_{ij}:=F_{\a\b}m^{\a}_i m^{\b}_j$ & 0 & 2 \\
 $\vp'_i:=F_{1i}=F_{\a\b}n^{\a} m^{\b}_i$ & $-1$ & 1 \\ \hline
\end{tabular}
\caption{Definition of GHP Maxwell scalars.}
\label{tableGHPmaxwell}
\end{center}
\end{table}
Maxwell equations are
\begin{equation}\label{Maxwelleq}
 \c^{\a}F_{\a\b}=0, \hspace{1cm} \c_{[\g}F_{\a\b]}=0.
\end{equation}
Projecting the expression $\c^{\a}F_{\a\b}$ in the frame $\{e^{\a}_a\}$, we get:
\begin{align}
\nonumber \c^{\a}F_{\a\b}=&[-(\text{\th}+\rho)F-(\text{\dh}^i-\t'^i)\vp_i-\k_i\vp'_i+\rho^{ij}F_{ij}]n_{\b} \\
\nonumber & +[(\text{\th}'+\rho')F-(\text{\dh}^i-\t^i)\vp'_i-\k'^i\vp_i+\rho'^{ij}F_{ij}]\ell_{\b} \\
\nonumber & +[(\text{\th}'+\rho')\vp_i+(\text{\th}+\rho)\vp'_i-\rho'_{i}{}^{j}\vp_j-\rho_{i}{}^{j}\vp'_j\\
 & +(\text{\dh}^j-\t^j-\t'^j)F_{ji}+(\tau_i-\tau'_i)F]m^i_{\b}. \label{max-div}
\end{align}
Therefore, the first equation in (\ref{Maxwelleq}) is equivalent to setting to zero all the terms between 
square brackets above independently.
As for $\c_{[\a}F_{\b\g]}$, we have the expansion
\begin{equation}
 \c_{[\a}F_{\b\g]}=A_i n_{[\a}\ell_{\b}m^i_{\g]}+B_{ij}\ell_{[\a}m^i_{\b}m^j_{\g]}+C_{ij}n_{[\a}m^i_{\b}m^j_{\g]}
 +D_{ijk}m^i_{[\a}m^j_{\b}m^k_{\g]},
\end{equation}
where the coefficients $A_i$, $B_{ij}=B_{[ij]}$, $C_{ij}=C_{[ij]}$ and $D_{ijk}=D_{[ijk]}$ are given by
\begin{align}
 &A_i=2[(\text{\dh}_i-\tau_i-\tau'_i)F-\text{\th}'\vp_i-\rho'^{j}{}_{i}\vp_j+\text{\th}\vp'_i+\rho^{j}{}_{i}\vp'_j
  +(\tau^j-\tau'^j)F_{ji}] \label{m4} \\
 &B_{ij}=-2\text{\dh}_{[i}\vp_{j]}+\text{\th}F_{ij}+2\t'_{[i}\vp_{j]}-2F_{k[i}\rho^{k}{}_{j]}+2\k_{[i}\vp'_{j]}+2\rho_{[ij]}F \label{m5} \\
 &C_{ij}=-2\text{\dh}_{[i}\vp'_{j]}+\text{\th}'F_{ij}+2\t_{[i}\vp'_{j]}-2F_{k[i}\rho'^{k}{}_{j]}+2\k'_{[i}\vp_{j]}-2\rho'_{[ij]}F \label{m6} \\
 &D_{ijk}=\text{\dh}_{[k}F_{ij]}-2\vp_{[i}\rho'_{jk]}-2\vp'_{[i}\rho_{jk]}. \label{m7}
\end{align}
The second Maxwell equation in (\ref{Maxwelleq}) is again equivalent to setting to zero all the previous components 
independently.

\subsection{Bianchi identity components}\label{app-bianchi}

In the proofs of theorems \ref{thm-grav-kundt} and \ref{thm-grav-sbh}, we will need some frame components 
of $\c_{[\e}R_{\a\b]\g\d}$, $\c_{[\e}C_{\a\b]\g\d}$, and their contracted forms. 
Of course, $\c_{[\e}R_{\a\b]\g\d}$ is identically zero independently of any field equations, but we will not need 
this fact in the proofs.
Below we give the identities needed, that are valid for an arbitrary spacetime.

For the proof of theorem \ref{thm-grav-kundt}, we need the following components:
\begin{align}
\nonumber  3\ell^{\a}m^{\b}_{(i}\ell^{\g}m^{\d}_{j)}n^{\e}\c_{[\e}R_{\a\b]\g\d} = &
  \text{\th}'\wt\Om_{ij}+\rho'^{k}{}_{(i}\wt\Om_{j)k}-\text{\th}\wt\Phi^S_{ij}-\rho^{k}{}_{(i}\wt\Phi_{j)k}-\rho_{(ij)}\wt\Phi \\
\nonumber  &- 2\rho^{k}{}_{(i}\wt\Phi^A_{j)k} -(\text{\dh}_{(i}-2\t_{(i}-\t'_{(i})\wt\Psi_{j)}+2\t^k\wt\Psi_{(ij)k} \\
  & -\t'^k\wt\Psi_{(ij)k}+\k_{(i}\wt\Psi'_{j)}-\k^k\wt\Psi'_{(ij)k}, \label{b1} 
\end{align}
\begin{align}
\nonumber 3\ell^{\a}m^{\b}_{(i}\ell^{\g} m^{\d}_{j)}m^{\e}_{k}\c_{[\e}R_{\a\b]\g\d} = &  
  (\text{\dh}_k-\t'_k)\wt\Om_{ij}-(\text{\dh}_{(i}-\t'_{(i})\wt\Om_{j)k}+\text{\th}\wt\Psi_{(ij)k} \\
\nonumber &  +2\rho^{l}{}_{k}\wt\Psi_{(ij)l}-\rho^{l}{}_{(i}\wt\Psi_{|k|j)l}-\rho^{l}{}_{(i}\wt\Psi_{j)kl}\\
\nonumber & +2\wt\Psi_{(i}\rho_{j)k}-\wt\Psi_k\rho_{(ij)}-\rho_{k(i}\wt\Psi_{j)}+\k_{(i}\wt\Phi_{j)k} \\
 & -\k_k\wt\Phi^S_{ij}-\k_l\wt\Phi_{k(ij)}{}^{l} +2\k_{(i}\wt\Phi^A_{j)k}. \label{b2} 
\end{align}

\noindent
We also need the following contracted components:
\begin{align}
\nonumber \ell^{\a}m^{\b}_{j}\ell^{\g}\c^{\d}R_{\a\b\g\d}=& \text{\th}\wt\Psi_j+\rho_{j}{}^{k}\wt\Psi_k+\rho\wt\Psi_j
  +(\text{\dh}^k-\t'^k)\wt\Om_{jk} \\
  & -\rho^{lk}\wt\Psi_{klj}-\rho^{lk}\wt\Psi_{jlk}+2\k^k\wt\Phi^A_{jk}+\k_j\wt\Phi+\k^k\wt\Phi_{jk} \label{cb1}
\end{align}
\begin{align}
\nonumber \ell^{\a}m^{\b}_{(i} m^{\g}_{j)}\c^{\d}R_{\a\b\g\d}=&-(\text{\th}+\rho')\wt\Om_{ij}-(\text{\th}+\rho)\wt\Phi^{S}_{ij}
 +(\text{\dh}^k-\t'^k)\wt\Psi_{(ij)k}\\
\nonumber &+\rho'_{(i}{}^{k}\wt\Om_{j)k}+\rho_{(i}{}^{k}\wt\Phi_{j)k}+2\rho_{(i}{}^{k}\wt\Phi^{A}_{j)k}-\rho^{kl}\wt\Phi_{k(ij)l}\\
 &-2\t^k\wt\Psi_{(ij)k}+\t'_{(i}\wt\Psi_{j)}-2\t_{(i}\wt\Psi_{j)}-\k^k\wt\Psi'_{(ij)k}+\k_{(i}\wt\Psi'_{j)} \label{cb2}
\end{align}

Theorem \ref{thm-grav-sbh} uses the Weyl tensor instead; we will need the following components of $\c_{[\e}C_{\a\b]\g\d}$:
\begin{multline}
3 m^{\a}_i m^{\b}_j l^{\g}n^{\d}n^{\e}\c_{[\e}C_{\a\b]\g\d}= 2\text{\th}'\Phi^A_{ij} +2\k'_{[i}\Psi_{j]}-2\t_{[i}\Psi'_{j]}+\t^k\Psi'_{kij}
 -\k'^k\Psi_{kij}  \\
  +2\text{\dh}_{[i}\Psi'_{j]}-4\rho'^{k}{}_{[i}\Phi^A_{j]k}+2\rho'_{[ji]}\Phi-2\rho^{k}{}_{[i}\Om'_{j]k}-2\Phi_{k[i}\rho'^{k}{}_{j]}, \label{b6} 
\end{multline}
\begin{equation}
 3 m^{\a}_i m^{\b}_j l^{\g}n^{\d}\ell^{\e}\c_{[\e}C_{\a\b]\g\d}=-(3 m^{\a}_i m^{\b}_j l^{\g}n^{\d}n^{\e}\c_{[\e}C_{\a\b]\g\d})', 
\end{equation}
\begin{align}
\nonumber 3 \ell^{\a} n^{\b} m^{\g}_i m^{\d}_j m^{\e}_k \c_{[\e}C_{\a\b]\g\d}= & 2\text{\dh}_k\Phi^A_{ij}+\rho^{l}{}_{k}\Psi'_{lij}
   -\rho'^{l}{}_{k}\Psi_{lij}-2\Psi_{[i}\rho'_{j]k}+2\Psi'_{[i}\rho_{j]k} \\
\nonumber   & +\text{\th}\Psi'_{kij}-\text{\th}'\Psi_{kij}-2(\t_k+\t'_k)\Phi^A_{ij}+(\t'^l-\t^l)\Phi_{klij}\\
   & +2\Phi_{k[i}\t_{j]}-2\t'_{[j}\Phi_{i]k}+\Om_{k[i}\k'_{j]}-\Om'_{k[i}\k_{j]},
\end{align}
\begin{align}
\nonumber 3 m^{\a}_i m^{\b}_j l^{\g}n^{\d} m^{\e}_k\c_{[\e}C_{\a\b]\g\d}=& 6\text{\dh}_{[k}\Phi^A_{ij]}+6\Psi_{[j}\rho'_{ik]}-6\Psi'_{[j}\rho_{ik]}\\
  & +6\rho^{l}{}_{[k}\Psi'_{|l|ij]}-6\rho'^{l}{}_{[k}\Psi_{ |l| ij]},
\end{align}
\begin{equation}
 3 \ell^{\a} n^{\b} \ell^{\g} m^{\d}_{[j} m^{\e}_{i]} \c_{[\e}C_{\a\b]\g\d}= -\tfrac{1}{2}3 m^{\a}_i m^{\b}_j l^{\g}n^{\d}\ell^{\e}\c_{[\e}C_{\a\b]\g\d},
\end{equation}
\begin{equation}
 3 \ell^{\a} n^{\b} n^{\g} m^{\d}_{[j} m^{\e}_{i]} \c_{[\e}C_{\a\b]\g\d}=-(3 \ell^{\a} n^{\b} \ell^{\g} m^{\d}_{[j} m^{\e}_{i]} \c_{[\e}C_{\a\b]\g\d})',
\end{equation}
and the contracted components:
\begin{align}
\nonumber m^{\a}_i m^{\b}_j \ell^{\g}\c^{\d}C_{\a\b\g\d}=& 2\text{\th}\Phi^{A}_{ij}+\text{\dh}^k\Psi_{kij}+2\rho'_{[i}{}^{k}\Om_{j]k}
 +2\Phi_{k[j}\rho_{i]}{}^{k}-\rho^{kl}\Phi_{ijkl} \\
 &+2\rho\Phi^{A}_{ij}+2\tau'_{[i}\Psi_{j]}-2\k_{[i}\Psi'_{j]}+\k^{k}\Psi'_{kij}-\t'^{k}\Psi_{kij} \label{b12}
\end{align}
\begin{equation}
 m^{\a}_i m^{\b}_j n^{\g}\c^{\d}C_{\a\b\g\d}=(m^{\a}_i m^{\b}_j \ell^{\g}\c^{\d}C_{\a\b\g\d})' \label{b13}
\end{equation}
\begin{align}
\nonumber \ell^{\a} n^{\b} m^{\g}_i \c^{\d}C_{\a\b\g\d}=& 2(\text{\dh}^k-\t'^k-\t^k)\Phi^A_{ik}+(\text{\th}+\rho)\Psi'_i-(\text{\th}'+\rho')\Psi_i\\
\nonumber & +\t'^k\Phi_{ki}-\t^k\Phi_{ik}+(\t'_i-\t_i)\Phi+\rho^{kl}\Psi'_{kil}-\rho'^{kl}\Psi_{kil}\\
 & +\rho'_{i}{}^{k}\Psi_k-\rho_{i}{}^{k}\Psi'_k-\k^k\Om'_{ki}+\k'^k\Om_{ki}
\end{align}
\begin{multline}
 m^{\a}_{[i} n^{\b} m^{\g}_{j]} \c^{\d}C_{\a\b\g\d}= -(\text{\th}'+\rho')\Phi^A_{ij}-\text{\dh}^k\Psi'_{[ij]k}
  +2\rho'_{[i}{}^{k}\Phi^A_{j]k}-\Phi_{k[i}\rho'_{j]}{}^{k}\\
   -\rho'_{kl}\Phi_{[i}{}^{k}{}_{j]}{}^{l}+\Om'_{k[i}\rho_{j]}{}^{k}-\k'_{[i}\Psi_{j]}-\Psi'_{[i}\t_{j]}-\k'^k\Psi_{[ij]k}+\t^k\Psi'_{[ij]k}
\end{multline}
\begin{equation}
 \ell^{\a} m^{\b}_{[i} m^{\g}_{j]} \c^{\d}C_{\a\b\g\d}=-(m^{\a}_{[i} n^{\b} m^{\g}_{j]} \c^{\d}C_{\a\b\g\d})'. \label{b16}
\end{equation}

\section{Details of the proofs}

In this appendix we give the proofs of theorems \ref{thm-grav-kundt} and \ref{thm-grav-sbh} concerning gravitational perturbations.

\subsection{Proof of theorem \ref{thm-grav-kundt}}\label{proofgkundt}

\subsubsection*{Calculation of $U^{\a\b\g\d}{}_{ij}(\Ds_{\s}\Ds^{\star} R)_{\a\b\g\d}$}

Taking into account that in the background spacetime we have $\rho_{ij}|_{\ve=0}=0$, $\k_i|_{\ve=0}=0$ and 
$(\c^{\d}R_{\a\b\g\d})|_{\ve=0}=0$, one can show that
\begin{equation}
 \tfrac{1}{8} \tfrac{d}{d\ve}|_{\ve=0}\left[ 2U^{\a\b\g\e}{}_{ij}(\c_{\e}+4\g_{\e})\c^{\d}R_{\a\b\g\d}\right]
  =  (\text{\dh}_{(i}-4\t_{(i}-\t'_{(i})\dot{S}^1_{j)} - \text{\th}\dot{S}^2_{ij}   \label{auxgrav0}
\end{equation}
where 
\begin{align}
S^1_j = &  \ell^{\a}m^{\b}_{j}\ell^{\g}\c^{\d}R_{\a\b\g\d},\\
S^2_{ij} = & \ell^{\a}m^{\b}_{(i}m^{\g}_{j)}\c^{\d}R_{\a\b\g\d}.
\end{align}
These components are given in (\ref{cb1}) and (\ref{cb2}). Linearizing around a Kundt spacetime, we get
\begin{align}
 \dot{S}^1_j=& \text{\th}\dot{\wt\Psi}_j+(\text{\dh}^k-\t'^k)\dot{\wt\Om}_{jk}+2\wt{\Phi}^A_{j}{}^{k}\dot\k_k+\wt{\Phi}\dot\k_j \\
\nonumber \dot{S}^2_{ij}=& -(\text{\th}'+\rho')\dot{\wt\Om}_{ij} -\dot{\text{\th}}\wt\Phi^S_{ij}-\text{\th}\dot{\wt\Phi}{}^S_{ij} 
 -\dot\rho\wt{\Phi}^S_{ij}+(\text{\dh}^k-2\t^k-\t'^k)\dot{\wt\Psi}_{(ij)k}\\
\nonumber & +\rho'_{(i}{}^{k}\dot{\wt\Om}_{j)k}+\wt{\Phi}_{(j}{}^{k}\dot{\rho}_{i)k}+2\wt\Phi^A_{(j}{}^{k}\dot\rho_{i)k} 
 -\dot\rho_{kl}\wt\Phi^{k}{}_{(ij)}{}^{l} + (\t'_{(i}-2\t_{(i})\dot{\wt\Psi}_{j)}\\
 & -\dot\k_k\wt\Psi'_{(ij)}{}^{k}+\dot\k_{(i}\wt\Psi'_{j)}.
\end{align}
Replacing these expressions in (\ref{auxgrav0}):
\begin{multline}
 \tfrac{d}{d\ve}|_{\ve=0}\left[ 2U^{\a\b\g\e}{}_{ij}(\c_{\e}+4\g_{\e})\c^{\d}R_{\a\b\g\d}\right] \\
  \equiv  T^1_{ij}[\dot{\wt\Om}]+T^2_{ij}[\dot{\wt\Psi}_i]+T^3_{ij}[\dot{\wt\Psi}_{ijk}]+T^4_{ij}[\dot{\wt\Phi}{}^S]
  +T^5_{ij}[\dot{\k}]+T^6_{ij}[\dot\rho,\dot\k], \label{DD*Rk1}
\end{multline}
where we have defined
\begin{align}
\nonumber T^1_{ij}[\dot{\wt\Om}]=& 8[\text{\th}\text{\th}'\dot{\wt\Om}_{ij}+\rho'\text{\th}\dot{\wt\Om}_{ij}
  +(\text{\th}\rho')\dot{\wt\Om}_{ij} -\text{\th}(\rho'_{(i}{}^{k}\dot{\wt\Om}_{j)k})] \\
 & +8(\text{\dh}_{(i}-4\t_{(i}-\t'_{(i})(\text{\dh}^{k}-\t'^{k})\dot{\wt\Om}_{j)k},  \\
 T^2_{ij}[\dot{\wt\Psi}_i]= & 8(\text{\dh}_{(i}-4\t_{(i}-\t'_{(i})\text{\th}\dot{\wt\Psi}_{j)}-8\text{\th}[(\t'_{(i}-2\t_{(i})\dot{\wt\Psi}_{j)}], \\
 T^3_{ij}[\dot{\wt\Psi}_{ijk}]= & -8 \text{\th}[(\text{\dh}^k-2\t^k-\t'^k)\dot{\wt\Psi}_{(ij)k}], \\
 T^4_{ij}[\dot{\wt\Phi}{}^S]= & 8\text{\th}\left[\dot{\text{\th}}\wt\Phi^S_{ij} + \text{\th}\dot{\wt\Phi}{}^S_{ij}\right], \\
 T^5_{ij}[\dot\k]= & 8\text{\th}\left[ \dot\k_k\wt\Psi'_{(ij)}{}^{k}-\dot\k_{(i}\wt\Psi'_{j)} \right], \\
\nonumber T^6_{ij}[\dot\rho,\dot\k]=& 8(\text{\dh}_{(i}-4\t_{(i}-\t'_{(i})(2\wt\Phi{}^A_{j)}{}^{k}\dot\k+\wt\Phi_{j)}{}^{k}\dot\k_k+\wt\Phi\dot\k_{j)}) \\
 & -8\text{\th}\left[ -\wt\Phi{}^S_{ij}\dot\rho +\wt\Phi_{(j}{}^{k}\dot\rho_{i)k}+2\wt\Phi^A_{(i}{}^{k}\dot\rho_{j)k} 
  -\dot\rho_{kl}\wt\Phi^{k}_{(ij)}{}^{l}\right].
\end{align}

\subsubsection*{Calculation of $U^{\a\b\g\d}{}_{ij}(\Ds^{\star}_{\s}\Ds R)_{\a\b\g\d}$}

In the same way as before, it is not difficult to show that
\begin{equation}
 \tfrac{1}{8}\tfrac{d}{d\ve}|_{\ve=0}\left[ 3 U^{\a\b\g\d}{}_{ij}(\c^{\e}+4\g^{\e})\c_{[\e}R_{\a\b]\g\d}\right]
  = \text{\th}\dot{S}^3_{ij} + (\text{\dh}^k-4\t^k-\t'^k)\dot{S}^1_{ijk}     \label{auxgrav1}
\end{equation}
where
\begin{align}
S^3_{ij} = &  3\ell^{\a}m^{\b}_{(i}\ell^{\g}m^{\d}_{j)}n^{\e}\c_{[\e}R_{\a\b]\g\d},\\
S^4_{ijk} = & 3\ell^{\a}m^{\b}_{(i}\ell^{\g} m^{\d}_{j)}m^{\e}_{k}\c_{[\e}R_{\a\b]\g\d}.
\end{align}
The GHP expression of these components is given in (\ref{b1}), (\ref{b2}).
The linearization gives
\begin{align}
\nonumber \dot{S}^3_{ij}=&  \text{\th}'\dot{\wt\Om}_{ij}+\rho'^{k}{}_{(i}\dot{\wt\Om}_{j)k}-\dot{\text{\th}}\wt\Phi^S_{ij} - \text{\th}\dot{\wt\Phi}{}^S_{ij}
 -\dot\rho^{k}{}_{(i}\wt\Phi_{j)k}-\dot\rho_{(ij)}\wt\Phi - 2\dot\rho^{k}{}_{(i}\wt\Phi^A_{j)k} \\
  & -(\text{\dh}_{(i}-2\t_{(i}-\t'_{(i})\dot{\wt\Psi}_{j)}+2\t^k\dot{\wt\Psi}_{(ij)k} -\t'^k\dot{\wt\Psi}_{(ij)k} 
   +\dot\k_{(i}\wt\Psi'_{j)}-\dot\k^k\wt\Psi'_{(ij)k}, \\
 \nonumber \dot{S}^4_{ijk}=& (\text{\dh}_k-\t'_k)\dot{\wt\Om}_{ij}-(\text{\dh}_{(i}-\t'_{(i})\dot{\wt\Om}_{j)k}+\text{\th}\dot{\wt\Psi}_{(ij)k} 
 +\dot\k_{(i}\wt\Phi_{j)k} -\dot\k_k\wt\Phi^S_{ij}\\
 & -\dot\k_l\wt\Phi_{k(ij)}{}^{l} + 2\dot\k_{(i}\wt\Phi^A_{j)k}.  
\end{align}
Therefore, we get
\begin{multline}
 \tfrac{d}{d\ve}|_{\ve=0}\left[ 3U^{\a\b\g\d}{}_{ij}(\c^{\e}+4\g^{\e})\c_{[\e}R_{\a\b]\g\d}\right] \\
  \equiv  T^7_{ij}[\dot{\wt\Om}]+T^8_{ij}[\dot{\wt\Psi}_i]+T^9_{ij}[\dot{\wt\Psi}_{ijk}]+T^{10}_{ij}[\dot{\wt\Phi}{}^S]
  +T^{11}_{ij}[\dot{\k}]+T^{12}_{ij}[\dot\rho,\dot\k], \label{D*DRk2}
\end{multline}
where
\begin{align}
\nonumber T^7_{ij}[\dot{\wt\Om}]=& 8[\text{\th}\text{\th}'\dot{\wt\Om}_{ij} +\text{\th}(\rho'^{k}{}_{(i}\dot{\wt\Om}_{j)k})] 
 +8(\text{\dh}^k-4\t^k-\t'^k)(\text{\dh}_k-\t'_k)\dot{\wt\Om}_{ij}  \\
 & - 8(\text{\dh}^k-4\t^k-\t'^k)(\text{\dh}_{(i}-\t'_{(i})\dot{\wt\Om}_{j)k}\\
 T^8_{ij}[\dot{\wt\Psi}_i]= & -8\text{\th}[(\text{\dh}_{(i}-2\t_{(i}-\t'_{(i})\dot{\wt\Psi}_{j)}], \\
 T^9_{ij}[\dot{\wt\Psi}_{ijk}]= & 8 \text{\th}[(2\t^k-\t'^k)\dot{\wt\Psi}_{(ij)k}]+8(\text{\dh}^k-4\t^k-\t'^k)\text{\th}\dot{\wt\Psi}_{(ij)k}, \\
 T^{10}_{ij}[\dot{\wt\Phi}{}^S]= & -8\text{\th}\left[\dot{\text{\th}}\wt\Phi^S_{ij} + \text{\th}\dot{\wt\Phi}{}^S_{ij}\right], \\
 T^{11}_{ij}[\dot\k]= & 8\text{\th}\left[ -\dot\k_k\wt\Psi'_{(ij)}{}^{k} + \dot\k_{(i}\wt\Psi'_{j)} \right], \\
\nonumber T^{12}_{ij}[\dot\rho,\dot\k]=& 8(\text{\dh}^k-4\t^k-\t'^k)\left[\dot\k_{(i}\wt\Phi_{j)k}+2\dot\k_{(i}\wt\Phi^A_{j)k}
   -\dot\k_k\wt\Phi^S_{ij}-\dot\k_{l}\wt\Phi_{k(ij)}{}^{l}\right] \\
 & -8\text{\th}\left[ \dot\rho_{k(i}\wt\Phi_{j)}{}^{k}+2\dot\rho_{k(i}\wt\Phi^A_{j)}{}^{k}+\dot\rho_{(ij)}\wt\Phi \right].
\end{align}

\subsubsection*{Final step}

Now with all the expressions for $T^r_{ij}$, $r=1,...,12$, we can put together (\ref{D*DRk2}) and (\ref{DD*Rk1}).
Using the background identity (\ref{idr1-kundt}) and the commutation relation for $[\text{\th}, \text{\dh}_i]$, one obtains at once
\begin{align}
 & T^2_{ij}+T^8_{ij}=0, \\
 & T^3_{ij}+T^9_{ij}=0, \\
 & T^4_{ij}+T^{10}_{ij}=0, \\
 & T^5_{ij}+T^{11}_{ij}=0. 
\end{align}
For the term $T^1_{ij}+T^7_{ij}$, we have to use the following background identities:
\begin{align}
 \text{\th}\rho'-\text{\dh}^k\t'_k= &-\t'^k\t'_k-\Phi-\tfrac{2\L}{d-1}, \\
 -2\text{\dh}_{[i}\t'_{k]}= & -2\Phi^A_{ik}-\text{\th}\rho'_{ki}+\text{\th}\rho'_{ik}, \\
\nonumber [\text{\dh}_{(i},\text{\dh}^k]\dot{\wt\Om}_{j)k}= & \rho'_{(i}{}^{k}\text{\th}\dot{\wt\Om}_{j)k}-\rho'^{k}{}_{(i}\text{\th}\dot{\wt\Om}_{j)k}
  +4\Phi^A_{(i}{}^{k}\dot{\wt\Om}_{j)k}-\Phi^{k}{}_{ij}{}^{l}\dot{\wt\Om}_{kl} \\
  & +2\Phi^S_{(i}{}^{k}\dot{\wt\Om}_{j)k}-\tfrac{2\L}{(d-1)}\left[ \dot{\wt\Om}_{ij}-\tfrac{\dot{\wt\Om}}{(d-2)}\d_{ij} \right]
\end{align}
where $\dot{\wt\Om}=\d^{ij}\dot{\wt\Om}_{ij}$ (observe that this is different from zero since it is the Riemann --not Weyl-- component).
Using also the identity for the commutator $[\text{\th},\text{\th}']$, we get
\begin{align}
 \nonumber\tfrac{1}{8}(T^1_{ij}+T^7_{ij})=& 2\text{\th}' \text{\th}\dot{\wt\Om}_{ij}+\text{\dh}^k\text{\dh}_k\dot{\wt\Om}_{ij}
 +\rho'\text{\th}\dot{\wt\Om}_{ij}-6\t^k\text{\dh}_k\dot{\wt\Om}_{ij}-4\t_{(i}\text{\dh}^k\dot{\wt\Om}_{j)k} + 4\t^k\text{\dh}_{(i}\dot{\wt\Om}_{j)k} \\
\nonumber & +2\left[ \Phi^S_{(i}{}^{k}+5\Phi^A_{(i}{}^{k} \right]\dot{\wt\Om}_{j)k}+\left[ 3\Phi-\tfrac{4(d+2)\L}{(d-1)(d-2)} \right]\dot{\wt\Om}_{ij} \\
 & +\left[-\Phi^{k}{}_{ij}{}^{l}+\tfrac{2\L}{(d-1)(d-2)}\d_{ij}\d^{kl} \right]\dot{\wt\Om}_{kl}.
\end{align}
The only terms remaining are those with perturbed spin coefficients, which can be conveniently rearranged as
\begin{align*}
 \tfrac{1}{8}(T^6_{ij}+T^{12}_{ij})=&2(\text{\dh}_{(i}-4\t_{(i}-\t'_{(i})[\wt\Phi^A_{j)}{}^{k}\dot\k_k]-2\text{\th}[\dot\rho_{k(i}\wt\Phi^A_{j)}{}^{k}]\\
 &+(\text{\dh}_{(i}-4\t_{(i}-\t'_{(i})[\wt\Phi\k_{j)}]-\text{\th}[\dot\rho_{(ij)}\wt\Phi]\\
 &+(\text{\dh}_{(i}-4\t_{(i}-\t'_{(i})[\wt\Phi_{j)}{}^{k}\dot\k_k]-\text{\th}[\dot\rho_{k(i}\wt\Phi_{j)}{}^{k}]\\
 &+(\text{\dh}^k-4\t^k-\t'^k)[\dot\k_{(i}\wt\Phi_{j)k}]-\text{\th}[\wt\Phi_{(i}{}^{k}\dot\rho_{j)k}]\\
 &-(\text{\dh}^k-4\t^k-\t'^k)[\dot\k_{k}\wt\Phi^S_{ij}]+\text{\th}[\dot\rho\wt\Phi^S_{ij}]\\
 &-(\text{\dh}^k-4\t^k-\t'^k)[\dot\k_{l}\wt\Phi_{k(ij)}{}^{l}]+\text{\th}[\dot\rho_{kl}\wt\Phi^{k}{}_{ij}{}^{l}]\\
 &+2(\text{\dh}^k-4\t^k-\t'^k)[\dot\k_{(j}\wt\Phi^A_{i)k}]-2\text{\th}[\wt\Phi^A_{i}{}^{k}\dot\rho_{j)k}].
\end{align*}
Each line in this equation can be dealt with using the perturbed Ricci identity
\begin{equation}
 \text{\dh}_i\dot\k_j-\text{\th}\dot\rho_{ji}=\dot\k_j\t'_i+\t_j\dot\k_i+\dot{\wt{\Om}}_{ji}.
\end{equation}
After a tedious calculation, the perturbed coefficient $\dot\rho_{ij}$ cancels everywhere and we get
\begin{equation*}
  \tfrac{1}{8}(T^6_{ij}+T^{12}_{ij})=2\wt\Phi_{(i}{}^{k}\dot{\wt{\Om}}_{j)k}+4\wt\Phi^A_{(i}{}^{k}\dot{\wt\Om}_{j)k}
  -\wt\Phi^{k}{}_{ij}{}^{l}\dot{\wt\Om}_{kl}+\wt\Phi\dot{\wt\Om}_{ij}-\wt\Phi^S_{ij}\dot{\wt\Om}+\dot\k_k\a_{ij}{}^{k}+\dot\k_{(i}\b_{j)},
\end{equation*}
where $\a_{ij}{}^{k}$ and $\b_{j}$ are defined by
\begin{align*}
 \a_{ij}{}^{k}:=&\text{\dh}_{(i}\wt\Phi_{j)}{}^{k}-3\t_{(i}\wt\Phi_{j)}{}^{k}+3\t^l\wt\Phi_{l(ij)}{}^{k}-\text{\dh}^k\wt\Phi^S_{ij}+3\t^k\wt\Phi^S_{ij}
 +2\text{\dh}_{(i}\wt\Phi^A_{j)}{}^{k}\\
 &-6\t_{(i}\wt\Phi^A_{j)}{}^{k}-\text{\dh}^{l}\wt\Phi_{l(ij)}{}^{k},\\
 \b_{j}:=&\text{\dh}^k\wt\Phi_{jk}+\text{\dh}_{j}\wt\Phi+2\text{\dh}^k\wt\Phi^A_{jk}-3\t_j\wt\Phi-3\t^k\wt\Phi_{jk}-6\t^k\wt\Phi^A_{jk}.
\end{align*}
To deal with these terms, we use the following background Bianchi identities:
\begin{align}
 0=& \text{\dh}_{j}\wt\Phi_{ik}-\text{\dh}_{k}\wt\Phi_{ij}+\wt\Phi_{ij}\t_k-\wt\Phi_{ik}\t_j-2\wt\Phi^A_{jk}\t_i+\wt\Phi_{kjil}\t^l,\\
 0=& \text{\dh}_{k}\wt\Phi_{ijlm}+\text{\dh}_{m}\wt\Phi_{ijkl}+\text{\dh}_{l}\wt\Phi_{ijmk}.
\end{align}
One can then show that $\a_{ij}{}^{k}\equiv 0$, and $\b_k=-\d^{ij}\a_{ijk}\equiv 0$.
Finally, replacing the background Riemann components with the Weyl components, the result is
\begin{align}
\nonumber & -\tfrac{1}{8}\tfrac{d}{d\ve}|_{\ve=0}\left\{U^{\a\b\g\d}{}_{ij}[(\Ds_s\Ds^{\star}+\Ds^{\star}_2\Ds)R]_{\a\b\g\d} \right\} \\
\nonumber  = &2\text{\th}' \text{\th}\dot{\wt\Om}_{ij}+\text{\dh}^k\text{\dh}_k\dot{\wt\Om}_{ij}
 +\rho'\text{\th}\dot{\wt\Om}_{ij}-6\t^k\text{\dh}_k\dot{\wt\Om}_{ij}-4\t_{(i}\text{\dh}^k\dot{\wt\Om}_{j)k} + 4\t^k\text{\dh}_{(i}\dot{\wt\Om}_{j)k} \\
\nonumber & +4\left[ \Phi^S_{(i}{}^{k}+4\Phi^A_{(i}{}^{k} \right]\dot{\wt\Om}_{j)k}+4\left[\Phi-\tfrac{(d+2)\L}{(d-1)(d-2)} \right]\dot{\wt\Om}_{ij} \\
 & +\left[-2\Phi^{k}{}_{ij}{}^{l}+\left(-\Phi^S_{ij} + \tfrac{2\L}{(d-1)(d-2)}\d_{ij}\right) \d^{kl} \right]\dot{\wt\Om}_{kl}.
\end{align}

\subsection{Proof of theorem \ref{thm-grav-sbh}}\label{proofgsbh}

\subsubsection*{Calculation of $W^{\a\b\g\d}{}_{ij}(\Ds_{\s}\Ds^{\star} C)_{\a\b\g\d}$}

Consider first the linearization of the term $W^{\a\b\g\d}{}_{ij}(\Ds_{\s}\Ds^{\star} C)_{\a\b\g\d}$. 
Using the fact that $(\Ds^{\star}C)_{\a\b\g}|_{\ve=0}=0$, we can evaluate the operator $W^{\a\b\g\d}{}_{ij}\Ds_{\s}$ 
on the background. One can then show that
\begin{multline}
\tfrac{d}{d\ve}|_{\ve=0}\left[ 2W^{\a\b\g\e}{}_{ij}(\c_{\e}+4\g_{\e})\c^{\d}C_{\a\b\g\d} \right] = \\
 4(\text{\th}'+\tfrac{2\rho'}{n})[\Phi^{2q}\dot S^1_{ij}]-4(\text{\th}+\tfrac{2\rho}{n})[\Phi^{2q}\dot S^2_{ij}]
 +8\text{\dh}_{[j}[\Phi^{2q}\dot S^{3}_{i]}]-\tfrac{8\rho}{n}\Phi^{2q}\dot S^4_{ij}-\tfrac{8\rho'}{n}\Phi^{2q}\dot S^5_{ij}, \label{DD*Csbh0}
\end{multline}
where we have defined
\begin{align}
 S^1_{ij} = & m^{\a}_i m^{\b}_j \ell^{\g}\c^{\d}C_{\a\b\g\d}\\
 S^2_{ij} = & m^{\a}_i m^{\b}_j n^{\g}\c^{\d}C_{\a\b\g\d}\\
 S^3_{i} = & \ell^{\a} n^{\b} m^{\g}_i \c^{\d}C_{\a\b\g\d}\\
 S^4_{ij} = & m^{\a}_{[j} n^{\b} m^{\g}_{i]} \c^{\d}C_{\a\b\g\d}\\
 S^5_{ij} = & \ell^{\a} m^{\b}_{[j} m^{\g}_{i]} \c^{\d}C_{\a\b\g\d}.
\end{align}
These components are explicitly given in (\ref{b12})-(\ref{b16}). For their linearization, we find:
\begin{align}
 \dot{S}^1_{ij} =& 2(\text{\th}+(1+\tfrac{1}{n})\rho)\dot\Phi^A_{ij}+\text{\dh}^k\dot\Psi_{kij}+\tfrac{2(n+1)}{n(n-1)}\Phi\dot\rho_{[ij]}, \\
 \dot{S}^2_{ij} =& (\dot{S}^1_{ij})', \\
 \dot{S}^3_{ij} =& 2\text{\dh}^k\dot\Phi^A_{ik}+(\text{\th}+\rho)\dot\Psi'_i-(\text{\th}'+\rho')\dot\Psi_i+\tfrac{(n+1)}{n}\Phi(\dot\t'_i-\dot\t_i), \\
 \dot{S}^4_{ij} =& \tfrac{1}{2} \dot{S}^2_{ij}, \\
 \dot{S}^5_{ij} =& -\tfrac{1}{2} \dot{S}^1_{ij}
\end{align}
(where we have used that $\Psi_{[ij]k}=-\frac{1}{2}\Psi_{kij}$). Replacing in (\ref{DD*Csbh0}), we get:
\begin{multline}
\tfrac{d}{d\ve}|_{\ve=0}\left[ 2W^{\a\b\g\e}{}_{ij}(\c_{\e}+4\g_{\e})\c^{\d}C_{\a\b\g\d} \right] \equiv \\
 T^1_{ij}[\dot\Phi^A]+T^2_{ij}[\dot\Psi_i,\dot\Psi'_i]
 +T^3_{ij}[\dot\Psi_{ijk},\dot\Psi'_{ijk}]+T^4_{ij}[\dot\rho,\dot\rho',\dot\t,\dot\t'], \label{DD*Csbh1}
\end{multline}
where 
\begin{align}
\nonumber T^1_{ij}[\dot\Phi^A]=& 16\text{\dh}_{[j}\text{\dh}^k[\Phi^{2q}\dot\Phi^A_{i]k}] 
    + 8(\text{\th}+\tfrac{\rho}{n})[\Phi^{2q}(\text{\th}'+(1+\tfrac{1}{n})\rho')\dot\Phi^A_{ij}]\\
   & +8(\text{\th}'+\tfrac{\rho'}{n})[\Phi^{2q}(\text{\th}+(1+\tfrac{1}{n})\rho)\dot\Phi^A_{ij}]\\
 T^2_{ij}[\dot\Psi_i,\dot\Psi'_i]=& 8\Phi^{2q} \left[ -\text{\dh}_{[i}(\text{\th}+\rho)\dot\Psi'_{j]} 
    +\text{\dh}_{[i}(\text{\th}'+\rho')\dot\Psi_{j]} \right], \\
 T^3_{ij}[\dot\Psi_{ijk},\dot\Psi'_{ijk}]=& 4\Phi^{2q} \left[-(\text{\th}+\tfrac{3\rho}{n})\text{\dh}^k\dot\Psi'_{kij} 
    +(\text{\th}'+\tfrac{3\rho'}{n})\text{\dh}^k\dot\Psi_{kij}  \right], \\
\nonumber T^4_{ij}[\dot\rho,\dot\rho',\dot\t,\dot\t']=& \tfrac{8(n+1)}{n}\Phi^{2q+1}
  \left[-\text{\dh}_{[i}\dot\t'_{j]}+\text{\dh}_{[i}\dot\t_{j]}
      -\tfrac{1}{n}\left(\tfrac{n-2}{n-1}\right)(\rho'\dot\rho_{[ij]}-\rho\dot\rho'_{[ij]}) \right. \\
   & \left. +\tfrac{1}{(n-1)}(-\text{\th}\dot\rho'_{[ij]}+\text{\th}'\dot\rho_{[ij]}) \right].
\end{align}

\subsubsection*{Calculation of $W^{\a\b\g\d}{}_{ij}(\Ds^{\star}_{\s}\Ds C)_{\a\b\g\d}$}

Consider now the term $W^{\a\b\g\d}{}_{ij}(\Ds^{\star }_{\s}\Ds C)_{\a\b\g\d}$. We have
\begin{align}
\nonumber & \tfrac{d}{d\ve}|_{\ve=0}\left[ 3 W^{\a\b\g\d}{}_{ij}(\c^{\e}+4\g^{\e})\c_{[\e}C_{\a\b]\g\d} \right]= \\
\nonumber  & 4(\text{\th}+(1+\tfrac{1}{n})\rho)][\Phi^{2q}\dot{S}^6_{ij}] + 4(\text{\th}'+(1+\tfrac{1}{n})\rho')[\Phi^{2q}\dot{S}^7_{ij}]
  +4\text{\dh}^{k}[\Phi^{2q}\dot{S}^8_{ijk}] \\ 
 & +4\text{\dh}^{k}[\Phi^{2q}\dot{S}^9_{ijk}] 
  +\tfrac{8\rho'}{n}\Phi^{2q}\dot{S}^{10}_{ij}+\tfrac{8\rho}{n}\Phi^{2q}\dot{S}^{11}_{ij}-\tfrac{4\rho'}{n}\Phi^{2q}\dot{S}^{12}_{ij}
  -\tfrac{4\rho}{n}\Phi^{2q}\dot{S}^{13}_{ij}, \label{D*DCsbh1}
\end{align}
where the $S^{r}$ factors, $r=6,...,13$, are defined by:
\begin{align}
 S^6_{ij} =& 3 m^{\a}_i m^{\b}_j l^{\g}n^{\d}n^{\e}\c_{[\e}C_{\a\b]\g\d}, \\
 S^7_{ij} =& 3 m^{\a}_i m^{\b}_j l^{\g}n^{\d}\ell^{\e}\c_{[\e}C_{\a\b]\g\d}, \\
 S^8_{ijk} =& 3 \ell^{\a} n^{\b} m^{\g}_i m^{\d}_j m^{\e}_k \c_{[\e}C_{\a\b]\g\d}, \\
 S^9_{ijk} =& 3 m^{\a}_i m^{\b}_j l^{\g}n^{\d} m^{\e}_k\c_{[\e}C_{\a\b]\g\d}, \\
 S^{10}_{ij} =& 3 \ell^{\a} n^{\b} \ell^{\g} m^{\d}_{[j} m^{\e}_{i]} \c_{[\e}C_{\a\b]\g\d}, \\
 S^{11}_{ij} =& 3 \ell^{\a} n^{\b} n^{\g} m^{\d}_{[j} m^{\e}_{i]} \c_{[\e}C_{\a\b]\g\d}, \\
 S^{12}_{ij} =& 3 m^{\a}_i m^{\b}_j l^{\g} g^{\d\e} \c_{[\e}C_{\a\b]\g\d}, \\
 S^{13}_{ij} =& 3 m^{\a}_i m^{\b}_j n^{\d} g^{\g\e}\c_{[\e}C_{\a\b]\g\d}.
\end{align}
To evaluate the linearization of these terms, we need the GHP form of some components of $3\c_{[\e}C_{\a\b]\g\d}$; 
these are given in (\ref{b6})-(\ref{b13}). We find:
\begin{align}
 \dot{S}^6_{ij} =& 2(\text{\th}'+\tfrac{3\rho'}{n})\dot\Phi^A_{ij}+2\text{\dh}_{[i}\dot\Psi'_{j]}-\tfrac{2(n+1)}{n}\Phi\dot\rho'_{[ij]}, \\
 \dot{S}^7_{ij} =& 2(\text{\th}+\tfrac{3\rho}{n})\dot\Phi^A_{ij}-2\text{\dh}_{[i}\dot\Psi_{j]}+\tfrac{2(n+1)}{n}\Phi\dot\rho_{[ij]}, \\
\nonumber \dot{S}^8_{ijk} =& 2\text{\dh}_k\dot\Phi^A_{ij}+(\text{\th}+\tfrac{\rho}{n})\dot\Psi'_{kij}-(\text{\th}'+\tfrac{\rho'}{n})\dot\Psi_{kij} 
   -\tfrac{2}{n}\d_{k[i}(\rho\dot\Psi'_{j]}-\rho'\dot\Psi_{j]}) \\
   & -\tfrac{2}{n}\left(\tfrac{n+1}{n-1}\right)\Phi\d_{k[i}(\dot\t'_{j]}-\dot\t_{j]}), \\
 \dot{S}^9_{ijk} =& 6\text{\dh}_{[k}\dot\Phi^A_{ij]}, \\
 \dot{S}^{10}_{ij} =& -(\text{\th}+\tfrac{3\rho}{n})\dot\Phi^A_{ij}+\text{\dh}_{[i}\dot\Psi_{j]}-\tfrac{(n+1)}{n}\Phi\dot\rho_{[ij]}, \\
 \dot{S}^{11}_{ij} =& -(\text{\th}'+\tfrac{3\rho'}{n})\dot\Phi^A_{ij}-\text{\dh}_{[i}\dot\Psi'_{j]}+\tfrac{(n+1)}{n}\Phi\dot\rho'_{[ij]}, \\
 \dot{S}^{12}_{ij} =& 2(\text{\th}+(1+\tfrac{1}{n})\rho)\dot\Phi^A_{ij}+\text{\dh}^k\dot\Psi_{kij}
 +\tfrac{2}{n}\left(\tfrac{n+1}{n-1}\right)\Phi\dot\rho_{[ij]}, \\
 \dot{S}^{13}_{ij} =& -2(\text{\th}'+(1+\tfrac{1}{n})\rho')\dot\Phi^A_{ij}+\text{\dh}^k\dot\Psi'_{kij}
 +\tfrac{2}{n}\left(\tfrac{n+1}{n-1}\right)\Phi\dot\rho'_{[ij]}.
\end{align}
Replacing these expressions in (\ref{D*DCsbh1}), we get
\begin{multline}
 \tfrac{d}{d\ve}|_{\ve=0}\left[ 3 W^{\a\b\g\d}{}_{ij}(\c^{\e}+4\g^{\e})\c_{[\e}C_{\a\b]\g\d} \right] \equiv \\
 T^5_{ij}[\dot\Phi^A]+T^6_{ij}[\dot\Psi_i,\dot\Psi'_i]+T^7_{ij}[\dot\Psi_{ijk},\dot\Psi'_{ijk}]+T^8_{ij}[\dot\rho,\dot\rho',\dot\t,\dot\t'],
 \label{D*DCsbh2}
\end{multline}
where $T^5_{ij}[\dot\Phi^A]$ is simply the result of putting together all terms containing explicitly $\dot\Phi^A_{ij}$, 
and analogously for $T^{6}_{ij}, T^{7}_{ij}, T^{8}_{ij}$.
The evaluation of these terms is tedious but straightforward, we only need the background identities (\ref{idb-bsbh}), 
(\ref{idr-bsbh}), and the background commutator identity for $[\text{\th},\text{\dh}_i]$.
The result is:
\begin{align}
\nonumber T^5_{ij}[\dot\Phi^A]=& 8(\text{\th}+\rho)\text{\th}'[\Phi^{2q}\dot\Phi^A_{ij}]+8(\text{\th}'+\rho')\text{\th}[\Phi^{2q}\dot\Phi^A_{ij}]
   +16\text{\dh}^k\text{\dh}_k[\Phi^{2q}\dot\Phi^A_{ij}] \\
   & -16\text{\dh}^k\text{\dh}_{[j}[\Phi^{2q}\dot\Psi^A_{i]k}]-\tfrac{16}{n}(\Phi+\tfrac{2\L}{n+1})\Phi^{2q}\dot\Phi^A_{ij}, \\
 T^6_{ij}[\dot\Psi_i,\dot\Psi'_i]=& 8\Phi^{2q} \left[ \text{\dh}_{[i}(\text{\th}+\rho)\dot\Psi'_{j]} 
    -\text{\dh}_{[i}(\text{\th}'+\rho')\dot\Psi_{j]} \right], \\
 T^7_{ij}[\dot\Psi_{ijk},\dot\Psi'_{ijk}]=& 4\Phi^{2q} \left[(\text{\th}+\tfrac{3\rho}{n})\text{\dh}^k\dot\Psi'_{kij} 
    -(\text{\th}'+\tfrac{3\rho'}{n})\text{\dh}^k\dot\Psi_{kij}  \right], \\
\nonumber T^8_{ij}[\dot\rho,\dot\rho',\dot\t,\dot\t']=& \tfrac{8(n+1)}{n}\Phi^{2q+1}\left[-\text{\th}\dot\rho'_{[ij]}+\text{\th}'\dot\rho_{[ij]}
    -\tfrac{1}{n}\left(\tfrac{n-2}{n-1}\right)(\rho\dot\rho'_{[ij]}-\rho'\dot\rho_{[ij]}) \right. \\
   & \left. -\tfrac{1}{(n-1)}(\text{\dh}_{[i}\dot\t'_{j]}-\text{\dh}_{[i}\dot\t_{j]}) \right].
\end{align}

\subsubsection*{Final step}

With all the expressions for $T^{r}_{ij}$, $r=1,...,8$, we add now (\ref{D*DCsbh2}) to (\ref{DD*Csbh1}).
For the evaluation of the terms with perturbed spin coefficients, we need the following perturbed Ricci identity
\begin{equation}
 \text{\th}'\dot\rho_{[ij]}+\text{\dh}_{[i}\dot\t_{j]}=-\tfrac{\rho'}{n}\dot\rho_{[ij]}-\tfrac{\rho}{n}\dot\rho'_{[ij]}-\dot\Phi^A_{ij},
\end{equation}
together with its primed version. We obtain:
\begin{align}
 & T^2_{ij}+T^6_{ij}=0,\\
 & T^3_{ij}+T^7_{ij}=0,\\
 & T^4_{ij}+T^8_{ij}=-16\left(\tfrac{n+1}{n-1}\right)\Phi^{2q+1}\dot\Phi^A_{ij}.
\end{align}
Therefore we are only left with terms containing $\dot\Phi^A_{ij}$.
For the terms with $\text{\th}$ derivatives in $T^1_{ij}$, we can commute the $\Phi^{2q}$ factor to the left using (\ref{idb-bsbh}) 
with $m=2$; thus
\begin{align}
\nonumber T^1_{ij}=& 8(\text{\th}+\rho)\text{\th}'[\Phi^{2q}\dot\Phi^A_{ij}]+8(\text{\th}'+\rho')\text{\th}[\Phi^{2q}\dot\Phi^A_{ij}]
 + 16\text{\dh}_{[j}\text{\dh}^k[\Phi^{2q}\dot\Phi^A_{i]k}] \\
 &-\tfrac{16(n-1)}{n}\left(\Phi+\tfrac{2\L}{n+1}\right)\Phi^{2q}\dot\Phi^A_{ij}.
\end{align}
We also need the commutator of $\text{\dh}$ derivatives; this is
\begin{equation}
 [\text{\dh}_{[j},\text{\dh}^k]\Phi^{2q}\dot\Phi^A_{i]k}=
   -\tfrac{2(n-2)}{n}\left(\tfrac{\rho\rho'}{n}+\tfrac{\L}{n+1}-\tfrac{\Phi}{n-1}\right)\Phi^{2q}\dot\Phi^A_{ij}.
\end{equation}
Then we finally get
\begin{equation}
 -\tfrac{d}{d\ve}|_{\ve=0}\{W^{\a\b\g\d}{}_{ij}[(\Ds^{\star}_2\Ds+\Ds_2\Ds^{\star})C]_{\a\b\g\d}\}
 =16 (\Box_{(0,2)}+V)[\Phi^{2q}\dot\Phi^{A}_{ij}],
\end{equation}
where we used the expression (\ref{MBox1}) for the generalized wave operator $\Box_{(b,s)}$, 
and the potential was defined in (\ref{Vgsbh}).


\begin{thebibliography}{99}



\bibitem{Andersson1} 
  S.~Aksteiner and L.~Andersson,
  {\it Linearized gravity and gauge conditions,}
  Class.\ Quant.\ Grav.\  {\bf 28}, 065001 (2011)
  [arXiv:1009.5647 [gr-qc]].

\bibitem{Aksteiner1} 
  S.~Aksteiner, L.~Andersson and T.~Bäckdahl,
  {\it New identities for linearized gravity on the Kerr spacetime},
  arXiv:1601.06084 [gr-qc].

\bibitem{Aksteiner2} 
  S.~Aksteiner and T.~Bäckdahl,
  {\it Symmetries of linearized gravity from adjoint operators},
  arXiv:1609.04584 [gr-qc].

\bibitem{Araneda:2016iwr} 
  B.~Araneda,
  {\it Symmetry operators and decoupled equations for linear fields on black hole spacetimes},
  Class.\ Quant.\ Grav.\  {\bf 34}, no. 3, 035002 (2017)
  doi:10.1088/1361-6382/aa51ff
  [arXiv:1610.00736 [gr-qc]].  

\bibitem{Araneda2}
 B.~Araneda, 
 work in progress.
 
 
\bibitem{Bini} 
  D.~Bini, C.~Cherubini, R.~T.~Jantzen and R.~J.~Ruffini,
  {\it Teukolsky master equation: De Rham wave equation for the gravitational and electromagnetic fields in vacuum},
  Prog.\ Theor.\ Phys.\  {\bf 107}, 967 (2002)
  doi:10.1143/PTP.107.967
  [gr-qc/0203069].

\bibitem{Bini2}
 D.~Bini, C.~Cherubini, R.~T.~Jantzen and R.~J.~Ruffini,
 {\it De Rham wave equation for tensor valued p-forms},
 International Journal of Modern Physics D, 12(08), 1363-1384 (2003).

\bibitem{Birmingham} 
  D.~Birmingham,
  {\it Topological black holes in Anti-de Sitter space},
  Class.\ Quant.\ Grav.\  {\bf 16}, 1197 (1999)
  doi:10.1088/0264-9381/16/4/009
  [hep-th/9808032].

\bibitem{Coley:2004jv} 
  A.~Coley, R.~Milson, V.~Pravda and A.~Pravdova,
  {\it Classification of the Weyl tensor in higher dimensions},
  Class.\ Quant.\ Grav.\  {\bf 21}, L35 (2004)
  doi:10.1088/0264-9381/21/7/L01
  [gr-qc/0401008].

\bibitem{Dias:2009ex} 
  O.~J.~C.~Dias, H.~S.~Reall and J.~E.~Santos,
  {\it Kerr-CFT and gravitational perturbations},
  JHEP {\bf 0908}, 101 (2009)
  doi:10.1088/1126-6708/2009/08/101
  [arXiv:0906.2380 [hep-th]].

\bibitem{Dotti:2013uxa} 
  G.~Dotti,
  {\it Nonmodal linear stability of the Schwarzschild black hole},
  Phys.\ Rev.\ Lett.\  {\bf 112}, 191101 (2014)
  doi:10.1103/PhysRevLett.112.191101
  [arXiv:1307.3340 [gr-qc]].


\bibitem{Dotti:2016cqy} 
  G.~Dotti,
  {\it Black hole non-modal linear stability: the Schwarzschild (A)dS cases},
  Class.\ Quant.\ Grav.\  {\bf 33}, no. 20, 205005 (2016)
  doi:10.1088/0264-9381/33/20/205005
  [arXiv:1603.03749 [gr-qc]].

  
\bibitem{Durkee:2010xq} 
  M.~Durkee, V.~Pravda, A.~Pravdova and H.~S.~Reall,
  {\it Generalization of the Geroch-Held-Penrose formalism to higher dimensions},
  Class.\ Quant.\ Grav.\  {\bf 27}, 215010 (2010)
  doi:10.1088/0264-9381/27/21/215010
  [arXiv:1002.4826 [gr-qc]].


\bibitem{Durkee:2010qu} 
  M.~Durkee and H.~S.~Reall,
  {\it Perturbations of higher-dimensional spacetimes},
  Class.\ Quant.\ Grav.\  {\bf 28}, 035011 (2011)
  doi:10.1088/0264-9381/28/3/035011
  [arXiv:1009.0015 [gr-qc]].

\bibitem{Durkee:2010ea} 
  M.~Durkee and H.~S.~Reall,
  {\it Perturbations of near-horizon geometries and instabilities of Myers-Perry black holes},
  Phys.\ Rev.\ D {\bf 83}, 104044 (2011)
  doi:10.1103/PhysRevD.83.104044
  [arXiv:1012.4805 [hep-th]].  

\bibitem{ehlers}
 J.~Ehlers,  
 {\it The geometry of the (modified) GHP-formalism},
 Communications in Mathematical Physics, 37(4), 327-329.  (1974).

\bibitem{Fackerell} 
  E.~D.~Fackerell and J.~R.~Ipser,
  {\it Weak electromagnetic fields around a rotating black hole,}
  Phys.\ Rev.\ D {\bf 5}, 2455 (1972).

\bibitem{Frolov:2017kze} 
  V.~Frolov, P.~Krtous and D.~Kubiznak,
  {\it Black holes, hidden symmetries, and complete integrability},
  arXiv:1705.05482 [gr-qc].  


\bibitem{Guica} 
  M.~Guica, T.~Hartman, W.~Song and A.~Strominger,
  {\it The Kerr/CFT Correspondence},
  Phys.\ Rev.\ D {\bf 80}, 124008 (2009)
  doi:10.1103/PhysRevD.80.124008
  [arXiv:0809.4266 [hep-th]].


\bibitem{Godazgar:2011sn} 
  M.~Godazgar,
  {\it The perturbation theory of higher dimensional spacetimes a la Teukolsky},
  Class.\ Quant.\ Grav.\  {\bf 29}, 055008 (2012)
  doi:10.1088/0264-9381/29/5/055008
  [arXiv:1110.5779 [gr-qc]].

\bibitem{Held}
 A.~Held,
 {\it A formalism for the investigation of algebraically special metrics}
 Commun.\ Math.\ Phys. {\bf 37}, no. 4, 311--326 (1974)

\bibitem{Hollands:2014lra} 
  S.~Hollands and A.~Ishibashi,
  {\it Instabilities of extremal rotating black holes in higher dimensions},
  Commun.\ Math.\ Phys.\  {\bf 339}, no. 3, 949 (2015)
  doi:10.1007/s00220-015-2410-0
  [arXiv:1408.0801 [hep-th]].

\bibitem{Ishibashi:2003ap} 
  A.~Ishibashi and H.~Kodama,
  {\it Stability of higher dimensional Schwarzschild black holes},
  Prog.\ Theor.\ Phys.\  {\bf 110}, 901 (2003)
  doi:10.1143/PTP.110.901
  [hep-th/0305185].

\bibitem{Ishibashi:2011ws} 
  A.~Ishibashi and H.~Kodama,
  {\it Perturbations and Stability of Static Black Holes in Higher Dimensions},
  Prog.\ Theor.\ Phys.\ Suppl.\  {\bf 189}, 165 (2011)
  doi:10.1143/PTPS.189.165
  [arXiv:1103.6148 [hep-th]].

\bibitem{Kodama} 
  H.~Kodama and A.~Ishibashi,
  {\it A Master equation for gravitational perturbations of maximally symmetric black holes in higher dimensions},
  Prog.\ Theor.\ Phys.\  {\bf 110}, 701 (2003)
  doi:10.1143/PTP.110.701
  [hep-th/0305147].

\bibitem{Kegeles} 
  L.~S.~Kegeles and J.~M.~Cohen,
  {\it Constructive Procedure For Perturbations Of Space-times},
  Phys.\ Rev.\ D {\bf 19}, 1641 (1979).
  doi:10.1103/PhysRevD.19.1641

\bibitem{Nakahara} 
  M.~Nakahara,
  ``Geometry, topology and physics,''
  Boca Raton, USA: Taylor $\&$ Francis (2003) 573 p

\bibitem{Penrose2}
  R.~Penrose and W.~Rindler,
  ``Spinors And Space-time. Vol. 2: Spinor And Twistor Methods In Space-time Geometry,''
  Cambridge, Uk: Univ. Pr. ( 1986) 501p

\bibitem{Penrose3}
 R.~Penrose,
  {\it Zero rest mass fields including gravitation: Asymptotic behavior,}
  Proc.\ Roy.\ Soc.\ Lond.\ A {\bf 284} (1965) 159.
   
\bibitem{Pravda:2005kp} 
  V.~Pravda and A.~Pravdova,
  {\it WANDs of the black ring},
  Gen.\ Rel.\ Grav.\  {\bf 37}, 1277 (2005)
  doi:10.1007/s10714-005-0110-3
  [gr-qc/0501003].

\bibitem{Tanahashi:2012si} 
  N.~Tanahashi and K.~Murata,
  {\it Instability in near-horizon geometries of even-dimensional Myers-Perry black holes},
  Class.\ Quant.\ Grav.\  {\bf 29}, 235002 (2012)
  doi:10.1088/0264-9381/29/23/235002
  [arXiv:1208.0981 [hep-th]].

\bibitem{Teukolsky} 
  S.~A.~Teukolsky,
  {\it Rotating black holes - separable wave equations for gravitational and electromagnetic perturbations},
  Phys.\ Rev.\ Lett.\  {\bf 29}, 1114 (1972).
  doi:10.1103/PhysRevLett.29.1114

\bibitem{Wald}
  R.~M.~Wald,
  \textit{Construction of Solutions of Gravitational, Electromagnetic, Or Other Perturbation Equations from Solutions of 
  Decoupled Equations,}
  Phys.\ Rev.\ Lett.\  {\bf 41}, 203 (1978).

\bibitem{Penrose4} 
  M.~Walker and R.~Penrose,
  {\it On quadratic first integrals of the geodesic equations for type [22] spacetimes,}
  Commun.\ Math.\ Phys.\  {\bf 18}, 265 (1970). 


\end{thebibliography}
\end{document}